\documentclass[letterpaper,final,twocolumn]{IEEEtran}
\usepackage[english]{babel}
\usepackage{amsmath,amssymb,amsfonts,mathptmx}
\usepackage{amsthm}
\usepackage{xcolor,graphicx,subfigure}
\usepackage{epsfig}

\usepackage{enumitem}

\catcode`\<=\active \def<{
        \fontencoding{T1}\selectfont\symbol{60}\fontencoding{\encodingdefault}}
\newcommand{\nin}{\not\in}

\newcommand{\tmtextit}[1]{{\itshape{#1}}}

\newtheorem{theorem}{Theorem}[section]

\newtheorem{lemma}[theorem]{Lemma}

\newtheorem{remark}[theorem]{Remark}

\newtheorem{proposition}[theorem]{Proposition}

\newcommand{\oprocendsymbol}{\hbox{$\bullet$}}
\newcommand{\oprocend}{\relax\ifmmode\else\unskip\hfill\fi\oprocendsymbol}

\newcommand{\real}{{\mathbb{R}}}
\newcommand{\naturals}{{\mathbb{N}}}

\newcommand{\GG}{\mathcal{G}}
\newcommand{\II}[1]{\mathcal{I}^{#1}}
\newcommand{\RR}{\mathbf{R}}
\newcommand{\zeros}{\bold{0}}
\newcommand{\ones}{\bold{1}}

\newcommand{\setdef}[2]{\left\{ #1 \; \big| \; #2\right\}}

\newcommand{\longthmtitle}[1]{\tmtextit{(#1).}}
\newcommand{\myclearpage}{\clearpage}
\renewcommand{\myclearpage}{}
\newcommand{\change}[1]{{\color{blue} #1}}

\newcommand{\range}{\operatorname{range}}
\DeclareMathOperator{\sgn}{sgn}
\DeclareMathOperator{\sat}{sat}
\usepackage[normalem]{ulem}

\allowdisplaybreaks

\myclearpage
\begin{document}

\title{\LARGE{Distributed bilayered control for transient frequency safety
  and system stability in power grids}\thanks{A preliminary version
    appeared as~\cite{YZ-JC:19-acc} at the 2019 American Control
    Conference. This work was supported by NSF Award CNS-1446891 and
    AFOSR Award FA9550-15-1-0108.}}

% \title{Distributed transient frequency bilayered control of power
%   networks}

\author{Yifu~Zhang \quad Jorge
  Cort{\'e}s%,~\IEEEmembership{Fellow,~IEEE}
  \IEEEcompsocitemizethanks{\IEEEcompsocthanksitem Yifu Zhang is with
    The MathWorks, Inc., Natick, MA 01760, USA ({\tt\small
      yifu.zhang19@gmail.com}).  During the preparation of this work,
    Yifu Zhang was affiliated with the Department of Mechanical and
    Aerospace Engineering, University of California, San Diego.
  
  Jorge Cort{\'e}s is with the Department of Mechanical and
    Aerospace Engineering, University of California, San Diego, La Jolla,
    CA 92093, USA ({\tt\small cortes@ucsd.edu}).}}
\maketitle

\begin{abstract}
  This paper considers power networks governed by swing nonlinear
  dynamics and subject to disturbances. We develop a bilayered
  control strategy for a subset of buses that simultaneously
  guarantees transient frequency safety of each individual bus and
  asymptotic stability of the entire network. The bottom layer is a
  model predictive controller that, based on periodically sampled
  system information, optimizes control resources to have transient
  frequency evolve close to a safe desired interval. The top layer is
  a real-time controller assisting the bottom-layer controller to
  guarantee transient frequency safety is actually achieved. We show
  that control signals at both layers are Lipschitz in the state and
  do not jeopardize stability of the network.  Furthermore, we
  carefully characterize the information requirements at each bus
  necessary to implement the controller and employ saddle-point
  dynamics to introduce a distributed implementation that only
  requires information exchange with up to 2-hop neighbors in the
  power network. Simulations on the IEEE 39-bus power network
  illustrate our results. 
\end{abstract}

% \begin{keywords}
%   Transient frequency, power network stability, distributed control,
%   model predictive control.
% \end{keywords}

\section{Introduction}\label{section:intro}
 The electric power system is operated around a nominal frequency to
maintain its stability and safety. Large frequency fluctuations can
trigger generator relay-protection mechanisms and load
shedding~\cite{PK-JP:04,NERC:11}, which may further jeopardize network
integrity, leading to cascading failures. Without appropriate
operational architectures and control safeguards in place, the
likelihood of such events is not negligible, given that the high
penetration of non-rotational renewable resources provides less
inertia, possibly inducing higher frequency
excursions~\cite{FM-FD-GH-DJH-GV:2018}. These observations motivate us
to develop control schemes to actively mitigate undesired transient
frequency deviations under disturbances and
contingencies. Specifically, we are interested in exploiting the
potential benefits of distributed controllers and architectures to
enable plug-and-play capabilities, the efficient orchestration among
the roles of the available resources, and handling the coordination of
large numbers of them in an adaptive and scalable fashion.

% based on the prediction of future disturbance and in a
%   receding-horizon optimization fashion, the proposed controller
%   enables global cooperation among control resources regulating on
%   individual buses spanned in power networks. Meanwhile, to address
%   the implementability and scalability of controller in large-scale
%   power networks, we show that the controller is distributed, i.e.,
%   control action on each bus only depends on local system information
%   and communication.
 
\textit{Literature review:} Power system stability is defined as the
ability of regaining operating equilibrium conditions in the presence
of disturbances while keeping deviations of system states within
acceptable levels~\cite{PK-JP:04}. A branch of
research~\cite{FD-MC-FB:13,HDC:11,TLV-HDN-AM-JS-KT:18} focuses on
characterizing equilibrium and convergence as a function of network
topology, initial conditions, and system parameters, without
explicitly accounting for the potential disruptions in power system
stability caused by mechanisms that are activated by frequency
excursions beyond safe limits.  Various control strategies have been
proposed to improve transient frequency behavior against disturbances,
including inertial placement~\cite{JF-HL-YT-FB:18}, droop coefficient
design~\cite{SSG-CZ-ED-YCC-SVD:18}, and demand-side frequency
regulations~\cite{FT-MA-DP-GS:15}. However, these methods rely on some
a-priori explicit frequency overshoot estimation based on
reduced-order models, and hence only provide approximate transient
frequency safety guarantees.  Combining the notion of control
barrier~\cite{ADA-SC-ME-GM-KS-PT:19} and Lyapunov~\cite{HKK:02}
functions, our previous work~\cite{YZ-JC:19-auto} proposes a feedback
controller that meets both requirements of transient frequency safety
and asymptotic stability. This controller is distributed and requires
no communication, in the sense that each control signal regulated on
an individual bus only depends on neighboring system information that
can be directly measured.  However, its non-optimization-based nature
may cause bounded oscillations in the closed-loop system due to the
lack of cooperation among control signals.  Our
work~\cite{YZ-JC:20-auto} employs a model predictive control
(MPC)-based approach to address this issue, but the prediction horizon
that can be used is limited by trade-offs between the discretization
accuracy and the computational complexity, limiting its performance.
In addition, the implementation of the MPC-based controller is only
partially distributed: given a set of regions in the network, a
centralized controller aggregates information and determines the
control actions within each region, independently of the
others. Challenges in employing MPC techniques in the context of power
networks~\cite{HJ-JL-YS-DJH:15,ANV-IAH-JBR-SJW:08,AF-MI-TD-MM:14}
include the fact that, as the equilibrium point heavily depends on
modeling and network parameters that cannot be precisely known, it is
analytically hard to establish robust stabilization given that the
objective function generally requires knowledge of the equilibrium
point; the widespread use in practice of MPC with linearized models
for prediction given the nonlinear nature of the dynamics of power
networks; and the processing power and information transmission, speed
and reliability requirements associated with a single operator for
measured state collection, online optimization, and decision making
given the large number of actors and volume of data.

\textit{Statement of contribution:} This paper proposes a distributed
controller framework implemented on buses available for control that
maintains network asymptotic stability and enforces transient
frequency safety under disturbances. If a bus frequency is initially
in a prescribed safe frequency interval, then it can only evolve
within the interval afterwards; otherwise, the controller leads
frequency to enter the safe interval within a finite time. The
proposed controller possesses a bilayer structure. The bottom layer
solves periodically a finite-horizon convex optimization problem and
globally allocates control resources to minimize the overall control
effort.  The optimization problem incorporates a prediction model for
the system dynamics, a stability constraint, and a relaxed frequency
safety constraint. The prediction model is a linearized and
discretized approximation of the nonlinear continuous-time power
network dynamics, carefully chosen to preserve its local nature while
keeping the complexity manageable. As a consequence, in the resulting
convex optimization problem, the objective function can be interpreted
as the sum of local control costs, and each constraint only involves
local decision variables. This enables us to apply saddle-point
dynamics to recover its solution in a distributed fashion by allowing
each bus (resp. line) to exchange system information within its
neighboring buses (resp. lines). On the other hand, the top layer, as
a real-time feedback controller, acts as a compensator, bridging the
mismatch between the actual continuous-time power network dynamics and
the sampled-based information employed in the bottom layer to
rigorously guarantee frequency safety.  The top layer control signal
regulating on a generic bus only depends on physical measurements of
system information within the range of its neighboring transmission
lines. We illustrate the performance of the proposed bilayered
controller architecture in the IEEE 39-bus power network.

\section{Preliminaries}\label{section:pre}
Here we gather notation and concepts used in the paper.

\subsubsection*{Notation}
Let $\naturals$, $\real$, and $\real_{\geqslant}$, $\real_{>}$ denote
the set of natural, real, nonnegative real, and strictly positive
numbers, respectively.  Variables belong to the Euclidean space unless
specified otherwise. Denote by $\lceil a \rceil$ as the ceiling of
$a\in\real$. For $A\in\mathbb{R}^{m\times n}$, let $[A]_i$ and
$[A]_{i,j}$ be its $i$th row and $(i,j)$th element, respectively.  We
denote by $A^{\dagger}$ its unique Moore-Penrose pseudoinverse and by
$\range(A)$ its column space. For $b\in\real^{n}$, $b_{i}$ denotes its
$i$th entry. Let $\ones_n$ and $\zeros_n$ in $\real^n$ denote the
vector of all ones and zeros, respectively.  $\|\cdot\|$ denotes the
2-norm on $\real^{n}$.  For any $c,d\in\naturals$, let
$[c,d]_{\naturals}= \left\{ x\in\naturals \big| c\leqslant x\leqslant
  d \right\}$.  Denote the sign function
$\sgn:\real\rightarrow\{-1,1\}$ as $ \sgn(a)= 1$ if $a\geqslant 0$, and
as $ \sgn(a)= -1$ if $a< 0$.  Define the saturation function
$\sat:\real \rightarrow \real$ with limits $a^{\min}<a^{\max}$ as
\begin{align*}
  \sat(a;a^{\max},a^{\min}) = 
  \begin{cases}
    a^{\max} & a\geqslant a^{\max},
    \\
    a^{\min} & a\leqslant a^{\min},
    \\
    a & \text{otherwise}.
  \end{cases}
\end{align*}
Given $\mathcal{C} \subset \real^{n}$, $\partial\mathcal{C}$ denotes
its boundary and $\mathcal{C}_{\text{cl}}$ denotes its closure.  For a
point $x\in\real^{n}$ and $r\in\real_{>}$, denote
$B_{r}(x)\triangleq\setdef{x'\in\real^{n}}{\|x'-x\|_{2}\leqslant r}$.
Given a differentiable function $l:\real^{n}\rightarrow\real$, we let
$\nabla l$ denote its gradient.  A function $f:\real_{\geqslant
}\times\real^{n}\rightarrow\real^{n},\ (t,x)\rightarrow f(t,x)$ is
Lipschitz in $x$ (uniformly in $t$) if for every $x_{0}\in\real^{n}$,
there exist $L,r>0$ such that $\|f(t,x)-f(t,y)\|_{2}\leqslant
L\|x-y\|_{2}$ for any $x,y\in B_{r}(x_{0})$ and any $t\geqslant
0$. Given a function
$\mathfrak{L}:\mathcal{Y}\times\mathcal{Z}\rightarrow\real$, a point
$(Y^{*},Z^{*})\in\mathcal{Y}\times\mathcal{Z}$ is a saddle point of
$\mathfrak{L}$ on the set $\mathcal{Y}\times\mathcal{Z}$ if
$\mathfrak{L}(Y^{*},Z) \leqslant
\mathfrak{L}(Y^{*},Z^{*})\leqslant\mathfrak{L}(Y,Z^{*})$ holds for
every $(Y,Z)\in\mathcal{Y}\times\mathcal{Z}$. For scalars
$a,b\in\real$, let $[a]_{b}^{+}=a$ if $b>0$, and
$[a]_{b}^{+}=\max\{a,0\}$ if $b\leqslant0$. For vectors
$a,b\in\real^{n}$, $[a]_{b}^{+}\in\real^{n}$ is the vector whose $i$th
component is $[a_{i}]^{+}_{b_{i}}$ for every $i\in[1,n]_{\naturals}$.
  
\subsubsection*{Graph theory}
We introduce algebraic graph theory basics
from~\cite{FB-JC-SM:08cor}. An undirected graph is a pair $\mathcal{G}
= \mathcal(\mathcal{I},\mathcal{E})$, where $\mathcal{I}$ is the
vertex set and $\mathcal{E} \subseteq \mathcal{I} \times \mathcal{I}$
is the edge set.
% An induced subgraph $\mathcal{G}_{\sigma} =
% (\mathcal{I}_{\sigma},\mathcal{E}_{\sigma})$ of $\mathcal{G}$ is a
% graph satisfying $\mathcal{I}_{\sigma}\subseteq\mathcal{I}$,
% $\mathcal{E}_{\sigma}\subseteq\mathcal{E}$, and
A graph is connected if there exists a path between any two vertices.
We denote by $\mathcal{N}(i)$ the set of neighbors of node~$i$.  An
orientation procedure is to, for each generic edge $e_{k} \in
\mathcal{E}$ with vertices $i,j$, choose either $i$ or $j$ as the
positive end and the other as the negative end. For a given
orientation, the incidence matrix $D=(d_{k i}) \in \mathbb{R}^{m
  \times n}$ associated with $\mathcal{G}$ is defined as
$ d_{k i} = 1$ if $i$ is the positive end of $e_{k}$, $d_{k i} = -1$
if $i$ is the negative end of $e_{k}$, and $d_{k i} = 0$
otherwise.
% \begin{align*}
%   d_{k i} =
%   \begin{cases}
%     1 & \text{if $i$ is the positive end of $e_{k}$},
%     \\
%     - 1 & \text{if $i$ is the negative end of $e_{k}$},
%     \\
%     0 & \text{otherwise}.
%   \end{cases}
% \end{align*}

\section{Problem statement}\label{section:ps}

We introduce here model for the power network and state the desired
performance goals on the controller design.  

We use a connected undirected graph
$\mathcal{G}=(\mathcal{I},\mathcal{E})$ to represent the power
network, where $\mathcal{I}=\{1,2,\cdots,n\}$ stands for the set of
buses (nodes) and $\mathcal{E} =
\{e_{1},e_{2},\cdots,e_{m}\}\subseteq\mathcal{I}\times\mathcal{I}$
represents the collection of transmission lines (edges).  At each bus
$i\in\mathcal{I}$, denote by $\omega_{i}\in\real$,
$\theta_{i}\in\real$ $p_{i}\in\real$, $M_{i}\in\real_{\geqslant}$, and
$E_{i}\in\real_{> }$ the shifted frequency with respect to the nominal
frequency, voltage angle, active power injection, inertial, and
damping (droop) coefficient, respectively. \change{Notice that we
  explicitly allow buses to have zero inertia. We assume at least one bus
  possesses strictly positive inertia.}  Given an arbitrary
orientation procedure of $\mathcal{G}$, let $D\in\real^{m\times n}$ be
the corresponding incidence matrix. In addition, for each generic
transmission line with positive end $i$ and negative end $j$, denote
$\lambda_{ij}\triangleq \theta_{i}-\theta_{j}$ as the voltage angle
difference between node $i$ and $j$; denote $b_{ij}\in\real_{>}$ as
the line susceptance.
%, and $f_{ij}\triangleq b_{ij}\lambda_{ij}$ as power flow.
Let $\II{u}\subset\mathcal{I}$ be the collection of bus indexes with
additional control inputs. Let $\theta\in\real^{n}$,
$\omega\in\real^{n}$, $\lambda\in\real^{m}$, denote the collection of
$\theta_{i}$'s, $\omega_{i}$'s, and $\lambda_{ij}$'s, respectively.
Let $Y_{b}\in\real^{m\times m}$ be the diagonal matrix whose $k$th
diagonal entry is the susceptance of the transmission line $e_{k}$
connecting $i$ and $j$, i.e., $[Y_{b}]_{k,k}=b_{ij}$. By definition,
%\begin{subequations}\label{sube:theta-lambda-f}
\begin{align}
  \lambda&=D\theta . \label{eqn:lambda-theta}
  % \\
  % f&=Y_{b}\lambda.
\end{align}
% \end{subequations}
Let $M \triangleq
\text{diag}(M_{1},M_{2},\cdots,M_{n})\in\real^{n\times n}$, and $E
\triangleq \text{diag}(E_{1},E_{2},\cdots,E_{n})\in\real^{n\times
  n}$. The nonlinear swing dynamics of power network can be
equivalently formulated by choosing either $(\theta,\omega)$ or
$(\lambda, \omega)$ to describe the system state. Here, we use the
latter one. \change{In this case, the dynamics can be described by the
  following differential algebraic equation}~\cite{ARB-DJH:81,AP:12},
\change{
\begin{subequations}\label{eqn:compact-form}
  \begin{align}
    \dot \lambda(t) &= D\omega(t),
    \\
    M\dot\omega(t) &=
    -E\omega(t)-D^{T}Y_{b}\sin(\lambda(t))+p+\alpha(t),\label{eqn:compact-form-2}
  \end{align}
\end{subequations}
}
where $ \alpha(t)\in\mathbb{A}\triangleq\left\{ z\in\real^{n} \big| \
  z_{w}=0 \text{ for }w\in \mathcal{I} \setminus \II{u}\right\}$ is
the control signal to be designed.  Furthermore, due to the
transformation~\eqref{eqn:lambda-theta}, one has
\begin{align}\label{eqn:initial-state}
  \lambda(0)&\in\range{(D)}.
\end{align}
Throughout the rest of the paper, if not specified, we assume that the
initial condition of system~\eqref{eqn:compact-form}
satisfies~\eqref{eqn:initial-state}.  We assume that the power
injection $p$ designed by the tertiary layer is balanced, i.e.,
$\ones_{n}^{T}p=0$. This assumption is reasonable, given that our
focus here is on the system transient frequency behavior, which
instead lies within the scope of primary and secondary control.
According to~\cite[Lemma 2]{FD-MC-FB:13}, the
system~\eqref{eqn:compact-form} with $\alpha\equiv 0_{n}$ has an
equilibrium $(\lambda^{\infty},\zeros_{n})\in\real^{m+n}$ that is
locally asymptotically stable if
\begin{align}\label{ineq:sufficient-eq}
  \|L^{\dagger}p\|_{\mathcal{E},\infty}<1,
\end{align}
where $L\triangleq D^{T}Y_{b}D$ and $\| z \|_{\mathcal{E},\infty}
\triangleq \max_{(i,j)\in\mathcal{E}} |z_{i}-z_{j}|$ for
$z\in\real^{n}$. In addition, $\lambda^{\infty}$ lies in $\Upsilon$
and is unique in the closure of $\Upsilon$, where $\Upsilon \triangleq
\setdef{\lambda}{|\lambda_{i}|<\pi/2,\ \forall
  i\in[1,m]_{\naturals}}$.

\begin{remark}\longthmtitle{Distributed 
    dynamics}\label{rmk:distributed-dynamics}
  {\rm We emphasize that the dynamics~\eqref{eqn:compact-form} is
    naturally distributed, i.e., the evolution of any given state is
    fully determined by the state information from its
    neighbors. Specifically, for each $(i,j)\in\mathcal{E}$,
    $\dot\lambda_{ij}$ is determined by $\omega_{i}$ and $\omega_{j}$,
    i.e., the states of neighbors  of edge $(i,j)$; for each
    $i\in\mathcal{F}$; $\dot\omega_{i}$ is determined by $M_{i}$, $\omega_{i}$,
    $E_{i}$, $p_{i}$, $\alpha_{i}$ and $\lambda_{ij}$, $b_{ij}$ with
    $(i,j)\in\mathcal{E}$ that are either state, parameter, and power
    injections belonging to node $i$, or states and parameters of its
    neighboring edges.  } \oprocend
\end{remark}

Given a target subset $\II{\omega}$ of $\II{u}$, our goal is to design
a distributed state-feedback controller, one per bus in $\II{u}$, that
maintains stability of the whole power network while cooperatively
guaranteeing frequency invariance and attractivity of nodes in
$\II{\omega}$. Formally, the controller $\alpha$ should make the
closed-loop system satisfy the following requirements:
\begin{enumerate}[wide]
\item \emph{Frequency safety:} For each $i\in\II{\omega}$, let
  $\underline\omega_{i}\in\real$ and $\bar\omega_{i}\in\real$ be lower
  and upper safe frequency bounds, with
  $\underline\omega_{i}<\bar\omega_{i}$.  If $\omega_{i}$ is initially
  safe, i.e., $\omega_{i}(0)\in[\underline\omega_{i},\bar\omega_{i}]$,
  then we require that the entire trajectory stay within
  $[\underline\omega_{i},\bar\omega_{i}]$. On the other hand, if
  $\omega_{i}$ is initially unsafe, then we require that there exists
  a finite time $t_{0}$ such that
  $\omega_{i}(t)\in[\underline\omega_{i},\bar\omega_{i}]$ for every
  $t\geqslant t_{0}$. This requirement is equivalent to asking the set
  $[\underline\omega_{i},\bar\omega_{i}]$ to be both invariant and
  attractive for each $i\in\II{\omega}$.
  
\item \emph{Local asymptotic stability:} The closed-loop system should
  preserve the asymptotic stability properties of the open-loop
  system~\eqref{eqn:compact-form} where $\alpha\equiv\zeros_{n}$. 

\item \emph{Lipschitz continuity:} The controller should be a
  Lipschitz function in the state argument. This ensures the existence
  and uniqueness of solution for the closed-loop system and rules out
  discontinuities in the control signal.

\item \emph{Economic cooperation:} Each bus in $\II{u}$ should
  cooperate with the others to reduce the overall control cost.

\item \emph{Distributed nature:} The controller $\alpha$ should be
  implementable in distributed way, i.e., node $i$ should be able to
  compute $\alpha_{i}$ by only exchanging information with its
  neighboring nodes and edges.
\end{enumerate}

% Our controller design to meet these requirements has a two-layer
% structure.
In Section~\ref{section:centralized-control}, we introduce a
centralized controller architecture that meets the requirements
(i)-(iv). We later build on this architecture in
Section~\ref{section:distributed-control} to provide a distributed
controller that satisfies all requirements~(i)-(v).

\begin{figure*}[tbh]
  \centering
  \includegraphics[width=.8\linewidth]{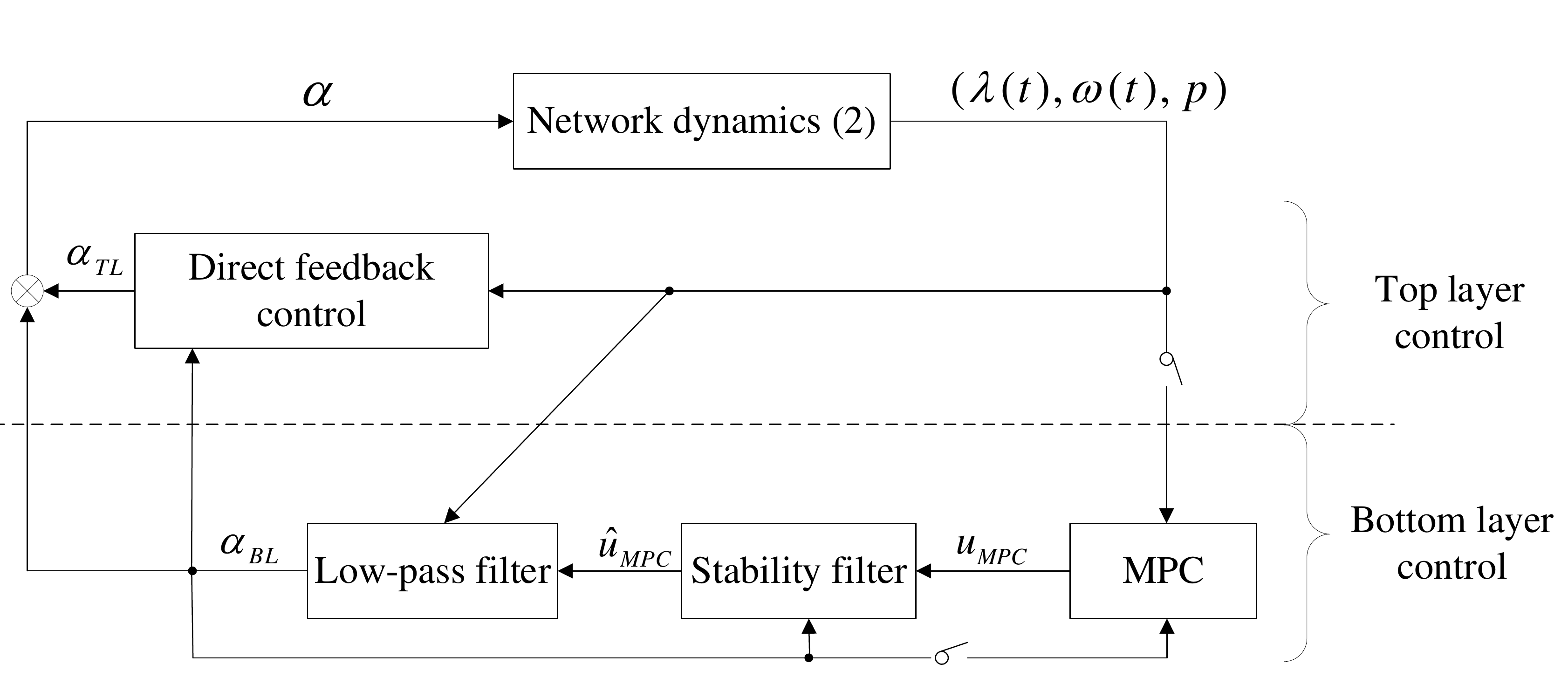}
  \caption{Block diagram of the closed-loop system with the proposed
    controller architecture.}\label{fig:block-diagram}
\end{figure*}

\section{Centralized bilayered
  controller}\label{section:centralized-control}

Here, we propose a centralized controller to address the requirements
posed in Section~\ref{section:ps}. Our idea for design starts from
considering MPC to account for the economic cooperation requirement;
however, MPC cannot be run continuously due to the computational
burden of its online optimization.  We therefore compute MPC solutions
periodically. Given the reliance of the MPC implementation on sampled
system states that are potentially outdated, we include additional
components that employ real-time state information to tune the output
of the MPC implementation and ensure stability and frequency safety.
The control signal $\alpha$ is defined by
\begin{align}\label{eqn:two-layer}
  \alpha = \alpha_{TL}+\alpha_{BL}.
\end{align}
Roughly speaking, the bottom-layer controller $\alpha_{BL}$
periodically and optimally allocates control effort, while respecting
a stability constraint and steering the frequency trajectories as a
first step to achieve frequency invariance and attractivity. The
top-layer controller $\alpha_{TL}$, implemented in real time, slightly
tunes the control trajectory generated by the bottom layer, ensuring
frequency invariance and attractivity.  Figure~\ref{fig:block-diagram}
shows the overall structure of the closed-loop system.
\change{Interestingly, as we show later, the combination of the
  stability filter, low pass filter, and direct feedback control
  stabilizes the system regardless of what is in the MPC block.}
In the following, we provide
detailed definitions of each of the design elements.

\subsection{Bottom-layer controller
  design}\label{subsection:option-loop}

We introduce here the bottom-layer control signal $\alpha_{BL}$, which
results from the combination of three components,
\change{cf. Figure~\ref{fig:block-diagram}: a MPC component, a
  stability filter, and a low-pass filter.} The MPC component
periodically samples the system state, solves an optimization problem
online, and updates its output signal $u_{MPC}$. The purpose of having
this MPC component is to efficiently allocate control resources to
achieve the frequency safety requirement. \change{The stability filter
  is designed to guarantee closed-loop asymptotic stability by
  enforcing monotonic decrease of an appropriate energy function
  (which we define later).}  Since $\hat u_{MPC}$ is merely a
piece-wise continuous signal, to avoid discontinuity in the control
signal, the low-pass filter further smooths it to generate an input
$\alpha_{BL}$ that is continuous in time.  The bottom-layer controller
by itself stabilizes the system (without the need of the top layer)
but does not guarantee frequency safety. This is precisely the role of
the top-layer design, which based on real-time system state
information, slightly tunes the control signal generated by the bottom
layer to achieve frequency safety while maintaining system
stability. Note that, except for the MPC component, all other
components can access real-time information.

Next, we introduce each component in the bottom layer and characterize
their properties.

\subsubsection{MPC component}
Based on the most recent sampled system information, the MPC component
updates its output after solving an optimization problem
online. Formally, denote $\{t^{w}\}_{w\in\naturals}$ as the
collection of sampling time instants, where
$t^{w+1}>t^{w}\geqslant  0$ holds for every $w\in\naturals$. At each
sampling time $t=t^{w}$, define a piece-wise continuous signal
$p^{fcst}_{t}:[t,t+\tilde t]\rightarrow\real^{n}$ as the predicted
value of the true power injection $p$ for the $\tilde t$ seconds
immediately following~$t$. Note that here we particularly allow the
predicted power injection to be time-varying, although its true value
is time-invariant. For convenience of exposition, we define
\begin{align*}
  x \triangleq(\lambda,\omega,\alpha_{BL})
\end{align*}
as the augmented collection of system states (the last state comes
from the low-pass filter component). Let $x(t^{w}) =
(\lambda(t^{w}),\omega(t^{w}),\alpha_{BL}(t^{w}))$ be
the augmented system state value at the sampling time $t^{w}$.

In the predicted model, we discretize the system dynamics with time
step $T>0$, and denote $N\triangleq \lceil\tilde t/T\rceil$ as the
predicted step length. At every $t=t^{w}$, the MPC component
solves the following optimization problem,
\begin{subequations}\label{opti:nonlinear}
  \begin{alignat}{2}
    & & & \min_{\hat X, \hat u,S}\quad g(\hat X,\hat u,
    S)\triangleq \sum_{k=1}^{N} \Big(
    \sum_{i\in\II{u}}c_{i}\hat
    \alpha_{BL,i}^{2}(k)+\sum_{i\in\II{\omega}}d_{i}s_{i}^{2}(k)
    \Big) \notag
    \\
    &\text{s.t.}&\quad & F\hat x(k+1)=A\hat x(k)+B_{1}\hat
    p^{fcst}(k)+B_{2}\hat u\label{opti:nonlinear-1}
    \\
    &&&\hat u\in\mathbb{A},\label{opti:control-location}
    \\
    &&&\hat x(1)=x(t^{w}),\label{opti:nonlinear-2}
    \\
    %&&&\hspace{-0.7cm}\lambda (0)=\lambda(t^{w}),\
    %\hat\omega(0)=\omega(t^{w}),\
    %\hat\alpha_{BL}(0)=\alpha_{BL}(t^{w}),\label{opti:nonlinear-2}
   % \\
    &&& \hspace{-0.7cm}\underline\omega_{i}-s_{i}(k)\leqslant
    \hat\omega_{i}(k)\leqslant \bar\omega_{i}+s_{i}(k),\ \hspace{-.0cm}\forall
    i\in\II{\omega}\hspace{-.0cm},\forall
    k\in[1,N]_{\naturals},\hspace{-0.4cm}\label{opti:nonlinear-3}
    \\
    &&&
    |\hat u_{i}|\leqslant\epsilon_{i}|\alpha_{BL,i}(t^{w})|,\quad
    \forall i\in\II{u}.\label{opti:sensitivity}
  \end{alignat}
\end{subequations}
% where for compactness, we define
% \begin{align*}
%  x\triangleq \begin{bmatrix}
% \lambda
%   \\
% \omega
%   \\
% \alpha_{BL}
%    \end{bmatrix},\quad \hat x\triangleq \begin{bmatrix}
%   \hat\lambda
%   \\
%   \hat \omega
%   \\
%   \hat\alpha_{BL}
%    \end{bmatrix}.
% \end{align*}
In this optimization,~\eqref{opti:nonlinear-1} combines the
linearized, discretized dynamics corresponding
to~\eqref{eqn:compact-form} as well as the low-pass filter introduced
later, and $\hat
x\triangleq(\hat\lambda,\hat\omega,\hat\alpha_{BL})\in\real^{m+2n}$
corresponds to the predicted system state.  Depending on the specific
discretization method, one can choose different matrices $F$,
$A\in\real^{(m+2n)\times (m+2n)}$ and $B_{1},B_{2}\in\real^{(m+2n)\times n}$
(Section~\ref{section:discretization} below contains a detailed
discussion on discretization); $\hat p^{fcst}_{t^{w}}(k)\triangleq
p^{fcst}_{t^{w}}(t^{w}+(k-1)T)$ for every
$k\in[1,N]_{\naturals}$;~\eqref{opti:control-location} specifics the
control availability for each bus;~\eqref{opti:nonlinear-2} is the
initial condition;~\eqref{opti:nonlinear-3} represents a soft version
of the frequency safety constraint, where we penalize in the cost
function the deviation of predicted frequency from its desired
bounds;~\eqref{opti:sensitivity} restricts the value of the control
input $\hat u_{i} \in \real$ with respect to the state of the low-pass
filter via a tunable parameter $\epsilon_{i}>0$;
% , which, on the one hand, delimits the relative size of $\hat
% u_{i}$ on $\alpha_{BL,i}$ via a tunable parameter $\epsilon_{i}>0$,
% and on the other, is related to the stability of the controller that
% we later explain;
%
% \marginJC{I'm not sure I understand this. ``sensitive'' or
%   ``sensitivity'' Why does that inequality capture sensitivity?? It
%   looks more like relative size. Also, the forward referencing leaves
%   the reader clueless here. }
%   
%   \marginy{Here I by sensitivity I meant $|\partial \hat
%   u_{i}/\partial\alpha_{BL,i}|$. It may lead to high-gain feedback
%   if we don't have this constraint.}
%
finally, the objective function $g$ combines the overall cost of
control effort and the penalty on the violation of the frequency
safety requirement, where $c_{i}> 0$ for each $i\in\II{u}$ and
$d_{i} >0$ for each $i\in\II{\omega}$ are design parameters.  For
compactness, we define
\begin{subequations}\label{sube:eqn:traj}
  \begin{align}
    \hat X&\triangleq[\hat x(1),\hat x(2),\cdots,x(N)],
    \\
    S &\triangleq[s(1),s(2),\cdots,s (N)],
    \\
   % \hat \Lambda&\triangleq[\hat \lambda(0),\hat \lambda(1),\cdots,\lambda(N)],
  %  \\
  %  \hat \Omega&\triangleq[\hat\omega(0),\hat\omega(1),\cdots,\hat\omega(N)],
  %  \\
  %  \hat A&\triangleq[\hat \alpha(0),\hat \alpha(1),\cdots,\hat \alpha(N)],
  %  \\
    \hat P_{t^{w}}^{fcst}&\triangleq[\hat p_{t^{w}}^{fcst}(1),\hat
    p_{t^{w}}^{fcst}(2),\cdots,\hat p_{t^{w}}^{fcst}(N)],\label{eqn:p-G-graph}
  \end{align}
\end{subequations}
where for every $k\in[1,N]_{\naturals}$, $s(k)$ is the collection
of $s_{i}(k)$'s over $i\in\II{\omega}$.
%as the collection of discretized state trajectories of flow,
%frequency, low-pass filter, and forecasted power injection.

We denote by $\RR (\mathcal{G}, \II{u}, \II{\omega},
\hat P_{t^{w}}^{fcst}, x(t^{w}))$ as the optimization
problem~\eqref{opti:nonlinear} to emphasize its dependence on network
topology, nodal indexes with exogenous control signals, nodal indexes
with transient frequency requirement, forecasted power injection, and
state values at the sampling time. We may simply use~$\RR$ if the
context is clear.  Also, we denote $(\hat X^{*}, \hat u^{*},S^{*})$ as its optimal solution.

% Specially, we may denote $\hat u^{*}$ by $\hat
% u^{*}(\mathcal{G},\II{u},\II{\omega},
% p^{fcst}_{t^{w}},f(t^{w}),\omega(t^{w}),\alpha_{BL}(t^{w}))$
% to emphasize its dependence on these values.

\begin{remark}\longthmtitle{Selection of frequency violation penalty
    coefficient}\label{rmk:violation-penalty}
  {\rm The parameter $d = \{d_i\}_{i\in\II{\omega}}$ in the objective
    function plays a fundamental rule in determining how the predicted
    frequency can exceed the safe bounds.  In the extreme case
    $d=\zeros_{|\II{\omega}|}$ (i.e., no penalty for frequency
    violation), the MPC controller loses its functionality of
    adjusting frequency.  As~$d$ grows, the controller ensures that
    the violation of the frequency safety requirement become smaller.
    The top-layer control introduced later adds additional input to
    the bottom-layer controller to ensure the frequency requirement
    is~satisfied. \oprocend }
\end{remark}

Given the open-loop optimization problem~\eqref{opti:nonlinear}, the
function $u_{MPC}$ corresponding to the MPC component in
Figure~\ref{fig:block-diagram} is defined as follows: for
$w\in\naturals$ and $ t\in[t^{w},t^{w+1})$, let
\begin{align}\label{eqn:uMPC}
  \hspace{-0.1cm} u_{MPC}(t) \!=\! \hat
  u^{*}(\mathcal{G},\II{u},\II{\omega},
   \hat P^{fcst}_{t^{w}},x(t^{w})).
\end{align}
%
%\begin{remark}\longthmtitle{Centralized implementation of MPC component}\label{rmk:implemet-centralized}
%  {\rm

Note the last two arguments that $\hat u^{*}$ depends on: forecasted
power injection value and state value of the entire network at a
sampling time. To implement~\eqref{eqn:uMPC}, a straightforward idea
is to have one operator globally gather the above two values, obtain $
\hat u^{*}$ by solving $\RR$, and finally broadcast $\hat u^{*}_{i}$
to the $i$th node.  Later in
Section~\ref{section:distributed-control}, we propose an alternative
distributed computation algorithm to reduce the computational burden.
 % However, later in Section~\ref{section:distributed-control} we
 % design one can assign multiple operators to multiple partitions of
 % the network, where each operator only gathers forecasted power
 % injection value and state value within its corresponding
 % partition..  } \oprocend
%\end{remark} 
The next result characterizes the dependence of the controller on the
sampled state values and predicted power injection.

\begin{proposition}\longthmtitle{Piece-wise affine and continuous
    dependence of optimal solution on sampling state and predicted
    power injection}\label{prop:pwl-optimal}
  Suppose $F$ is invertible, then the optimization problem
  $\RR(\mathcal{G},\II{u},\II{\omega}, \hat
  P^{fcst}_{t^{w}},x(t^{w}))$ in~\eqref{opti:nonlinear} has a unique
  optimal solution $(\hat X^{*}, \hat u^{*},S^{*})$. Furthermore,
  given $\mathcal{G},\ \II{u}$, and $\II{\omega}$, $\hat u^{*}$ is
  continuous and piece-wise affine in $(\hat
  P_{t^{w}}^{fcst},x(t^{w}))$, that is, there exist $l\in\naturals,
  \{H_{\xi}\}_{\xi=1}^{l}$, $\{K_{\xi}\}_{\xi=1}^{l}$,
  $\{h_{\xi}\}_{\xi=1}^{l}$, and $\{k_{\xi}\}_{\xi=1}^{l}$ with
  suitable dimensions such that
  \begin{align}\label{eqn:pwa-u}
    \hat u^{*}=K_{\xi}z+k_{\xi}, \text{ if } z\in\left\{ y \; | \;
      H_{\xi}y\leqslant h_{\xi} \right\} \text{ for } \xi
    \in[1,l]_{\naturals}
  \end{align}
  %
  % \marginJC{There is no dependency of time? So the input is constant
  % throughout the whole time interval $[t^j,t^{w+1}]$??}
  % \marginy{Yes. It is intentionlly designed to be a constant to
  % reduce the %computational time for solving the MPC.}
  %
  holds for every $z\in\real^{(N+2)n+m}$, where $z$ is the collection
  of $(\hat P_{t^{w}}^{fcst},x(t^{w}))$ in column-vector form.
\end{proposition}
\begin{proof}
  We start by noting that $\RR$ is feasible (hence at least one
  optimal solution exists) for any given $z$.  This is because, given
  a state trajectory $\hat X$ of~\eqref{opti:nonlinear-1} with input
  $\hat u=\zeros_{n}$ and initial condition~\eqref{opti:nonlinear-2},
  choosing a sufficiently large $s(k)$ for each
  $k\in[1,N]_{\naturals}$ makes it satisfy
  constraint~\eqref{opti:nonlinear-3}.  The uniqueness follows from
  the facts that I) $g$ is strongly convex in $(\hat u,S)$;
  II) $\hat X$ is uniquely and linearly determined by $\hat u$; III)
  all constraints are linear in $(\hat X, \hat u,S)$.  To
  show continuity and piece-wise affinity, we separately consider
  $2^{|\II{u}|}$ cases, depending on the sign of each
  $\{\alpha_{BL,i}(t^{w})\}_{i\in\II{u}}$. Specifically, let
  $\eta\triangleq\{\eta_{i}\}_{i\in\II{u}}\in\{1,-1\}^{|\II{u}|}$ and
  define $\mathfrak{B}^{\eta}\triangleq\left\{ z
    \big|(-1)^{\eta_{i}}\alpha_{BL,i}(t^{w})\geqslant 0,\ \forall
    i\in\II{u} \right\}$.  Note that every $z$ lies in at least one of
  these sets and that, in any $\mathfrak{B}^{\eta}$, the sign of each
  $\alpha_{BL,i}(t^{w})$ with $i\in\II{u}$ is fixed.  Hence all
  the $|\II{u}|$ constraints in~\eqref{opti:sensitivity} can be
  transformed into one of the following forms
  \begin{subequations}\label{sube:ineq:sgn-alpha}
    \begin{alignat}{2}
    \hspace{-0.38cm}  -\epsilon_{i}\alpha_{BL,i}(t^{w}) & \leqslant \hat
      u_{i}\leqslant\epsilon_{i}\alpha_{BL,i}(t^{w}) & & \text{if
      }\alpha_{BL,i}(t^{w})\geqslant 0,\label{sube:ineq:sgn-alpha-1}
      \\
        \hspace{-0.2cm}  \epsilon_{i}\alpha_{BL,i}(t^{w}) & \leqslant \hat
      u_{i}\leqslant-\epsilon_{i}\alpha_{BL,i}(t^{w}) && \ 
      \text{if } \alpha_{BL,i}(t^{w})\leqslant
      0 . \label{sube:ineq:sgn-alpha-2}
    \end{alignat}
  \end{subequations}
  Note that if $\alpha_{BL,i}(t^{w})=0$, then $\hat
  u_{i}=0$. Therefore, in every $\mathfrak{B}^{\eta}$, $z$ appears in
  $\RR$ in a linear fashion; hence,
  it is easy to re-write $\RR$ into the following form:
  \begin{alignat}{2}
    &\min_{q} & \quad & q^{T}Vq\notag
    \\
    &\text{s.t.}&\quad &Gq\leqslant W+J^{\eta}z,
  \end{alignat}
  where $q$ is the collection of $(\hat X, \hat
  u,S)$ in vector form and $V\succeq 0$, $G$, $W$ and $J^{\eta}$
  are matrices with suitable dimensions. Note that only $J^{\eta}$
  depends on $\eta$. By~\cite[Theorem~1.12]{FB:03}, for every
  $\eta\in\{-1,1\}^{|\II{u}|}$, $s^{*}$ is a continuous and piece-wise
  affine function of $z$ whenever $z\in\mathfrak{B^{\eta}}$.  Since
  each $\mathfrak{B^{\eta}}$ consists of only linear constraints and
  the union of all $\mathfrak{B^{\eta}}$'s with
  $\eta\in\{1,-1\}^{|\II{u}|}$ is $\real^{(N+2)n+m}$, one has that
  $s^{*}$ is piece-wise affine in $z$ on $\real^{(N+2)n+m}$. Lastly,
  to show the continuous dependence of $s^{*}$ on $z$ on
  $\real^{(N+2)n+m}$, note that since such a dependence holds on every
  closed set $\mathfrak{B^{\eta}}$, we only need to prove that $s^{*}$
  is unique for every $z$ lying on the boundary shared by different
  $\mathfrak{B^{\eta}}$'s. This holds trivially as $s^{*}$ is unique
  for every $z\in\real^{(N+2)n+m}$, which we have proven above. 
\end{proof}

Notice that the continuity and piece-wise affinity established in
Proposition~\ref{prop:pwl-optimal} together suffice to ensure that
$\hat u^{*}$ is globally Lipschitz in $z$, and hence in the sampled
system state.  To see this point, one can easily check that $\max_{\xi
  \in [1,l]_{\naturals}}\|K_{\xi}\|$ qualifies as a global Lipschitz
constant.

In addition, Proposition~\ref{prop:pwl-optimal} also suggests an
alternative to directly solve~$\RR$ without treating it as an
optimization problem. Specifically, we can first compute and store
$\{H_{\xi}\}_{\xi=1}^{l}$, $\{K_{\xi}\}_{\xi=1}^{l}$,
$\{h_{\xi}\}_{\xi=1}^{l}$, and $\{k_{\xi}\}_{\xi=1}^{l}$, and then compute
$\hat u^{*}$ online via~\eqref{eqn:pwa-u}. However, such an approach,
usually called explicit MPC~\cite{AA-BA:09}, suffers from the curse of
dimensionality, in that the number of regions $l$ grows exponentially
fast in $m+n$, input size $|\II{u}|$, and horizon length~$N$.

\subsubsection{Stability and low-pass filters}
Here we introduce the stability and low-pass filters, explain the
motivation behind their definitions and characterize their
properties. Note that the sampling mechanism used for the MPC
component inevitably introduces delays in the bottom
layer. Specifically, for any time $t\in(t^{w},t^{w+1})$, i.e., between
two adjacent sampling times, $u_{MPC}(t)$ is fully determined by the
old sampled system information at time $t^{w}$, as opposed to the
current information. To eliminate the potential negative effect of
delay on system stability, \change{we introduce a stability filter to
  enforce closed-loop stability}. The low-pass filter after the
stability filter simply smooths the output of the stability filter to
ensure that the output of the bottom layer is continuous in time.
Formally, for every $i\in\mathcal{I}$ at any $t\geqslant 0$, define
the stability filter as
\begin{align}\label{eqn:stability-filter}
  \hspace{-1cm} \hat u_{MPC,i}(\alpha_{BL}(t),u_{MPC}(t)) &\notag
  \\
  &\hspace{-3cm}=\sat(u_{MPC,i}(t);\epsilon_{i}|\alpha_{BL,i}(t)|,-\epsilon_{i}|\alpha_{BL,i}(t)|),
% \\
%  \hat u_{MPC,i}=& \notag
%  \\
%    &\hspace{-3.9cm}\begin{cases}
%      0 & \hspace{-1cm}\text{if $\exists\tau\in[t^{w},t]$ s.t. $|u_{MPC,i}(t)|> \epsilon_{i}|\alpha_{i}(\tau)|$,}
%      \\
%      u_{MPC,i}(t) & \hspace{1.2cm}\text{otherwise.}
%    \end{cases}
\end{align}
and define the low-pass filter as                      %
\begin{align}\label{eqn:lp-filter}
  \dot\alpha_{BL,i}(t)&=-\frac{1}{\tau_{i}}\alpha_{BL,i}(t)-\omega_{i}(t)+\hat
  u_{MPC,i}(t),\quad\forall i\in\II{u},\notag
  \\
  \alpha_{BL,i}&\equiv0,\quad\forall i\in\mathcal{I}\backslash\II{u},
\end{align}
where the tunable parameter $\tau_{i}\in\real_{>}$ determines the
bandwidth of the low-pass filter. In addition, although for
compactness we define a stability filter for every $i\in\mathcal{I}$,
one can easily see that $\hat u_{MPC,i}\equiv0$ for every
$i\in\mathcal{I}\backslash\II{u}$.

Both the stability and the low-pass filters possess a
natural distributed structure: for each $i\in\II{u}$, $\alpha_{BL,i}$
only depends $\omega_{i}$ and $\hat u_{MPC,i}$, where the latter one
only depends on $u_{MPC,i}$ and $\alpha_{BL,i}$. This implies that to
implement $\hat u_{MPC,i}$ and $\alpha_{BL,i}$, it only requires
local information at node $i$. Throughout the rest of the paper, we
interchangeably use $\hat u_{MPC,i}(\alpha_{BL}(t),u_{MPC}(t))$ and
$\hat u_{MPC,i}(t)$ for simplicity.

The next result establishes that $\hat u_{MPC}$ is Lipschitz
continuous in the system state and an important property of the
bottom-layer controller $\alpha_{BL}$ that we use later to establish
stability.

\begin{lemma}\longthmtitle{Lipschitz continuity and stability
    condition} \label{lemma:Lipschitz-stability}
  For the signal $\hat u_{MPC}$ defined
  in~\eqref{eqn:stability-filter}, $\hat u_{MPC}$ is Lipschitz in
  system state at every sampling time $t=t^{w}$ with
  $w\in\naturals$. Furthermore, if $\alpha_{TL}$ is Lipschitz in
  system state, then both $\alpha_{TL}$ and $\alpha_{BL}$ are
  continuous in time.  Additionally,
  \begin{align}\label{ineq:stability-condition}
    \alpha_{BL,i}(t)\hat
   u_{MPC,i}(t) \leqslant \epsilon_{i}
  \alpha_{BL,i}^{2}(t),\quad\forall t\geqslant 0,\ \forall
  i\in\mathcal{I}.
  \end{align}
\end{lemma}
\begin{proof}
  If $t = t^{w}$, then since $|\hat
  u_{i}^{*}|\leqslant\epsilon_{i}|\alpha_{BL,i}(t^{w})|$
  by~\eqref{opti:sensitivity} and $u_{MPC,i}(t^{w})=\hat u_{i}^{*}$
  for every $i\in\II{u}$, using~\eqref{eqn:stability-filter} we deduce
  that $\hat u_{MPC,i}(\alpha_{BL}(t),u_{MPC}(t))|_{t=t^{w}}=\hat
  u_{i}^{*}$. The Lipschitz continuity follows by
  Proposition~\ref{prop:pwl-optimal}. To show the time-domain
  continuity, since $\hat u_{MPC}$ is Lipschitz at every sampling
  point and the top-layer controller is also Lipschitz by hypothesis
  (we demonstrate this point later in Section~\ref{subsection:top
    layer}), one has that the solutions of both $\alpha_{TL}$ and the
  closed-loop system~\eqref{eqn:compact-form} exist and are unique and
  continuous in time. Note that $u_{MPC}$ in~\eqref{eqn:uMPC} is
  defined to be a piece-wise constant signal. One has,
  by~\eqref{eqn:stability-filter}, that $\hat u_{MPC}$ is piece-wise
  continuous, which further makes $\alpha_{BL}$ a continuous signal in
  time due to the low-pass filter.
  Condition~\eqref{ineq:stability-condition} simply follows from the
  definition of saturation function.
\end{proof}

\begin{remark}\longthmtitle{Link between designs of the MPC component
    and stability filter}\label{rmk:independent-design} {\rm Note
    that, regardless of the MPC component output $u_{MPC}$, the output
    of the stability filter $\hat u_{MPC}$ defined
    in~\eqref{eqn:stability-filter} always meets
    condition~\eqref{ineq:stability-condition}. This 
    implies that any inaccuracy in the MPC component (e.g., errors in
    sampled state measurement, forecasted power injection, or system
    parameters) cannot cause instability.
  %provides flexibility in the MPC component design
  %and robustness against, e.g., inaccuracy in sampled state
  %measurement, forecasted power injection, as well as system
  %parameters.
    However, to ensure the Lipschitz continuity in
    Lemma~\ref{lemma:Lipschitz-stability}, we formulate
    constraint~\eqref{opti:sensitivity} employing the same coefficient
    $\epsilon_{i}$ in the stability
    filter~\eqref{eqn:stability-filter}. It is in this sense that both
    are linked.  \oprocend }
\end{remark}

\begin{remark}\longthmtitle{Continuous versus periodic sampling in the
    MPC component}\label{rmk:unnecessity-filter}
  {\rm \change{Note that if the MPC component were to sample the
      system state in a continuous fashion instead, then the
      constraint~\eqref{opti:sensitivity} would ensure that the output
      of the MPC component already satisfies the stability
      condition~\eqref{ineq:stability-condition}, and hence there
      would be no need for the stability filter. In this regard, the
      role of the stability filter is to filter out the unstable parts
      in $u_{MPC}$ caused by non-continuous sampling.}  \oprocend }
\end{remark}

\subsection{Discretization with sparsity
  preservation}\label{section:discretization}

As we have introduced the dynamics of the low-pass and stability
filters, we are now able to explicitly explain the computation of
matrices $F$, $A$, $B_{1}$ and $B_{2}$ in the prediction
model~\eqref{opti:nonlinear-1}.  We first construct a continuous-time linear
model by neglecting the top-layer controller and the stability filter
($\alpha\approx\alpha_{BL}$ and $\hat u_{MPC}\approx u_{MPC}$), and
then linearizing the nonlinear dynamics in
Figure~\ref{fig:block-diagram}. Our second step consists of
appropriately discretizing this linear model.

Notice that the transformation from a nonlinear continuous-time
nonlinear model to a discrete one does not affect closed-loop system
stability due to the presence of the stability filter. In fact, any
prediction model in the MPC component cannot jeopardize stability
(cf. Remark~\ref{rmk:independent-design}). On the other hand, such a
model simplification is reasonable since $\alpha_{BL}$ is designed to
only slightly tune the control signal, and we have described in
Remark~\ref{rmk:unnecessity-filter} how the stability filter barely
changes its input.

We obtain the linear model by assuming $\alpha\approx\alpha_{BL}$ and
$\hat u_{MPC}\approx u_{MPC}$, and approximating the dynamics in
Figure~\ref{fig:block-diagram}~by
\begin{align}\label{eqn:approxiate-dynamics}
  \dot \lambda(t)&=D\omega(t),\notag
  \\
  M\dot\omega(t)&=-E\omega(t)-D^{T}Y_{b}\lambda(t)+p+\alpha_{BL}(t),\notag
  \\
  M_{i}\dot\alpha_{BL,i}(t)&=-\frac{1}{\tau_{i}}\alpha_{BL,i}(t)-\omega_{i}(t)+
  \hat u_{MPC,i}(t),\quad\forall i\in\II{u},\notag
  \\
  \alpha_{BL,i}&\equiv0,\quad\forall i\in\mathcal{I}\backslash\II{u},
\end{align}
where the first two equations come from~\eqref{eqn:compact-form} by
linearizing the nonlinear sinusoid function via
$\sin(Y_{b}\lambda(t))\approx Y_{b}\lambda(t)$.  Now we re-write the
above linear dynamics into the compact form,
\begin{align}\label{eqn:approxiate-dynamics-compact}
  \tilde G \dot x(t)=\tilde A x(t)+\tilde B_{1}p+\tilde B_{2} u_{MPC}(t),
\end{align}
for certain matrices $\tilde A$, $\tilde B_{1}$, and $\tilde B_{2}$,
with $\tilde A$ stable~\cite{AP:12} and with $\tilde G$ a diagonal
matrix whose diagonals are $1$, $M_{i}$ with $i\in[1,n]_{\naturals}$,
or $0$. Additionally, one can easily check that the linearized
dynamics~\eqref{eqn:approxiate-dynamics}
and~\eqref{eqn:approxiate-dynamics-compact} preserve the locality
of~\eqref{eqn:compact-form-2} and~\eqref{eqn:lp-filter}.

We consider the following three discretization methods with step size
$ T>0$ to construct $F$, $A$, $B_{1}$, and $B_{2}$ matrices
in~\eqref{opti:nonlinear-1} approximating the continuous
dynamics~\eqref{eqn:approxiate-dynamics-compact}. For explanatory
simplicity, we here assume $\tilde G$ is invertible.
\begin{enumerate}[label=\alph*)]
\item Impulse invariant discretization:
  \begin{align}\label{sube:discretization-exp}
    F\triangleq I_{m+2n},\ A\triangleq e^{\tilde G^{-1}\tilde A
       T},\ B_{s}\triangleq\int_{0}^{ T}e^{\tilde
      G^{-1}\tilde A \tau}\text{d}\tau \tilde B_{s},\ s=1,2,
  \end{align}
\item Forward Euler discretization:
  \begin{align}\label{sube:discretization-linear-forward}
    F\triangleq \tilde G,\ A\triangleq  T\tilde A+\tilde G, \
    B_{s}\triangleq  T\tilde B_{s},\ s=1,2,
  \end{align}
\item Backward Euler discretization:
  \begin{align}\label{sube:discretization-linear}
    F\triangleq \tilde G- T\tilde A,\ A\triangleq I_{m+2n}, \
    B_{s}\triangleq  T \tilde B_{s},\ s=1,2,
  \end{align}
\end{enumerate}
where $F$ should be invertible for uniqueness of solution of the
discretized dynamics.

Note that with a fixed $ T$, the impulse invariant and
backward Euler methods usually have better approximation accuracy than
the forward Euler method. In fact, since all eigenvalues of $\tilde A$
have non-positive real part, a basic discretization requirement is
that all eigenvalues of $F^{-1}A$ are in the unit circle to maintain
stability.  One can easily prove that the impulse invariant and
backward Euler discretization always meet this requirement for any
$ T>0$, but the forward Euler method requires a sufficiently
small $ T$ to preserve stability; therefore, with a same
predicted time horizon $\tilde t$, the forward Euler method has the
largest predicted step length $N$ and hence makes the optimization
problem~$\RR$ harder to solve. On the other hand, the backward Euler
method might require a small enough $ T$ to guarantee the
invertibility of $F$, but numerically we have found this to be easily
satisfiable.  Therefore, we set aside the forward Euler method from
our considerations of discretization. On the other hand, the impulse
invariant method fails to preserve the sparsity of $\tilde A$, $\tilde
B_{1}$, and $\tilde B_{2}$, which are essential for the design of
distributed solvers of~$\RR$.  Instead, the matrices $F$, $A$, $B_{1}$
and $B_{2}$ resulting from the backward Euler discretization are all
sparse.  This justifies our choice, throughout the rest of the paper,
of the backward Euler method for discretization.

\subsection{Top-layer controller design}\label{subsection:top layer}

In this section we describe the top-layer controller.  By design,
cf.~\eqref{opti:nonlinear}, the bottom-layer controller makes a
trade-off between the control cost and the violation of frequency
safety, and hence does not strictly guarantee the latter.  This is
precisely the objective of the top-layer controller: ensuring
frequency safety at all times by slightly adjusting, if necessary, the
effect of the bottom-layer controller.  Formally, for every
$i\in\II{\omega}$, let $\bar\gamma_{i},\underline \gamma_{i}>0$, and
$\underline \omega_{i}^{\text{thr}}, \bar\omega_{i}^{\text{thr}} \in
\real$ with $\underline \omega_{i}< \underline \omega_{i}^{\text{thr}}
< 0 < \bar\omega_{i}^{\text{thr}}<\bar\omega_{i}$.  We use the design
from~\cite{YZ-JC:19-auto} for the top layer. For $i\in\II{\omega }$, $
\alpha_{TL,i}(x(t),p) $ takes the form
\begin{align}\label{eqn:second-layer-control}
  \begin{cases}
    \min\{0,\frac{\bar\gamma_{i}(\bar\omega_{i}-
      \omega_{i}(t))}{\omega_{i}(t)-\bar\omega_{i}^{\text{thr}}}+v_{i}(x(t),p)\}
    & \omega_{i}(t)>\bar\omega_{i}^{\text{thr}},
    \\
    0 & \underline\omega_{i}^{\text{thr}}\leqslant
    \omega_{i}(t)\leqslant \bar\omega_{i}^{\text{thr}},
    \\
    \max\{0,\frac{\underline\gamma_{i}(\underline\omega_{i}-\omega_{i}(t))
    }{\underline\omega_{i}^{\text{thr}}-\omega_{i}(t)}+v_{i}(x(t),p)\}
    & \omega_{i}(t)<\underline\omega_{i}^{\text{thr}},
  \end{cases}
\end{align}
where
\begin{align*}
  v_{i}(x(t),p) & \triangleq
  E_{i}\omega_{i}(t)+[D^{T}]_{i}\sin(Y_{b}\lambda(t))-p_{i}-\alpha_{BL,i}(t),
\end{align*}
and for $i \in \mathcal{I} \backslash\II{\omega}$, simply
$\alpha_{TL,i} \equiv0$.  The top-layer controller can be implemented
in a decentralized fashion: for each $\alpha_{TL,i}$ with
$i\in\II{\omega}$ on bus $i$, its implementation only requires the bus
frequency $\omega_{i}$, aggregated power flow
$[D^{T}]_{i}\sin(Y_{b}\lambda)$, power injection $p_{i}$, and $i$th
component of the bottom-layer signal $\alpha_{BL,i}$, all of which are
local to bus~$i$. Additionally, similarly to~\cite{YZ-JC:19-auto}, one
can show that $\alpha_{TL}$ is locally Lipschitz in $x$. For brevity,
we may use $\alpha_{TL,i}(x(t),p)$ (respectively, $v_{i}(x(t),p)$) and
$\alpha_{TL,i}(t)$ (respectively $v_{i}(t)$) interchangeably.

Each $\alpha_{TL,i}$, with $i\in\II{\omega}$, behaves as a passive and
myopic transient frequency regulator without prediction
capabilities. We offer the following observations about its
definition: first, $\alpha_{TL,i}$ only depends on local system
information and does not incorporate any global knowledge; second,
$\alpha_{TL,i}$ vanishes as long as the current frequency is within
$[\underline\omega_{i}^{\text{thr}},\bar\omega_{i}^{\text{thr}}]$, a
subset of the safe frequency interval, with no consideration for the
possibility of future large disturbances; third, $\alpha_{TL,i}$ can
be non-zero when the current frequency is out of
$[\underline\omega_{i}^{\text{thr}},\bar\omega_{i}^{\text{thr}}]$ and
hence close to the safe frequency boundaries. However, this could also
lead to over-reaction, especially when $\bar \gamma_{i}$ and
$\underline\gamma_{i}$ are small, as the disturbance may disappear
suddenly, in which case even without the top-layer controller, the
frequency would remain safe afterwards.  As pointed out above, the
top-layer controller only steps in if the input from the bottom-layer
controller is not sufficient to ensure frequency safety.

% \begin{figure*}[htb]
%   \begin{subequations}\label{eqn:second-layer-control}
%     \begin{alignat}{2}
%     &\forall i\in\II{\omega }, \text{ let }
%     \alpha_{TL,i}(x(t),p) &&=
%     \begin{cases}
%       \min\{0,\frac{\bar\gamma_{i}(\bar\omega_{i}-\omega_{i}(t))}{\omega_{i}(t)-\bar\omega_{i}^{\text{thr}}}+v_{i}(x(t),p)\} 
%       & \omega_{i}(t)>\bar\omega_{i}^{\text{thr}},
%       \\
%       0 & \underline\omega_{i}^{\text{thr}}\leqslant
%       \omega_{i}(t)\leqslant \bar\omega_{i}^{\text{thr}},
%       \\
%       \max\{0,\frac{\underline\gamma_{i}(\underline\omega_{i}-\omega_{i}(t))
%       }{\underline\omega_{i}^{\text{thr}}-\omega_{i}(t)}+v_{i}(x(t),p)\}
%       & \omega_{i}(t)<\underline\omega_{i}^{\text{thr}},
%     \end{cases}
%     \\
% &\hspace{2.4cm}v_{i}(x(t),p)&&\triangleq
%   E_{i}\omega_{i}(t)+[D^{T}]_{i}\sin(Y_{b}\lambda(t))-p_{i}-\alpha_{BL,i}(t),
%   \\
%  & \forall i\in\mathcal{I}\backslash\II{\omega},\text{ let }\hspace{0.6cm}\alpha_{TL,i}&&\equiv0.
%   \end{alignat}
% \end{subequations}
%   \hrulefill
%   \vspace{-.5cm}
% \end{figure*}

\subsection{Frequency safety and local asymptotic stability}
Having introduced the elements of both layers in
Figure~\ref{fig:block-diagram}, we are now ready to show that the
proposed centralized control strategy meets requirements~(i)-(iv) in
Section~\ref{section:ps}. We focus on the first two requirements,
since we have already established the Lipschitz continuity of each
individual component, and the MPC component by design takes care of
the economic cooperation among  the controlled buses.

For the open-loop system~\eqref{eqn:compact-form} with $\alpha\equiv
0_{n}$, under condition~\eqref{ineq:sufficient-eq}, the following
energy function~\cite{TLV-HDN-AM-JS-KT:18} is identified to prove
local asymptotic stability and estimate the region of attraction,
\begin{align}\label{eqn:energy-func}
  V(x)\triangleq\frac{1}{2}\sum_{i=1}^{\bar n}M_{i} \omega_{i}^{2} +
  \sum_{j=1}^{m}[Y_{b}]_{j,j}a(\lambda_{j}),
\end{align}
where $a(\lambda_{j}) \triangleq \cos\lambda_{j}^{\infty} -
\cos\lambda_{j} - \lambda_{j}\sin\lambda_{j}^{\infty} +
\lambda_{j}^{\infty}\sin\lambda_{j}^{\infty}$ for every
$j\in[1,m]_{\naturals}$. \change{For notational simplicity, here we
  assume that the first $\bar n$ nodes have strictly positive inertia,
  whereas the rest $n-\bar n$ nodes have zero inertia.}  Due to the
extra dynamics introduced by the low-pass filter, we here consider the
following energy function for the closed-loop system,
\begin{align}\label{eqn:engergy-func-cl}
  \bar V(x)= V(x)+\frac{1}{2}\sum_{i\leqslant \bar n, i\in\II{u}}M_{i}\alpha^{2}_{BL,i}.
\end{align}
Furthermore, define the level set 
\begin{align}\label{set:region}
  \mathcal{T}_{\rho}\triangleq\setdef{x}{\lambda\in \Upsilon_{cl},\
    \bar V(x)\leqslant \rho c},
\end{align}
where $\rho\geqslant 0$ and $c \triangleq
\min_{\tilde\lambda\in\partial\Upsilon}\bar
V(\tilde\lambda,\zeros_{n},\zeros_{n})$.
Now we are ready to prove that system~\eqref{eqn:compact-form} with
the proposed controller guarantees frequency safety and local
asymptotic stability jointly.

\begin{theorem}\longthmtitle{Bilayered control with
    stability and frequency guarantees}\label{thm:two-layer-control}
  Under condition~\eqref{ineq:sufficient-eq}, assume that
  $\epsilon_{i}\tau_{i}<1$ for every $i\in\II{u}$, and $x(0)\in
  \mathcal{T}$, then the system~\eqref{eqn:compact-form} with the
  bilayered controller defined
  by~\eqref{eqn:two-layer},~\eqref{eqn:uMPC},~\eqref{eqn:stability-filter},~\eqref{eqn:lp-filter},
  and \eqref{eqn:second-layer-control} satisfies
  \begin{enumerate}
    % \item\label{item:solution} The solution of the closed-loop system
    %   exists and is unique for every $t\geqslant 0$.
  \item\label{item:invariance} for any $i\in\II{\omega}$, if
    $\omega_{i}(0)\in[\underline\omega_{i},\bar\omega_{i}]$, then
    $\omega_{i}(t)\in[\underline\omega_{i},\bar\omega_{i}]$ for every
    $t\geqslant 0$;
  \item\label{item:attractivity} for any $i\in\II{\omega}$, if
    $\omega_{i}(0)\notin [\underline\omega_{i},\bar\omega_{i}]$, then
    there exists $t_{0}$ such that
    $\omega_{i}(t)\in[\underline\omega_{i},\bar\omega_{i}]$ for every
    $t\geqslant t_{0}$. Furthermore, $\omega_{i}(t)$ monotonically
    approaches $[\underline\omega_{i},\bar\omega_{i}]$ before entering~it;
  \item\label{item:convergence} if the initial state
    $(\lambda(0),\omega(0),\alpha_{BL}(0))$ is in $\mathcal{T}_{\rho}$
    for some $0<\rho<1$, then $(\lambda(t),\omega(t),\alpha_{BL}(t))$
    stays in $\mathcal{T}_{\rho}$ for all $t> 0$, and converges to
    $(\lambda_{\infty},\zeros_{n},\zeros_{n})$. Furthermore,
    $\alpha(t)$, $\alpha_{TL}(t)$, $\alpha_{BL}(t)$, \change{$\hat
      u_{MPC}(t)$, and $u_{MPC}(t)$ all converge to $\zeros_{n}$ as
      $t\rightarrow \infty$. }
  \end{enumerate}
\end{theorem}
\begin{proof}
  It is easy to see that statement~\ref{item:invariance} is equivalent
  to asking that, for any $i\in\II{\omega}$ at any $t\geqslant 0$,
  \begin{subequations}\label{eqn:ith-dyanmics}
    \begin{align}
      \dot\omega_{i}(t)\leqslant 0\text{ if
      }\omega_{i}(t)=\bar\omega_{i},\label{eqn:ith-dyanmics-a}
      \\
      \dot\omega_{i}(t)\geqslant 0\text{ if
      }\omega_{i}(t)=\underline\omega_{i}.\label{eqn:ith-dyanmics-b}
   \end{align}
 \end{subequations}
 For simplicity, we only prove~\eqref{eqn:ith-dyanmics-a},
 and~\eqref{eqn:ith-dyanmics-b} follows similarly. Note that
 by~\eqref{eqn:compact-form-2},~\eqref{eqn:two-layer},
 and~\eqref{eqn:second-layer-control}, one has
 \begin{align}\label{eqn:Mi-dynamics}
   M_{i}\dot\omega_{i}(t) &=
   -E_{i}\omega_{i}(t)-[D^{T}]_{i}\sin(Y_{b}\lambda(t))+p_{i}+\alpha_{i}(t)\notag
   \\
   & =
   -E_{i}\omega_{i}(t)-[D^{T}]_{i}\sin(Y_{b}\lambda(t))+p_{i}+\alpha_{BL,i}(t)
   + \alpha_{TL,i}(t)\notag
   \\
   &=-v_{i}(t)+\alpha_{TL,i}(t).
 \end{align}
 $\omega_{i}(t)=\bar\omega_{i}$, then
 $-v_{i}(t)+\alpha_{TL,i}(t)=-v_{i}(t)+\min\{0,v_{i}(t)\}\leqslant0$;
 hence condition~\eqref{eqn:ith-dyanmics-a} holds for every
 $i\in\II{\omega}$ with $M_{i}>0$.

 \change{To establish the result for the case when $M_{i}=0$, we
   reason as follows. Starting from the last line
   of~\eqref{eqn:Mi-dynamics}, the following holds when
   $\omega_{i}(t)>\bar\omega_{i}^{\text{thr}}$,
   \begin{align*}
     E_{i}\omega_{i}(t)&=E_{i}\omega_{i}(t)-v_{i}(t)+\alpha_{TL,i}(t)
     \\
     &=\min\{E_{i}\omega_{i}(t)-v_{i}(t),E_{i}\omega_{i}(t) +
     \frac{\bar\gamma_{i}(\bar\omega_{i}-
       \omega_{i}(t))}{\omega_{i}(t)-\bar\omega_{i}^{\text{thr}}}\}
     \\
     &\leqslant
     E_{i}\omega_{i}(t)+\frac{\bar\gamma_{i}(\bar\omega_{i}-
       \omega_{i}(t))}{\omega_{i}(t)-\bar\omega_{i}^{\text{thr}}},
   \end{align*}
   and hence $\frac{\bar\gamma_{i}(\bar\omega_{i}-
     \omega_{i}(t))}{\omega_{i}(t)-\bar\omega_{i}^{\text{thr}}}\geqslant
   0$, implying that
   $\omega_{i}(t)\leqslant\bar\omega_{i}$. Similarly, one can prove
   that $\omega_{i}(t)\geqslant \underline\omega_{i}$ for every
   $t\geqslant 0$.  }
 
 Note that~\ref{item:attractivity} follows from~\ref{item:invariance}
 and~\ref{item:convergence}. This is because, for any
 $i\in\mathcal{I}$, if $\omega_{i}$ converges to
 $0\in(\underline\omega_{i},\bar\omega_{i})$, there must exist a
 finite time $t_{0}$ such that
 $\omega_{i}(t_{0})\in[\underline\omega_{i},\bar\omega_{i}]$, which,
 by~\ref{item:invariance}, implies that
 $\omega_{i}(t)\in[\underline\omega_{i},\bar\omega_{i}]$ at any
 $t\geqslant t_{0}$.  We then prove statement~\ref{item:convergence},
 To show the invariance of $\mathcal{T}_{\rho}$, first, it is easy to
 see that $c>0$ by noticing that
 $\lambda_{\infty}\nin\partial\Upsilon$, and
 $V(\tilde\lambda,\zeros_{n},\zeros_{n})$ is non-negative, equaling 0
 if and only if $\tilde \lambda=\lambda_{\infty}$. Next, we show that
 $\dot {\bar V}\leqslant0$ for every $x\in\mathcal{T}_{\rho}$. We
 obtain after some computations that
 \begin{align*}
   \dot{\bar V}
   =&-\omega^{T}(t)E\omega(t)+\sum_{i\leqslant\bar n , i\in\II{\omega}}\omega_{i}(t)\alpha_{TL,i}(t)
   \\
   &-\sum_{,i\in\II{u}}\left(\frac{1}{\tau_{i}}\alpha_{BL,i}^{2}(t)-\alpha_{BL,i}(t)\hat
     u_{MPC,i}(t)\right).
 \end{align*}
 Note that by the definition of $\alpha_{TL}$
 in~\eqref{eqn:second-layer-control},
 $\omega_{i}(t)\alpha_{TL,i}(t)\leqslant0$ holds for every
 $i\in\II{\omega}$ at every $t\geqslant 0$, in that
 $\alpha_{TL,i}(t)=0$ whenever
 $\underline\omega_{i}^{\text{thr}}\leqslant \omega_{i}(t)\leqslant
 \bar\omega_{i}^{\text{thr}}$, and $\alpha_{TL,i}(t)\geqslant 0$
 (reps. $\leqslant 0$) if $\omega_{i}(t)\geqslant
 \bar\omega_{i}^{\text{thr}}>0$ (respectively, $\omega_{i}(t)\leqslant
 \underline\omega_{i}^{\text{thr}}<0$). Therefore, together with
 condition~\eqref{ineq:stability-condition} in
 Lemma~\ref{lemma:Lipschitz-stability}, we have
 \begin{align}\label{eqn:d-bar-V}
   \dot {\bar V}\leqslant-\omega^{T}(t)E\omega(t)-\sum_{ i\in\II{u}}
   (\frac{1}{\tau_{i}}-\epsilon_{i})\alpha_{BL,i}^{2}(t) \leqslant0,
 \end{align}
 and hence $\bar V(x(t))\leqslant\rho c$ for all $t\geqslant 0$.
 Finally, by the definition of $c$, one can check that $\lambda$ stays
 in $\Upsilon_{\text{cl}}$ all the time, otherwise there exists some
 $t\geqslant 0$ such that $\lambda(t)\in\partial\Upsilon$, resulting
 in $\bar V(x(t))\geqslant c>\rho c$. Therefore, the set
 $\mathcal{T}_{\rho}$ is invariant.
 
 The convergence of state follows by LaSalle Invariance
 Principle~\cite[Theorem~4.4]{HKK:02}. Specifically, $\omega(t)$ and
 $\alpha_{BL}(t)$ converge to $\zeros_{n}$ (notice that
 $\alpha_{BL,i}\equiv0$ for each
 $i\in\mathcal{I}\backslash\II{u}$). Next we show that
 $\lim_{t\rightarrow\infty}\alpha_{TL,i}(t)=0$ for every
 $i\in\II{\omega}$, which implies that
 $\lim_{t\rightarrow\infty}\alpha_{TL}(t)=\zeros_{n}$ as
 $\alpha_{TL,i}\equiv 0$ for each
 $i\in\mathcal{I}\backslash\II{\omega}$. This simply follows
 from~\eqref{eqn:second-layer-control} since $\alpha_{TL,i}(t)=0$
 whenever
 $\underline\omega_{i}^{\text{thr}}\leqslant\omega_{i}(t)\leqslant
 \bar\omega_{i}^{\text{thr}}$, where $0\in(
 \underline\omega_{i}^{\text{thr}}, \bar\omega_{i}^{\text{thr}})$, and
 we have shown that
 $\lim_{t\rightarrow\infty}\omega(t)=\zeros_{n}$. The convergence of
 $\alpha(t)$ follows from its
 definition~\eqref{eqn:two-layer}. \change{
   Since~\eqref{eqn:stability-filter} implies that $|\hat
   u_{MPC,i}(t)|\leqslant\epsilon_{i} |\alpha_{BL,i}(t)|$ for every
   $i\in\mathcal{I}$ at every $t\geqslant 0$, one has
   $\lim_{t\rightarrow\infty}\hat u_{MPC}(t)=\zeros_{n}$. Finally,
   since $\hat u^{*}(\mathcal{G},\II{u},\II{\omega}, \hat
   P^{fcst}_{t^{w}},x(t^{w}))$ is the optimal solution of the
   optimization problem~\eqref{opti:nonlinear}, it must satisfy
   constraint~\eqref{opti:sensitivity}, and since
   $\lim_{w\rightarrow\infty}\alpha_{BL}(t^{w})=\lim_{t\rightarrow\infty}\alpha_{BL}(t^{t})=\zeros_{n}$,
   one has $\lim_{w\rightarrow\infty}\hat
   u^{*}(\mathcal{G},\II{u},\II{\omega}, \hat
   P^{fcst}_{t^{w}},x(t^{w}))=\zeros_{n}$. Finally, the convergence of
   $u_{MPC}(t)$ follows from its definition~\eqref{eqn:uMPC}.
% \\
%  to prove the convergence of $u_{MPC}(t)$, by~\eqref{eqn:uMPC}, one has 
% \begin{align}\label{eqn:uMPC-convergence}
%  \lim_{t\rightarrow\infty}u_{MPC}(t) &=\lim_{w\rightarrow\infty} \hat
%   u^{*}(\mathcal{G},\II{u},\II{\omega},  \hat P^{fcst}_{t^{w}},x(t^{w}))\notag
%   \\
%   &=\hat u^{*}(\mathcal{G},\II{u},\II{\omega}, \lim_{w\rightarrow\infty} \hat P^{fcst}_{t^{w}},\lim_{w\rightarrow\infty}x(t^{w}))\notag
%   \\
% &=\hat u^{*}(\mathcal{G},\II{u},\II{\omega}, [p,p,\cdots,p],\zeros_{m+n})\notag
% \\
% &=\zeros_{n},
% \end{align}
% where in the second line we put the limit inside $\hat u^{*}$ since by Proposition~\ref{prop:pwl-optimal},  $\hat u^{*}$ is  piece-wise continuous with respect to  $(\hat P^{fcst}_{t^{w}},x(t^{w}))$. The last line holds as one can easily check that the objective function in the optimization problem~\eqref{opti:nonlinear} is 0 when $(\hat P^{fcst}_{t^{w}},x(t^{w}))=([p,p,\cdots,p],\zeros_{m+n})$ and $\hat u=\zeros_{n}$
 }
\end{proof}

\change{% It is worth mentioning again the independence of stability on
  % the MPC component,
  Since the MPC component cannot jeopardize system closed-loop
  asymptotic stability, cf.~Remark~\ref{rmk:independent-design}, as
  one can see in the proof of
  Theorem~\ref{thm:two-layer-control}\ref{item:convergence}, the
  convergence of $\lambda(t)$, $\omega(t)$, $\alpha_{BL}(t)$,
  $\alpha(t)$, $\alpha_{TL}(t)$, $\alpha_{BL}(t)$, and $\hat
  u_{MPC}(t)$ does not require any a priori assumption on the output
  $u_{MPC}(t)$ of the MPC component. In the simulations, we show that
  even if we perturb $u_{MPC}(t)$ by intentionally shifting its output
  from its true value by a constant, the convergence of the remaining
  signals still holds. On the other hand, the convergence of
  $u_{MPC}(t)$ depends on the convergence of $\alpha_{BL}(t)$.  In
  addition, since both $\hat u_{MPC}(t)$ and $u_{MPC}(t)$ converge to
  $\zeros_{n}$~(and so does their difference), we also conclude that
  the stability filter ultimately lets the MPC component output signal
  pass, i.e., the stability filter preserves the optimality of
  the MPC component in the long run.}

%
% \marginJC{I don't understand the last sentence. Do we really wanna say
%   that? Maybe it is that the system is already getting to where it's
%   going, so MPC prescribes do nothing, and the stability filter is not
%   really bypassing MPC. (what does bypass mean anyway?). I would
%   remove it. }
%   \marginy{The last sentence corresponds to [R3:3]. Here I meant that the stability filter preserves whatever the MPC component output as $t\rightarrow\infty$. This implies that the stability filter doesn't just ignore the MPC's output for stability reason.}
%

\change{One can also verify the independence between stability and the
  MPC component by noting that all stability results
  of~Theorem~\ref{thm:two-layer-control} do not rely on any assumption
  on the forecasted power injection. Although the MPC component is a
  full-state feedback, due to this independence,
  Theorem~\ref{thm:two-layer-control} still holds if the measured
  state is delayed or inaccurate. This means that one could instead
  employ an output feedback controller by designing a state observer
  and feeding the estimated state into the MPC component without
  endangering stability. The minimal set of measured information
  required to realize the controller are: $\omega_{i}$,
  $\alpha_{BL,i}$, $[D^{T}]_{i}\sin(Y_{b}\lambda)$, and $p_{i}$ for
  every $i\in\II{\omega}$. This information is used in the stability
  and low-pass filters, and the top-layer controller. Of course,
  inaccurate state and forecasted power injection lead to non-optimal
  control commands in the MPC component and higher~cost.

  \begin{remark}\longthmtitle{Frequency safety with time-varying power
      injection}\label{rmk:TV-power}
    {\rm % All results in Theorem~\ref{thm:two-layer-control} are
      % established for the dynamics~\eqref{eqn:compact-form} with
      % constant power injection $p$.
      If the power injection is time-varying,
      % denote $p$ as its
      % value at time $t\geqslant 0$ and correspondingly replace $p$
      % in~\eqref{eqn:second-layer-control} by $p(t)$.
      one can see from the proof of
      Theorem~\ref{thm:two-layer-control} that~\ref{item:invariance}
      still holds and~\ref{item:attractivity} partially holds, in the
      sense that the frequency would approach the safe interval but
      may not enter it within a finite time.  \oprocend }
\end{remark}
}

\begin{remark}\longthmtitle{Independence of controller on equilibrium
    point}\label{rmk:independ-eq}
  {\rm It should be pointed out that in
    Theorem~\ref{thm:two-layer-control}, the proposed controller is
    able to locally stabilize the system without a priori knowledge on
    the steady-state voltage angle~$\lambda_{\infty}$. Specifically,
    both $\alpha_{BL}$ and $\alpha_{TL}$ are not functions
    of~$\lambda_{\infty}$.
    % To this aspect, there is no need to know or to compute
    % $\lambda_{\infty}$ for controller design purpose.
  } \oprocend
\end{remark}

% \textcolor{red}{
% \begin{remark}\longthmtitle{Unilateral convergence dependence of MPC component output}\label{rmk:unilaterl-dependence}
% {\rm
% Note that in the proof of Theorem~\ref{thm:two-layer-control}~\ref{item:convergence}, the convergence of $\lambda(t)$, $\omega(t)$, $\alpha_{BL}(t)$, $\alpha(t)$, $\alpha_{TL}(t)$, $\alpha_{BL}(t)$, and $\hat u_{MPC}(t)$ do not reply any a-priori assumption of MPC component output $u_{MPC}(t)$. This consists with our statement in Remark~\ref{rmk:independent-design} that the MPC component cannot affect closed-loop system stability. In the simulation we will show that even if we pollute $u_{MPC}(t)$ by intentionally shifting its from its true value by a constant, the convergence of the remaining signals still holds. On the other hand,  
% the convergence of $u_{MPC}(t)$ depends on the convergence of $\alpha_{BL}(t)$.  In addition, since both $\hat u_{MPC}(t)$ and $u_{MPC}(t)$ converges to $\zeros_{n}$~(so is their difference), we also conclude that the stability filter ultimately bypasses whatever the MPC component output signal is in the long run. 
% }
% \oprocend
% \end{remark}
%}
\begin{remark}\longthmtitle{Control framework without bottom
    layer}\label{rmk:no-bottom-layer}
  {\rm In our previous work~\cite{YZ-JC:19-auto}, we have shown that
    the top-layer controller by itself makes the closed-loop system
    meet all requirements except for the economic cooperation. Such a
    lack of cooperation can be observed in two aspects. First, since
    $\alpha_{TL}$ is only defined for nodes in $\II{\omega}$, those in
    $\II{u}\backslash\II{\omega}$ do not get involved in controlling
    frequency transients. Second, the top-layer control is a
    non-optimization-based state feedback, where each $\alpha_{TL,i}$
    with $i\in\II{\omega}$ is merely in charge of controlling the
    transient frequency for its own node~$i$.
    \oprocend  }  
\end{remark}

\section{Controller
  decentralization}\label{section:distributed-control}

The centralized bilayered controller meets the
requirements~(i)-(iv) stated in Section~\ref{section:ps}. In this
section, we focus on the requirement (v) on the distributed
implementation of the controller.  While introducing each controller
component in Figure~\ref{fig:block-diagram}, our discussion has shown
that only the MPC component requires access to global system
information, whereas all other components can be implemented in a
distributed fashion.  In this section, we show that by having each
node and edge communicate within its 2-hop neighbors, one can solve the
optimization problem $\RR$ in~\eqref{opti:nonlinear} online and hence
exactly recover the MPC component $\hat u^{*}$
in~\eqref{eqn:uMPC}. The key idea is to properly assign the decision
variables in the optimization problem to each node so that the cost
function can be represented as sum of local costs and the constraints
can be written locally. Once this is in place, we report to
saddle-point dynamics to find the solution of~$\RR$ in a distributed
way.

\subsection{Strong convexification of the objective function}

We start here by transforming the optimization problem~$\RR$ into an
equivalent form whose objective function is strongly convex in all its
arguments. Such property is useful later when characterizing the
convergence properties of distributed algorithm to the optimizer.
Formally, let
\begin{align}\label{eqn:strong-convex-cost}
  g^{\text{aug}}(\hat X,\hat u,S)&\triangleq
  \sum_{k=1}^{N-1}\|F\hat x(k+1)-A\hat x(k) -B_{1}\hat p^{fcst}(k)-B_{2}\hat
  u\|_{2}^{2} \notag
  \\
  & \quad +\sum_{k=1}^{N}\left( \sum_{i\in\mathcal{I}}c_{i}\hat
    \alpha_{BL,i}^{2}(k)+\sum_{i\in\II{\omega}}d_{i}s_{i}^{2}(k)
  \right)
  \\
  & \quad + \|\hat x(0)-x(t^{w})\|_{2}^{2}. \notag
\end{align}
We denote by $\RR^{\text{aug}}$ the optimization problem with
objective function $g^{\text{aug}}$ and constraints given
by~\eqref{opti:nonlinear-1}-\eqref{opti:sensitivity}.  Letting
$Y\triangleq(\hat X,\hat
u,S)\in\real^{(m+2n+|\II{\omega}|)N+n}$, we can re-write
$\RR^{\text{aug}}$ into the following  compact form
\begin{subequations}\label{opti:nonlinear-qp}
  \begin{alignat}{2}
    & & & \min_{Y}\frac{1}{2}Y^{T}HY+f^{T}Y+a\notag
    \\
    &\text{s.t.}&\quad & R_{1} Y\leqslant
    r_{1}\label{opti:nonlinear-1-qp},
    \\
    &&& R_{2}Y= r_{2}\label{opti:nonlinear-2-qp},
  \end{alignat}
\end{subequations}
for suitable
\begin{align*}
  H & \in\real^{((m+2n+|\II{\omega}|)N+n)\times
    ((m+2n+|\II{\omega}|)N+n)},
  \\
  f & \in\real^{(m+2n+|\II{\omega}|)N+n} , \qquad a \in\real ,
  \\
R_{1}  & \in\real^{(2|\II{\omega}|N+2|\II{u}|)\times
    ((m+2n+|\II{\omega}|)N+n)},
  \\
  R_{2} & \in \real^{((m+2n)N+n-|\II{u}|)\times
    ((m+2n+|\II{\omega}|)N+n)},
  \\
  r_{1} & \in\real^{2|\II{\omega}|N+2|\II{u}|} , \qquad r_{2}
  \in\real^{(m+2n)N+n-|\II{u}|} .
\end{align*}
%
%\marginJC{Which ones are these matrices and vectors?}
%

The next result shows the equivalence between $\RR$ and
$\RR^{\text{aug}}$.

\begin{lemma}\longthmtitle{Equivalent transformation to strong
    convexity}\label{lemma:transform-strong-convex}
  The optimization problem $\RR$ and~$\RR^{\text{aug}}$ posses exactly
  the same optimal solution. Furthermore, if $F$ is invertible, then
  $g^{\text{aug}}$ is strongly convex in $(\hat X,\hat
  u,S)$.
\end{lemma}
\begin{proof}
  The equivalence between $\RR$ and~$\RR^{\text{aug}}$ follows by
  noting that $g^{\text{aug}}$ corresponds to augmenting $g$ with
  equality constraints. For notational simplicity, we assume that
  $c_{i}=1$ for all $i\in\mathcal{I}$ and $d_{i}=1$ for all
  $i\in\II\omega$ (the proof holds for general positive values with
  minor modifications).  To show strong convexity, one can write $H$
  as an upper-triangular block matrix, whose diagonal matrices are
  $F^{T}F+J^{T}J$, $F^{T}F+A^{T}A+J^{T}J$, $A^{T}A+J^{T}J$,
  $B_{2}^{T}B_{2}$, and $I_{|\II{\omega}|N}$, where
  $J\in\real^{(m+2n)\times n}$ is a matrix mapping the whole state
  $\hat x$ to the partial state $\hat \alpha_{BL}$, i.e., $\hat
  \alpha_{BL}=J\hat x$. It is easy to see that both $J$ and $B_{2}$
  are full-column-rank matrices, which, together with the
  invertibility assumption on $F$, implies that all five matrices are
  positive definite.  Hence, all eigenvalues of $H$ are real and
  strictly positive, leading to strong convexity of~$g^{aug}$, as
  claimed.
\end{proof}

\subsection{Separable objective with locally expressible constraints}\label{subsection:locality}

Next, we explain how the problem data defining the
optimization~$\RR^{\text{aug}}$ has a structure that makes it amenable
to distributed algorithmic solutions. We start by assigning the
decision variables $Y=(\hat X,\hat u,S)$ in
$\RR^{\text{aug}}$ to the nodes and edges in the network.  We
partition the states into voltage angle difference, frequency, and
low-pass filter state, i.e., $\hat x=(\hat \lambda,\hat \omega,\hat
\alpha_{BL})$. For every $k\in[0,N]_{\naturals}$,
$i\in[1,n]_{\naturals}$, and $j\in[1,m]_{\naturals}$, we assign
$\omega_{i}(k)$, $\hat u_{i}$, and $\hat \alpha_{BL,i}(k)$ to the
$i$th node, and $\hat\lambda_{j}(k)$ to the $j$th edge. For every
$i\in\II{\omega}$, we assign $s_{i}(k)$ to the $i$th node. In the
subsequent discussion, we say a constraint or function is \emph{local}
for the power network $\GG$ if its decision variables are all from
either of the following two cases: a) a node $i\in\mathcal{I}$ and its
neighboring edges $(i,j)\in\mathcal{E}$, and b) an edge
$(i,j)\in\mathcal{E}$ and its neighboring nodes $i$ and $j$.  We claim
that
\begin{enumerate}
\item\label{item:sparse-constraints} if $F$, $A$, $B_{1}$ and $B_{2}$ are
  determined by~\eqref{sube:discretization-linear}, then every
  constraint in~\eqref{opti:nonlinear} is local.
\item\label{item:sparse-cost} the objective function $g^{\text{aug}}$
  can be written as a sum of local objective functions.
\end{enumerate}

To see~\ref{item:sparse-constraints}, note
that~\eqref{opti:control-location}-\eqref{opti:sensitivity} are a
collection of constraints, each depending only on variables owned by a
single node. Constraint~\eqref{opti:nonlinear-1} is also local by
noticing the following two points. First, the dynamics of each state
in~\eqref{eqn:approxiate-dynamics} is uniquely determined by the
states of its neighbors.  Second, we have shown in
Section~\ref{section:discretization} that the backward Euler
discretization~\eqref{sube:discretization-linear} preserves locality.
% Specifically, each entry from the first $m$ entries of
% $F\hat x(k+1)-A\hat x(k) -B(\hat p^{fcst}(k)+\hat u)$
% involves some $\hat\lambda_{ij}$, and its neighboring buses
% $\hat\omega_{i}$, and $\hat \omega_{j}$, i.e., variables belonging to
% an edge and to its neighboring buses; each entry from the middle $n$
% entries involves some $\hat\omega_{i}$, $\hat\alpha_{BL,i}$, and
% $\hat\lambda_{ij}$ with $j\in\mathcal{N}(i)$, i.e., variables
% belonging to a bus and to its neighboring edges; each entry in the last
% $n$ entries involves some $\hat\alpha_{BL,i}$ and $\hat u_{i}$, which are
% decision variables belonging to a single bus.
%
% \marginJC{I don't think the reader can follow this discussion about
% middle rows and last rows, it is quite hand-wavy} \marginy{Now I
% think the discussion is trivial after defining the
% locality, %so I just commented them out.}
%
To see~\ref{item:sparse-cost}, first note that the sum of
$\alpha_{BL,i}^{2}(k)$ (respectively, $s_{i}^{2}(k)$) over $i$ is
naturally the sum of local variables. Second, the two-norm square of
$F\hat x(k+1)-A\hat x(k) -B_{1}\hat p^{fcst}(k)-B_{2}\hat u$ for every
$k\in[1,N-1]_{\naturals}$ is the sum of square of all its $m+2n$
entries, where each entry is local due to the locality of discretized
dynamics.  Similarly, $\|\hat x(0)-x(t^{w})\|_{2}^{2}$ is also
the sum of local variables.

\subsection{Distributed implementation via saddle-point dynamics}

Here we introduce a saddle-point dynamics to recover the unique
optimal solution $Y^{*}$ of $\RR^{\text{aug}}$ in a distributed
fashion. We start from the Lagrangian of~$\RR^{\text{aug}}$
\begin{align}\label{eqn:lag}
  \mathfrak{L}(Y,\eta,\mu) =
  g^{\text{aug}}(Y)+\eta^{T}(R_{1}Y-r_{1})+\mu^{T}(R_{2}Y-r_{2}),
\end{align}
where $\eta\in\real^{2|\II{\omega}|N+2|\II{u}|}_{\geqslant 0}$ and
$\mu\in\real^{(m+2n)N+n-|\II{u}|}$ are the Lagrangian multiplier
corresponding to constraints~\eqref{opti:nonlinear-1-qp}
and~\eqref{opti:nonlinear-2-qp}, respectively. Note that we have shown
that a) $\RR$ is feasible (cf.  Proposition~\ref{prop:pwl-optimal}),
b) $\RR$ and $\RR^{\text{aug}}$ are equivalent (cf.
Lemma~\ref{lemma:transform-strong-convex}), and c) all constraints in
$\RR^{\text{aug}}$ are linear. These three points together imply that
the refined Slater condition and strong duality
hold,~\cite[Section~5.2.3]{SB-LV:04}, which further implies that at
least one primal-dual solution $(Y^{*},\eta^{*},\mu^{*})$ of
$\RR^{\text{aug}}$ exists, and the set of primal-dual solutions is
exactly the set of saddle points of $\mathfrak{L}$ on the set
$\real^{(m+2n+|\II{\omega}|)N+n}\times(\real^{2|\II{\omega}|N+2|\II{u}|}_{\geqslant
  0}\times
\real^{(m+2n)N+n-|\II{u}|})$~\cite[Section~5.4.2]{SB-LV:04}.
Therefore, one can apply the saddle-point
dynamics~\cite{AC-EM-SHL-JC:18-tac} to recover one solution
$(Y^{*},\eta^{*},\mu^{*})$, where $\hat u^{*}$ is the MPC output
signal we need. Formally, the saddle-point dynamics of
$\RR^{\text{aug}}$ is
\begin{subequations}\label{sube:eqn:saddle-points}
  \begin{align}
    \epsilon_{Z}\frac{dZ}{d\tau}
    &=-\nabla_{Z}\mathfrak{L}(Z,\eta,\mu)=-(HZ+f+R_{1}^{T}\eta +
    R_{2}^{T}\mu),\label{eqn:sp-1}
    \\
    \epsilon_{\eta}\frac{d\eta}{d\tau}
    &=[\nabla_{\eta}\mathfrak{L}(Z,\eta,\mu)]^{+}_{\eta} =
    [R_{1}Z-r_{1}]^{+}_{\eta},\label{eqn:sp-2}
    \\
    \epsilon_{\mu}\frac{d\mu}{d\tau}
    &=\nabla_{\mu}\mathfrak{L}(Z,\eta,\mu)=R_{2}Z-r_{2},\label{eqn:sp-3}
  \end{align}
\end{subequations}
where $\epsilon_{Z}$, $\epsilon_{\eta}$, and $\epsilon_{\mu}$ are
tunable positive scalars.

Given the strong convexity of $g^{\text{aug}}$, the following result
states the global convergence of the
dynamics~\eqref{sube:eqn:saddle-points}, and its proof directly
follows from~\cite[Theorem~4.2]{AC-EM-SHL-JC:18-tac}.

\begin{theorem}\longthmtitle{Global asymptotic convergence of
    saddle-point dynamics}\label{thm:convergence-saddle-point}
  Starting from any initial condition $(Z(0),\eta(0)\mu(0))$, it holds
  that $Z(\tau)$ globally asymptotically converges to the unique
  optimal solution $Y^{*}$ of $\RR^{\text{aug}}$.
\end{theorem}

To conclude, we justify how the saddle-point
dynamics~\eqref{sube:eqn:saddle-points} can be implemented in a
distributed fashion to recover $Y^{*}$. We first assign $(Z,\eta,\mu)$
to different nodes and edges.  In~\eqref{sube:eqn:saddle-points}, the
primal variable $Z$ corresponds to $Y$, and its assignment is exactly
the same, as discussed at the beginning of
Section~\ref{subsection:locality}. Since all constraints are local
with respect to a node or an edge, we assign each entry of
$(\eta,\mu)$ to the corresponding node or edge.
% Since each entry in $\eta$ corresponds to an inequality constraint
% in either~\eqref{opti:nonlinear-3} or~\eqref{opti:sensitivity}, it
% belongs to the bus where the constraint is on. The other dual
% variable $\mu$ corresponds to the equality
% constraints~\eqref{opti:nonlinear-1}
% and~\eqref{opti:control-location}. Since\eqref{opti:control-location}
% can be treated as $n-|\II{u}|$ separate scalar constraints on each
% bus in $\mathcal{F}\backslash\II{u}$, there are also $n-|\II{u}|$
% entries of $\mu$ where each of them belongs to a bus in
% $\mathcal{F}\backslash\II{u}$. The remaining $(m+2n)N$ entries of
% $\mu$ correspond to constraint~\eqref{opti:nonlinear-1}, Each entry
% in the dual variable $\eta$ belongs to the buses in either
% $\II{\omega}$ or $\II{u}$, corresponding to
% constraints~\eqref{opti:nonlinear-3}
% and~\eqref{opti:sensitivity}. For the other dual variable $\mu$,
% each of its $n-|\II{u}|$ items belongs to buses in
% $\mathcal{F}\backslash\II{u}$, corresponding to the equality
% constraint~\eqref{opti:control-location}, and its remaining
% $(m+2n)N$ items are associated with decision variable $Y$,
% corresponding to~\eqref{opti:nonlinear-1}.
%
% \marginJC{Again, this whole paragraph above needs to be revised for clarity. I
%   don't now what ``item'' stands for. When discussing them, there
%   seems to be a confusion between variables and contraints (i don't
%   understand what we mean by ``corresponding to...''. Also, ``Alone
%   this assignment'' makes no sense in English}
%
With this assignment, and due to locality, the dual variables
dynamics~\eqref{eqn:sp-2} and~\eqref{eqn:sp-3} are distributed, i.e.,
for each entry of $\eta$ or $\mu$, if it belongs to a node (resp.,
edge), then its time derivative only depends on primal and dual
variables of its own and of neighboring edges (resp., nodes). On the
other hand, the primal dynamics~\eqref{eqn:sp-1} requires 2-hop
communication, i.e., for each entry of $Z$, if it belongs to a node
(resp., edge), then its time derivative depends on primal and dual
variables of its neighboring nodes (resp., edges).

\change{Note that here we do not distinguish between communication
  network topology and the underlying physical network topology, in
  that they are identical.  That is to say, each node or edge needs to
  communicate with its neighboring nodes and edges exactly determined
  by the given physical network. In general, any communication
  topology that has the physical topology as a subgraph will also be
  valid, which is a common assumption, see
  e.g.,~\cite{EM-CZ-SL:17,MHN-etal:14,PT-AR:13,PG-MDD-TK-BDS-AR:13}.
  It would still be possible to use an independent communication
  network at the cost of sacrificing performance. For instance,
  in~\cite{XW-SM-BDOA:19}, the trade-off is to have $(m+n)^{2}$ agents
  (as opposed to $(m+n)$ agents here) in total to form the
  communication network; in~\cite{MZ-SM:12}, each agent needs to
  maintain an estimation of the entire optimal solution, leading to
  $O(mN+nN)$ number of estimations in total for each agent, where $N$
  denotes the number of prediction step.  Here, instead, each agent
  only estimates its own component of the optimal solution, which is
  of size $O(N)$.
  % That being said, if our goal were not to solve an optimization
  % problem like~\eqref{opti:nonlinear} or its equivalent
  % version~\eqref{opti:nonlinear-qp} where the physical dynamics is
  % encoded in the constraint~\eqref{opti:nonlinear-1}, the identical
  % topology might be unnecessary. For instance, in the economic
  % dispatch problem~(e.g., (5) in~\cite{AC-JC:15-tcns}, (9)-(10)
  % in~\cite{CZ-UT-NL-SL:14}): the optimal solution can be recovered by
  % distributed communication where the communication network can be
  % arbitrarily designed so long as it is connected.  This independence
  % comes from the irrelevance between constraints in their optimization
  % problem and the physical network dynamics. 
}

\begin{remark}\longthmtitle{Time scale in saddle-point
    dynamics}\label{rmk:time-spd}
  {\rm Since the MPC component updates its output at time instants
    $\{t^{w}\}_{w\in\naturals}$ according to~\eqref{eqn:uMPC}, a
    requirement on the saddle-point
    dynamics~\eqref{sube:eqn:saddle-points} solving $\RR$ (or
    equivalently $\RR^{\text{aug}}$) is that it returns the optimal
    solution within $t^{w+1}-t^{w}$ seconds starting from
    $t^{w}$ for every $w\in\naturals$.  To achieve this, one may
    tune $\epsilon_{Y}$, $\epsilon_{\eta}$, and $\epsilon_{\mu}$ to
    accelerate the convergence of the saddle-point dynamics. In
    practice, this corresponds to
    running~\eqref{sube:eqn:saddle-points} on a faster time scale,
    which puts requirements on the hardware regarding communication
    bandwidth and computation time. 
    % \change{On the other hand, there are
    % algorithms~\cite{MMC-JC-BG:19-cdc,PG-MDD-TK-BDS-AR:13} that
    % expedite the convergence without requiring a high bandwidth, but
    % it is currently unclear if one can implement them in a
    % distributed way to solve our optimization problem containing
    % inequality constraints.}  However,
    % Theorem~\ref{thm:convergence-saddle-point} only justifies
    % asymptotic convergence, for now, we are unable to provide an
    % analytical guideline to determine the above three
    % parameters.
  } \oprocend
\end{remark}

\begin{remark}\longthmtitle{Comparison with controller with regional
    coordination based on network
    decomposition}\label{rmk:comparison-semi} {\rm The proposed
    distributed algorithm treats each bus and transmission line as an
    agent, and recovers the optimal solution by allowing each agent to
    exchange information only with its neighbors. In our previous
    work~\cite{YZ-JC:19-acc,YZ-JC:20-auto}, we have proposed an
    alternative algorithm that does not rely on participation of every
    agent at the expense of not recovering the global optimal
    solution.  The basic idea of this alternative implementation is to
    consider a set of regions in the network. Each region,
    independently of the rest, possesses its own centralized
    controller in charge of gathering regional information and
    broadcasting control signals to controllers within the region.  To
    account for the couplings in the dynamics, flows that connect a
    region and the rest of the network are assumed constant when
    computing the controller in each region.  Although there can be
    nodes and edges shared by multiple regions, the control signal
    regulated on a shared node belongs to only one region.  This
    implementation does not recover the exact optimal solution and
    only ensures partial cooperation among the inputs.}
  \oprocend
 \end{remark}

\section{Numerical examples}\label{sec:simulations}

We verify our results on the IEEE 39-bus power network shown in
Figure~\ref{fig:IEEE39bus}.  We run all simulations in MATLAB~2018b in
a desktop with an i7-8700k CPU@4.77GHz and 16GB DDR4
memory@3600MHz. All parameters in the power network
dynamics~\eqref{eqn:compact-form} come from the Power System
Toolbox~\cite{KWC-JC-GR:09}.  Let
$\II{\omega}=\{30,31,32,37\}$ be four generator buses with transient
frequency requirements.  The safe frequency region is
$[\underline\omega_{i},\bar\omega_{i}]=[-0.2Hz,\ 0.2Hz]$ for every
$i\in\II{\omega}$ (as $\omega$ corresponds to the shifted frequency,
the safe frequency region without shifting is thus $[59.8Hz,\
60.2Hz]$). Let $\{3,7,25\}$ be another three non-generator buses that
can provide control signals, so that $\II{u}=\{3,7,25,30,31,32,37\}$.
To set up the optimization problem~\eqref{opti:nonlinear} used in the
MPC component~\eqref{eqn:uMPC}, we
use~\eqref{sube:discretization-linear} for the discretization. The
controller parameters are summarized in
Table~\ref{table:control-parameter}. In addition, we apply the
saddle-points dynamics~\eqref{sube:eqn:saddle-points} to generate the
output of the MPC component in a distributed fashion.
%  let: $\tilde t=10s$, $ T=0.2$ so that
% $N=50$; $\epsilon_{i}=1.9$ and $\tau_{i}=0.5$ for every $i\in\II{u}$;
% $c_{i}=1$ if $i\in\II{\omega}$, whereas $c_{i}=4$ if
% $i\in\II{u}\backslash\II{\omega}$; $d=100$;
% $p^{fcst}_{t}(\tau)=p(\tau)$ for every $\tau\in[t,t+\tilde t]$; let
% $\{t^{w}\}_{j\in\naturals}=\{j\}_{j\in\naturals}$, i.e., the MPC
% component samples and updates its output  every $1s$; For the
% top-layer controller~\eqref{eqn:second-layer-control}, we choose
% $\bar\gamma_{i}=\underline\gamma_{i}=1$ and
% $\bar\omega_{i}^{\text{thr}}=-\underline\omega_{i}^{\text{thr}}=0.1Hz$
% for every $i\in\II{\omega}$.

\begin{table}
  % \normalsize
  \renewcommand\arraystretch{1.2}
  \centering
    \begin{tabular}{|c|c||c|c|}
      \hline
      parameter & value &       parameter & value
      \\
      \hline
      $\tilde t$ & $10s$ &  $p^{fcst}_{t}(\tau),\ \forall \tau\in[t,t+\tilde t]$ & $p(\tau)$
      \\
      \hline
      $T$ & $0.2s$ & $c_{i},\ \forall i\in\II{\omega}$ & $4$
      \\
      \hline
      $N$ & $50$ & $c_{i},\ \forall i\in\II{u}\backslash\II{\omega}$ & $1$
      \\
      \hline
      $d$ & $100$ &$t^{w},\ \forall w\in\naturals$ & $w$
      \\
      \hline
      $\epsilon_{i},\ \forall i\in\II{u}$ & $1.9$ &$\bar\gamma_{i}$
      and $\underline\gamma_{i},\ \forall i\in\II{\omega}$ & $1$ 
      \\
      \hline
      $\tau_{i},\ \forall i\in\II{u}$ & $0.5s$
      &$\bar\omega_{i}^{\text{thr}}$ and
      $-\underline\omega_{i}^{\text{thr}}$ & $0.1Hz$ 
      \\
      \hline
    \end{tabular}
  \caption{Controller parameters.}\label{table:control-parameter}
  \vspace*{-5ex}
\end{table}

\begin{figure}[htb]
  \centering%
  \includegraphics[width=.85\linewidth]{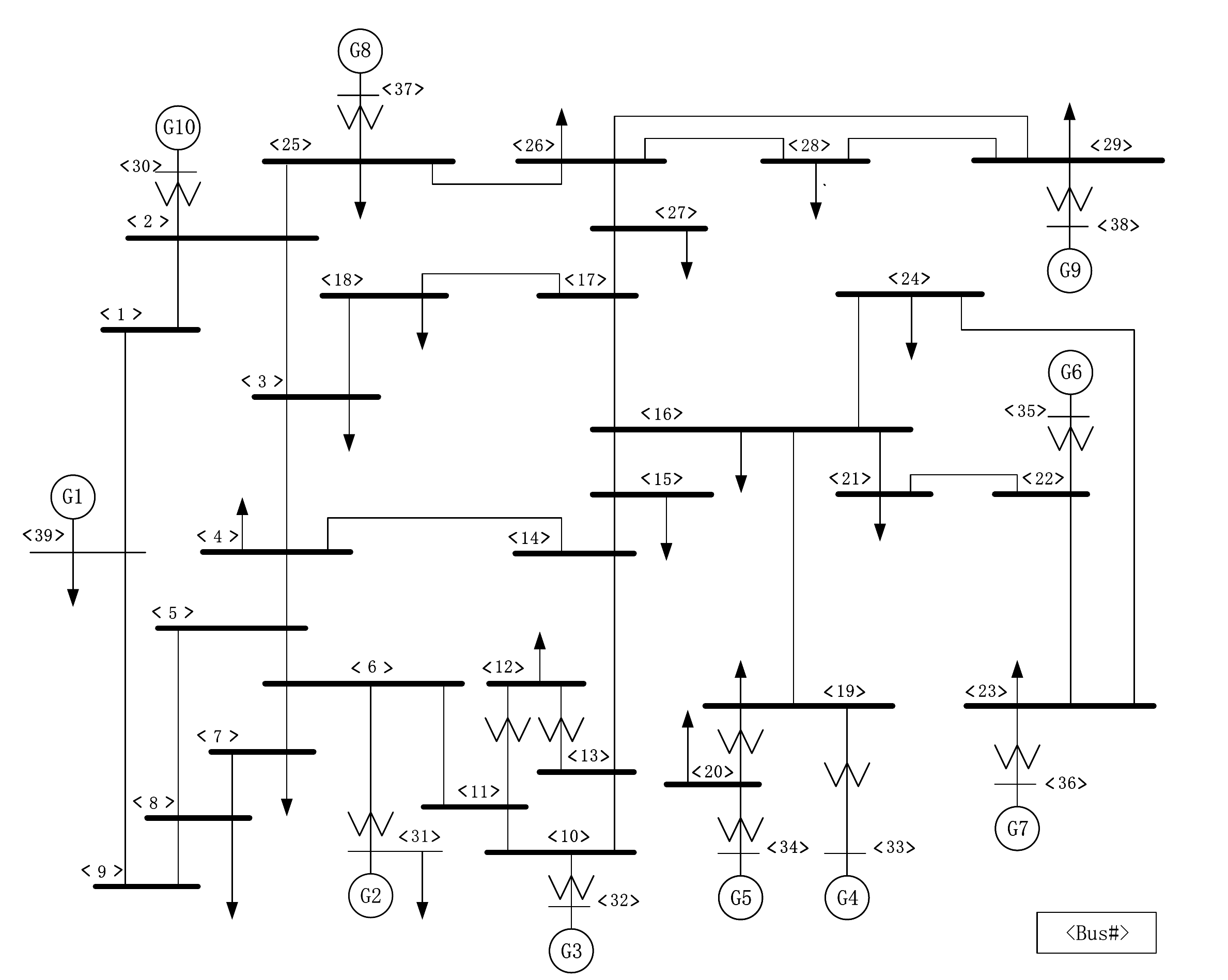}
  \caption{IEEE 39-bus power network.}\label{fig:IEEE39bus}
\end{figure}

We first show that the bilayered controller defined
by~\eqref{eqn:two-layer},~\eqref{eqn:uMPC},~\eqref{eqn:stability-filter},~\eqref{eqn:lp-filter},
\eqref{eqn:second-layer-control} is able to maintain the transient
frequency of selected nodes within the safe region without changing
the equilibrium point
(cf. Theorem~\ref{thm:two-layer-control}\ref{item:invariance}
and~\ref{item:convergence}). Although in the
dynamics~\eqref{eqn:compact-form} we assume that the power injection
is constant, in simulations we perturb all non-generator nodes by a
time-varying power injection. Specifically, for every
$i\in[1,29]_{\naturals}$, let $p_{i}(t)=(1+\delta(t))p_{i}(0)$ where
\begin{align*}
  \delta(t)=
  \begin{cases}
    0.2\sin(\pi t\slash50) & \hspace{0.5cm}\text{if $0\leqslant
      t\leqslant 25$},
    \\
    0.2& \hspace{0.5cm}\text{if $ 25< t\leqslant 125$},
    \\
    0.2\sin(\pi (t-100)\slash50)& \hspace{0.5cm}\text{if $ 125<
      t\leqslant 150$},
    \\
    0& \hspace{0.5cm}\text{if $ 150< t$}.
  \end{cases}
\end{align*}
The deviation $\delta(t) p_{i}(0)$ has both fast ramp-up and ramp-down
periods and a long intermediate constant period. We have chosen it
this way to test the capability of the controller against both
slow-varying and fast-varying disturbances.
Figure~\ref{fig:trajectories}\subref{fig:frequency-response-open-loop}
shows the open-loop frequency responses of nodes 30, 31, 32, and 37
(i.e., nodes with the frequency safety requirement). All four
frequency trajectories, which almost overlap with each other,
%
% \marginJC{Why is it that they all look very similar? Are there
%   inertias, etc., all alike? It'd be nicer, I think, to see more heretogeneity.}
%   
%   \marginy{The homogeneity is mostly due to the fast swing
%   dynamics. This can be seen from the following few points: I) In
%   Fig.3 a), the open-loop frequency response are almost the same. In
%   fact, if I make the disturbance faster, i.e., replacing 50 in
%   $\delta$ by 3, the responses are still almost same. II) The
%   ineritial of 30,31,32,and 37 are indeed almost same (they are all
%   around 8); however if I keep all parameters same but replace the
%   inerial of node 30 by 50, and re-do Fig. 3b), then $\omega_{30}$
%   still overlaps with $\oemga_{31},\omega_{32}...$. I think we can
%   make the repsonse different if we further replace 50 by 500.  III)
%   A way to make the responses not overlap with each other would be
%   choosing different penalty coefficients and control cost
%   coefficients for different nodes, but this would require more
%   figures, which I don't think is really necessary.  }
%
exceed the lower safe frequency bound $59.8Hz$. However, with the
controller enabled, in
Figure~\ref{fig:trajectories}\subref{fig:frequency-response-closed-loop},
their frequencies all evolve within the safe region, and they all
return to $60Hz$ as the disturbance
disappears. Figure~\ref{fig:trajectories}\subref{fig:control-response-region1}
shows the corresponding control signals. Note that, due to our
specific choice of $c_{i}$'s, the controller tends to use more
non-generator control signals (i.e., $\alpha_{3}$, $\alpha_{7}$, and
$\alpha_{25}$) than generator ones (i.e., $\alpha_{31}$,
$\alpha_{32}$, $\alpha_{33}$, and $\alpha_{37}$). Also, note that they
split into two groups and the control signals within each group
possess almost the same trajectories.
%  in that they correspond to a
% same value of $c_{i}$.

% 
% Additionally, we solve the quadratic programming
% problem~\eqref{opti:nonlinear} in the MPC component by the built-in
% MATLAB function $mpcqpsolver$~\cite{mpcqpsolver} with default
% set-ups, and it only takes no more than 0.007 second each time to
% update the MPC update.

\begin{figure*}[tbh!]
  \centering
  \subfigure[\label{fig:frequency-response-open-loop}]{\includegraphics[width=.24\linewidth]{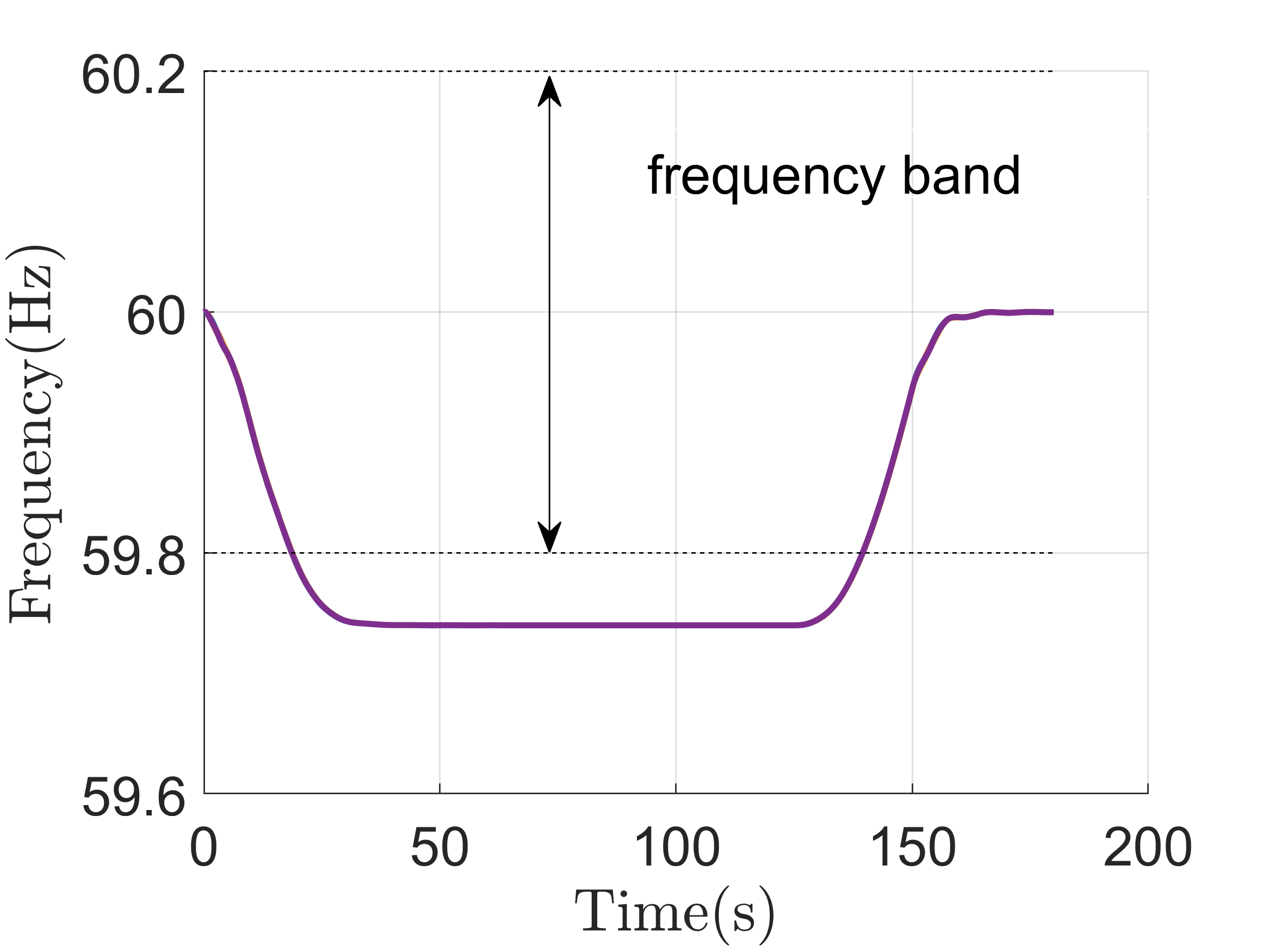}}
  \subfigure[\label{fig:frequency-response-closed-loop}]{\includegraphics[width=.24\linewidth]{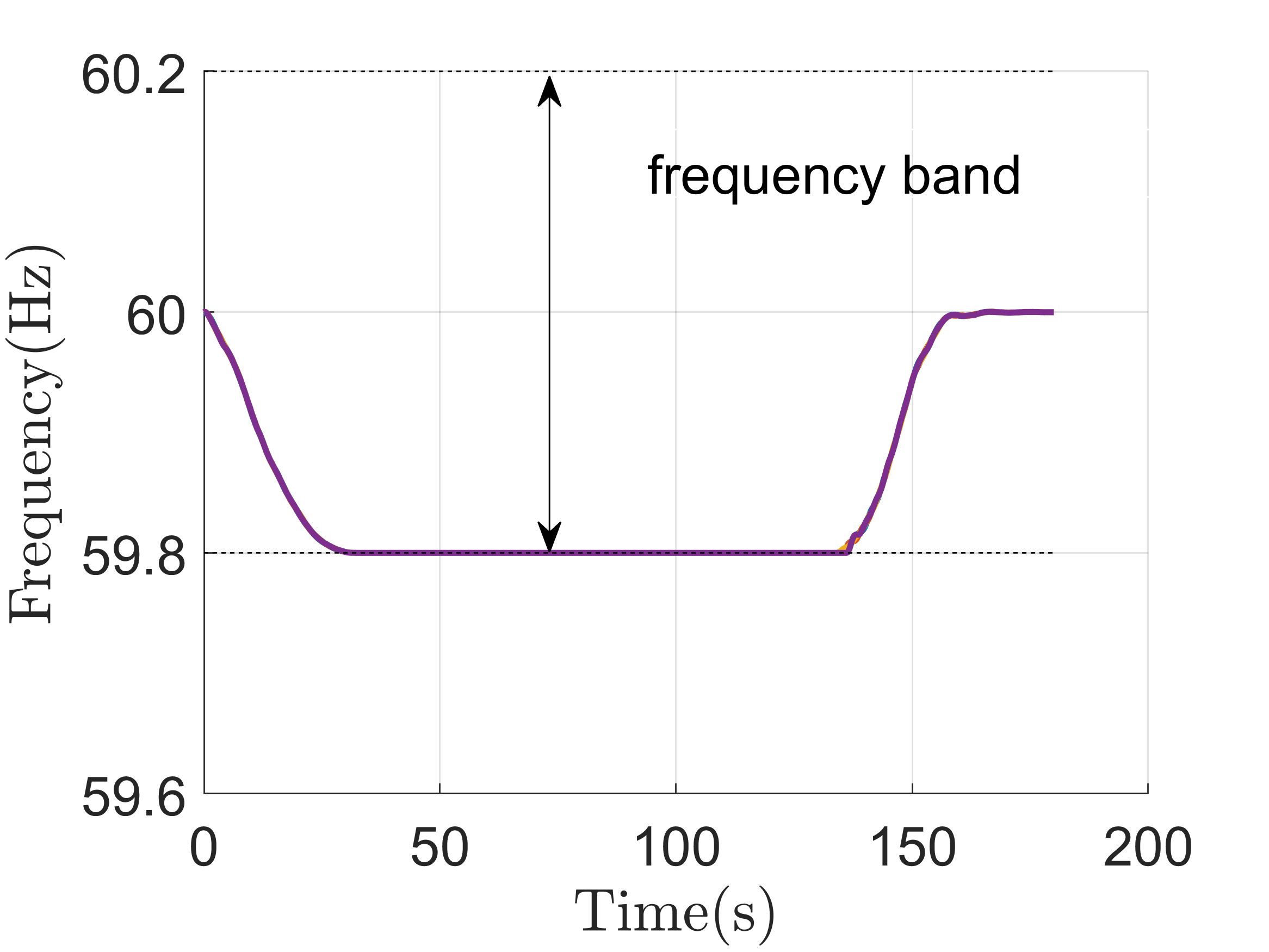}}
  \subfigure[\label{fig:control-response-region1}]{\includegraphics[width=.24\linewidth]{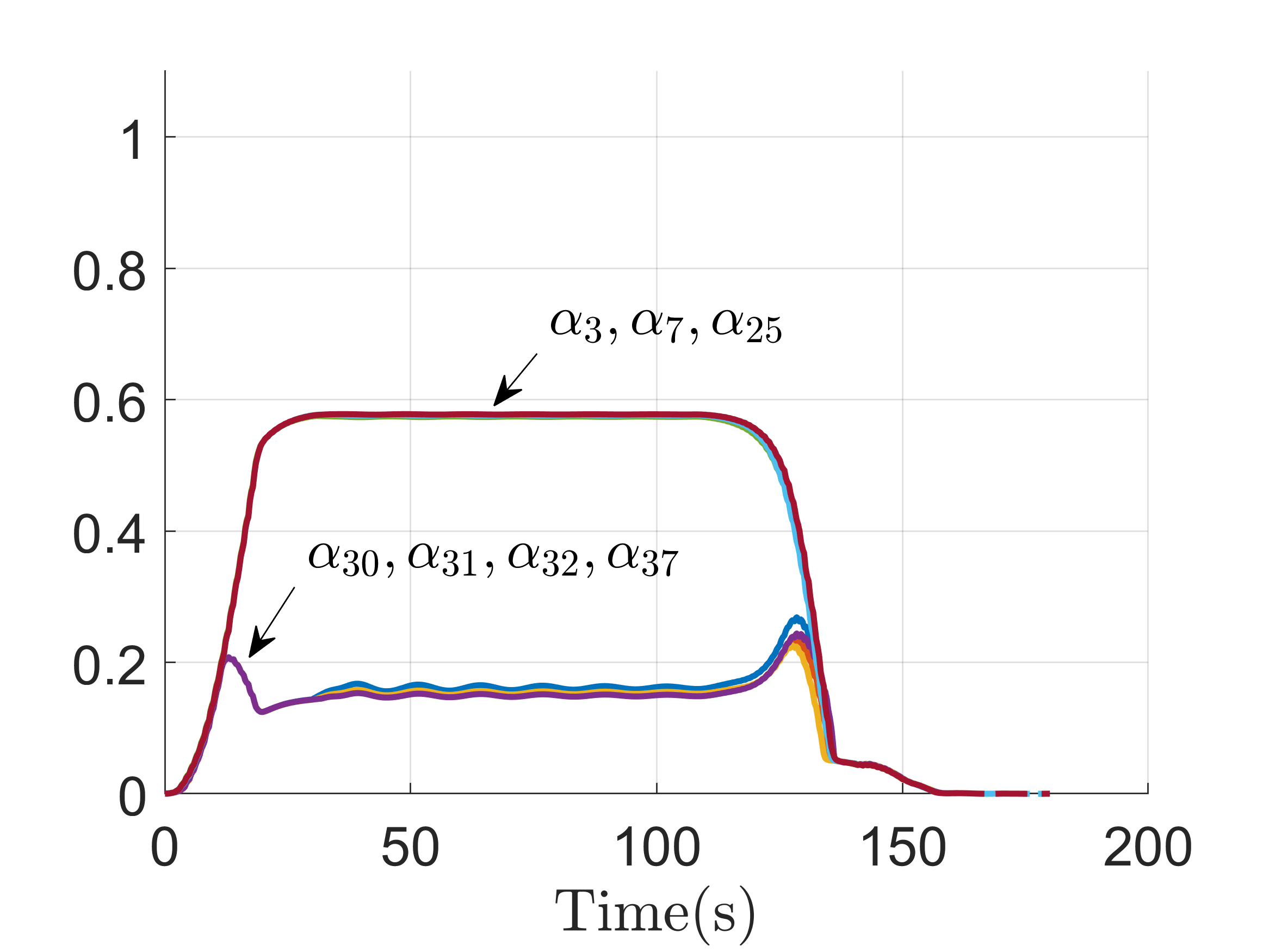}}
  \caption{Frequency and control input trajectories with and without
    transient frequency control.
    Plot~\subref{fig:frequency-response-open-loop} shows the open-loop
    frequency responses at nodes 30, 31, 32, and 37, all exceeding the
    lower safe bound. The closed-loop system with the distributed
    control has all responses stay inside the safe region in
    plot~\subref{fig:frequency-response-closed-loop}. Plot~\subref{fig:control-response-region1}
    shows the corresponding control trajectories.
  }\label{fig:trajectories}
\end{figure*}

Next we compare the performance of the proposed controller with other
approaches. Figures~\ref{fig:trajectories-fully-decen}\subref{fig:frequency-response-closed-loop-fully-decen}
and~\subref{fig:control-response-fully-decen} show the frequency
trajectories and control signals using the controller with regional
coordination based on network decomposition proposed
in~\cite{YZ-JC:19-acc}. As mentioned in
Remark~\ref{rmk:comparison-semi}, although this controller achieves
frequency safety, it only allows control cooperation within a limited
region, instead of the entire network. This can be seen from
Figure~\ref{fig:trajectories-fully-decen}\subref{fig:control-response-fully-decen},
where, with the same control cost coefficients
(cf. Table~\ref{table:control-parameter}), the two groups of control
trajectories are not as uniform as those in
Figure~\ref{fig:trajectories}\subref{fig:control-response-region1} and
have a larger magnitude.
Figures~\ref{fig:trajectories-fully-decen}\subref{fig:frequency-response-pure-df}
and~\subref{fig:control-response-pure-df} are the frequency and
control trajectories with only the top-layer controller, as proposed
in~\cite{YZ-JC:19-auto}, cf. Remark~\ref{rmk:no-bottom-layer}. Since
it is a non-optimization-based control strategy, each control signal
does not cooperate with others. In this specific scenario, the
top-layer controller leads to fluctuations even during the time
interval [25,125]$s$, when the disturbance is constant. This is
because the top-layer controller is myopic, without further
consideration for the effects of the rest of the network. The economic
advantage of the proposed bilayered control can be also seen by
computing the overall control cost over [0,180]$s$
%  by
% $\int_{0}^{180}\sum_{i\in\II{u}}c_{i}\alpha_{i}^{2}(\tau)\text{d}\tau$,
% which is the continuous version of the discrete control cost in the
% objective function of~\eqref{opti:nonlinear}.
of the proposed controller, the controller in~\cite{YZ-JC:19-acc}, and
the controller in~\cite{YZ-JC:19-auto}, which are around 163, 231 and
656, resp.

% 
% Compare it with Figure~\ref{fig:frequency-response-closed-loop} in
% their zoom-in parts around 30 seconds, we see that the frequency
% responses with the full-distributed controller does not have the
% slight over-reaction in the semi-distributed controller. This is
% because that the latter one vaguely overestimates the severity of
% the disturbance (cf. Remark~\ref{rmk:const-virtual-inter}), and thus
% makes some additional control effort, causing the frequency slightly
% higher than the lower safe
% bound. Figure~\ref{fig:trajectories-fully-decen}\subref{fig:control-response-fully-decen}
% and~\subref{fig:control-response-reigonal-all} show the control
% trajectories of fully-distributed and semi-distributed controllers,
% respectively, where the latter one is nothing but merging
% Figure~\ref{fig:control-response-region1} and
% Figure~\ref{fig:control-response-region2} into a single plot. Since
% we only have $c_{i}$ to be either $1$ or $4$, The trajectories both
% Figure~\ref{fig:control-response-fully-decen} and
% Figure~\ref{fig:control-response-reigonal-all} can be roughly
% slitted into two parts, one corresponds to control trajectories with
% value around 0.5, and the other one with 0.2. However, the
% trajectories in Figure~\ref{fig:control-response-fully-decen} turn
% out to be more uniform and smooth than that in
% Figure~\ref{fig:control-response-reigonal-all}, in that the later
% ones disables control cooperation among different regions
% (cf. Remark~\ref{rmk:semi-dencentralized-loss}).

\begin{figure*}[tbh!]
  \centering
  \subfigure[\label{fig:frequency-response-closed-loop-fully-decen}]{\includegraphics[width=.24\linewidth]{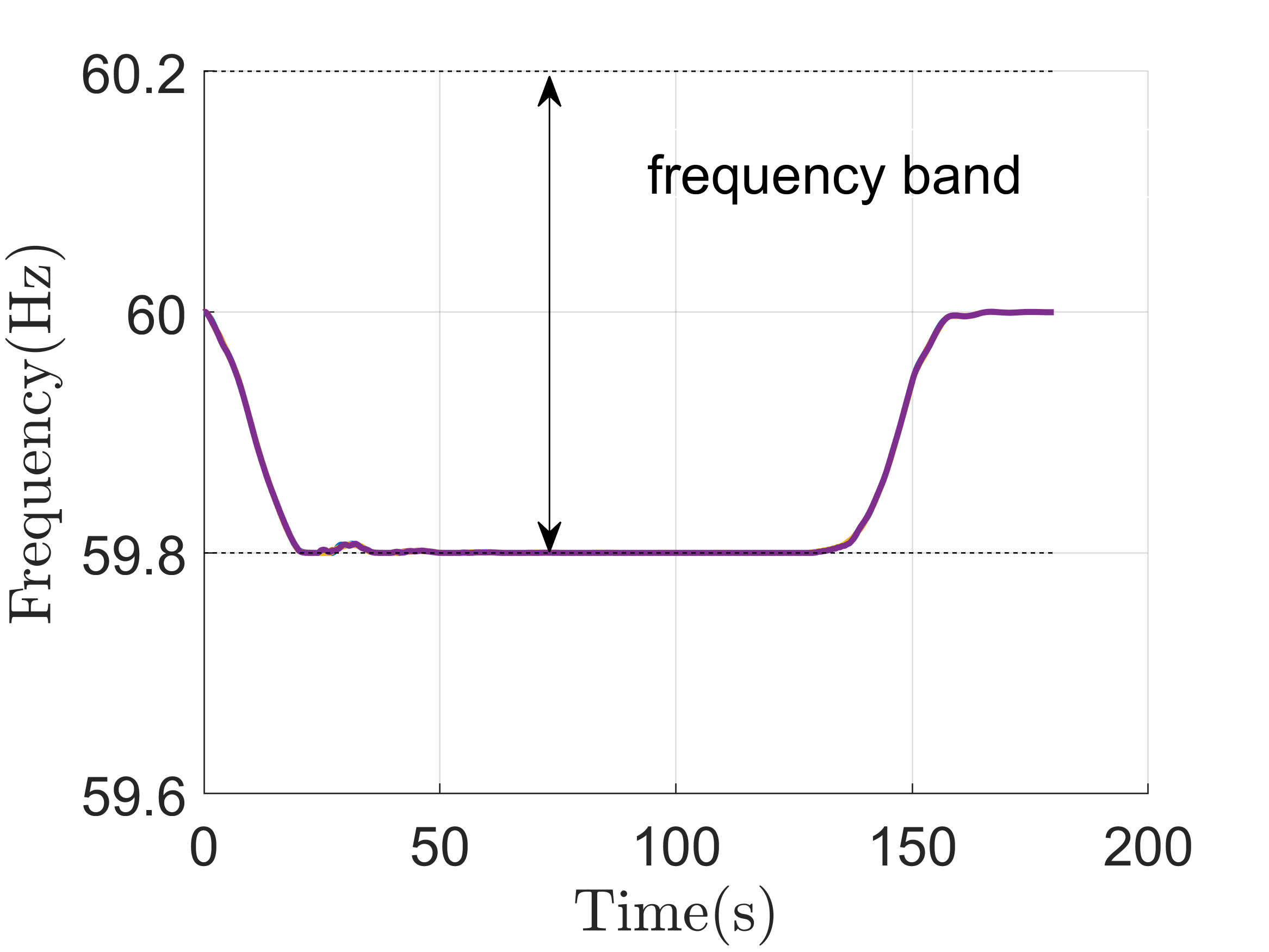}}
  \subfigure[\label{fig:control-response-fully-decen}]{\includegraphics[width=.24\linewidth]{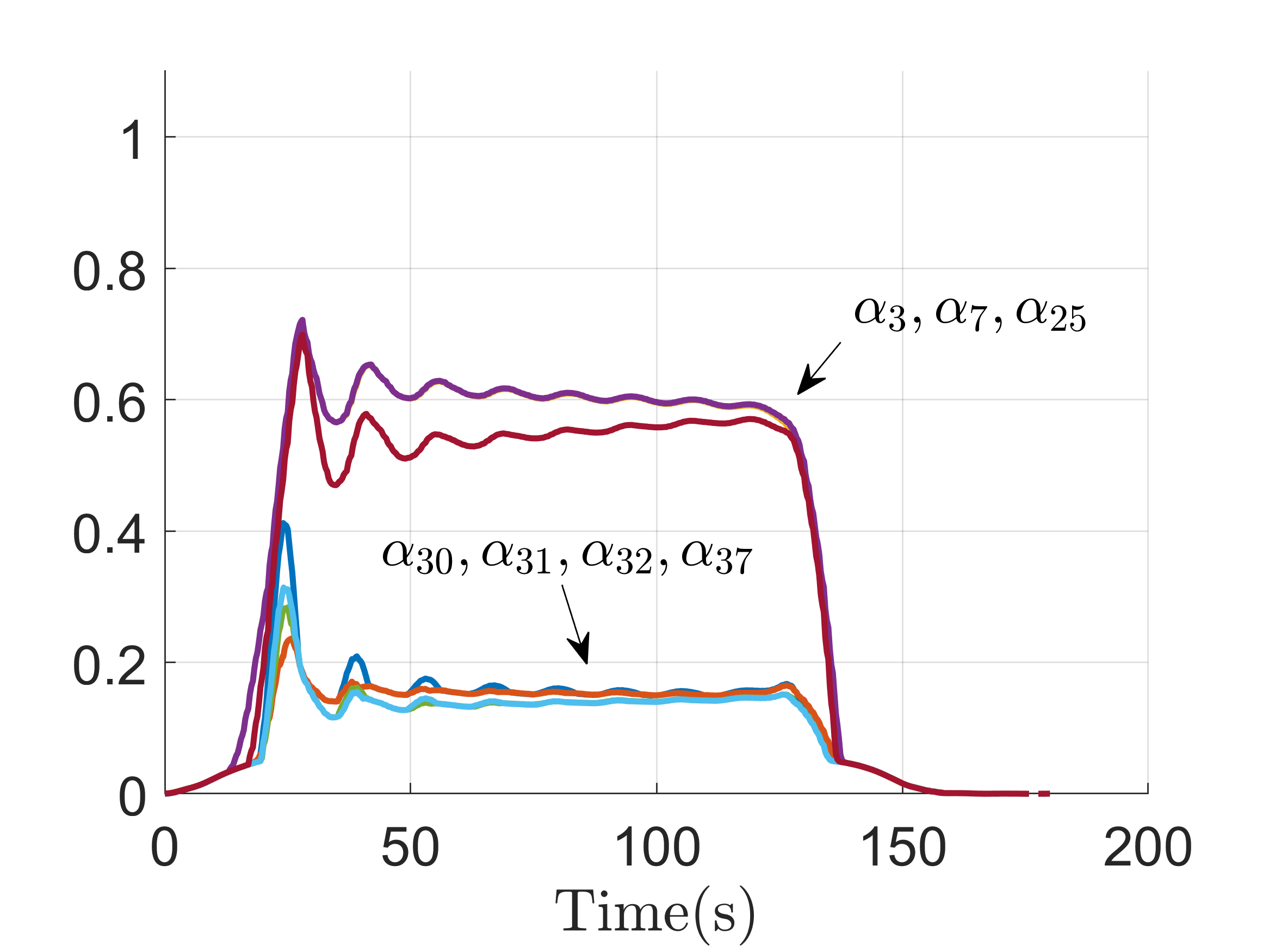}}
  \subfigure[\label{fig:frequency-response-pure-df}]{\includegraphics[width=.24\linewidth]{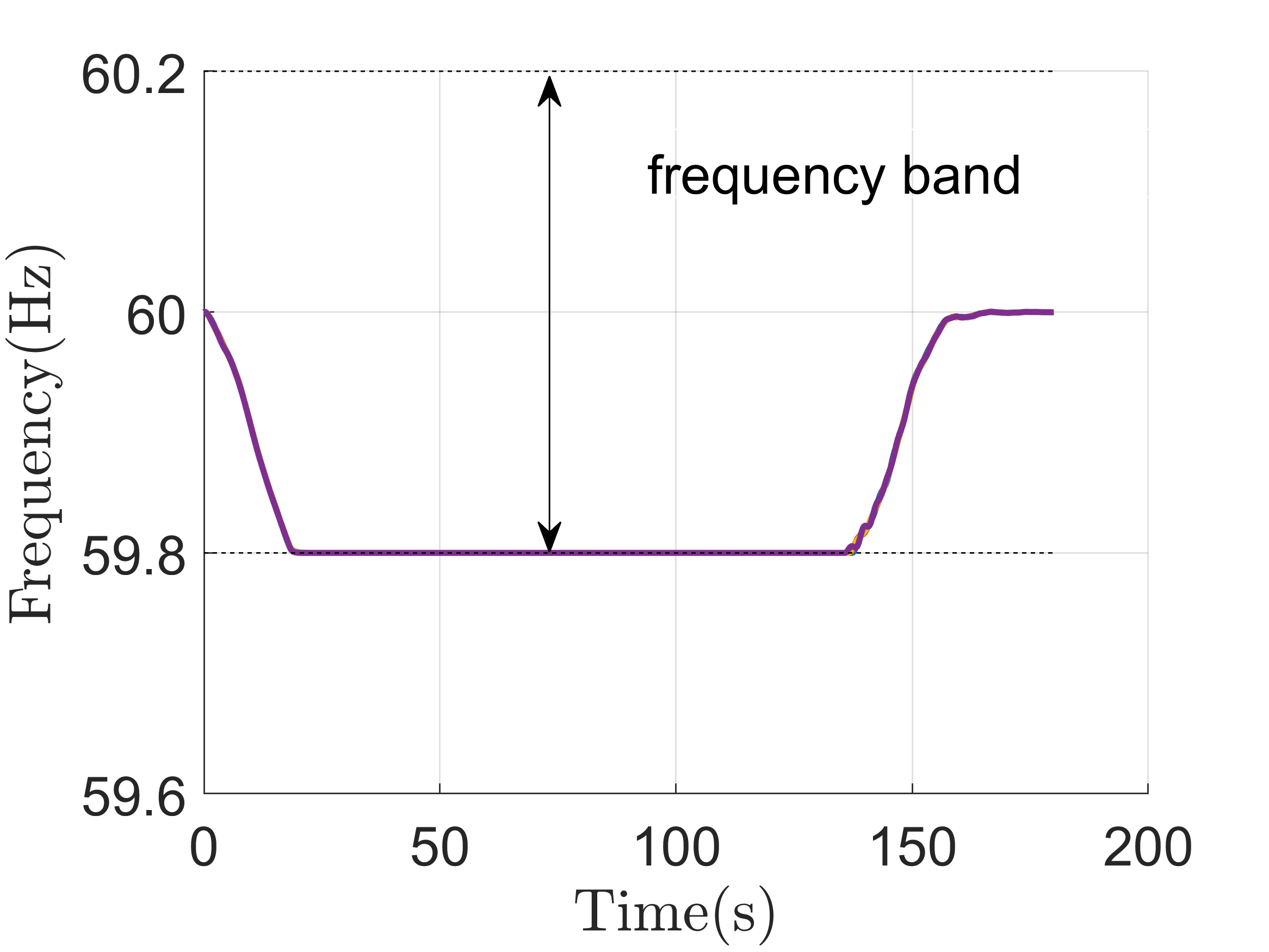}}
  \subfigure[\label{fig:control-response-pure-df}]{\includegraphics[width=.24\linewidth]{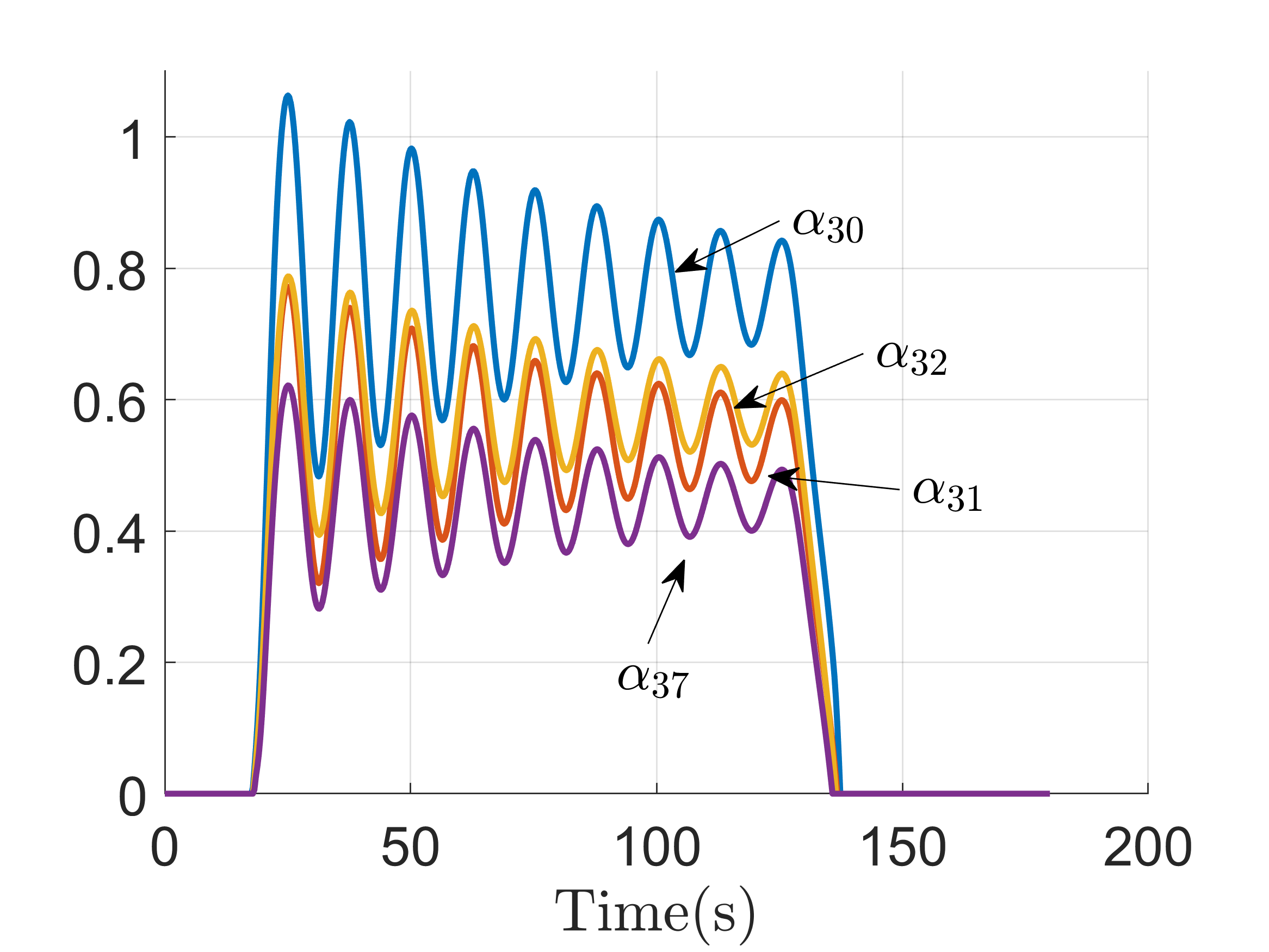}}
  \caption{Comparison of frequency and control trajectories with other
    approaches. Plot~\subref{fig:frequency-response-closed-loop-fully-decen}
    and~\subref{fig:control-response-fully-decen} employ the
    controller with regional coordination based on network
    decomposition proposed in~\cite{YZ-JC:19-acc}.
    Plot~\subref{fig:frequency-response-pure-df}
    and~\subref{fig:control-response-pure-df} correspond to the
    top-layer controller, a non-optimization-based control strategy
    proposed
    in~\cite{YZ-JC:19-auto}. }\label{fig:trajectories-fully-decen}
\end{figure*}

% 
% Figure~\ref{fig:trajectories-pure-df} shows the frequency and
% control trajectories with the bottom-layer disabled (
% cf. Remark~\ref{rmk:no-bottom-layer}). Note that in
% Figure~\ref{fig:frequency-response-pure-df} all targeted frequencies
% do not exceed the lower bound; however, since the top layer control
% signal is non-optimization-based, we see no cooperation among each
% control signal and observe bounded but relatively large fluctuations
% in the control trajectories, as apposed to the semi- and fully
% distributed control trajectories in Figure~\ref{fig:trajectories}
% and Figure~\ref{fig:trajectories-fully-decen}.

% 
% 
% \begin{figure}[tbh!]
%   \centering
%   \subfigure[\label{fig:frequency-response-pure-df}]{\includegraphics[width=.48\linewidth]{epsfiles/IEEE39-frequency-response-with-control-generator-pure-df.png}}
%     \subfigure[\label{fig:control-response-pure-df}]{\includegraphics[width=.48\linewidth]{epsfiles/IEEE39-control-response-no-bottom-layer.png}}
%     \caption{Frequency and control input trajectories with controller with no-bottom layer.
%       Plot~\subref{fig:frequency-response-pure-df} shows
%       the open-loop frequency responses of node 30,31,32, and  37. Although all of them are within in the safe region, due to the non-participation of the bottom layer, the control signals shown in plot~\subref{fig:control-response-pure-df}
% tends to behave in a selfish fashion, which further causes oscillation.
%        }\label{fig:trajectories-pure-df}
% \vspace*{-1.5ex}
% \end{figure}

Next, we examine the role of the bottom and top layers in determining
the value of the input signal of our distributed controller.  For node
$30$, Figure~\ref{fig:trajectories-30}\subref{fig:control-response-30}
shows that $\alpha_{BL,30}$ is responsible for the larger share in the
overall control signal $\alpha_{30}$, whereas $\alpha_{TL,30}$
provides a slightly tuning during most of the time.  If we reduce the
penalty $d_{30}$ from 100 to 10, in
Figure~\ref{fig:trajectories-30}\subref{fig:control-response-30-smaller-penalty},
the dominance of $\alpha_{BL,30}$ decreases, in accordance with our
discussion in Remark~\ref{rmk:violation-penalty}. On the contrary, if
we raise $d_{30}$ to 1000, the contribution of the top layer becomes
much smaller, as shown in
Figure~\ref{fig:trajectories-30}\subref{fig:control-response-30-large-penalty}.

\change{We further look into the bottom-layer control signals at node
  30. Using the same set-up as in
  Figure~\ref{fig:trajectories-30}\subref{fig:control-response-30}, we
  plot in
  Figure~\ref{fig:trajectories-30-detailed}\subref{fig:control-response-30-BL-normal}
  the MPC component output signal $u_{MPC,30}$ and the stability
  filter output signal $\hat u_{MPC,30}$. They are almost identical
  except for a paltry difference around 140s. Next, in
  Figure~\ref{fig:trajectories-30-detailed}\subref{fig:control-response-30-BL-shifted},
  we purposefully add 0.1 to $u_{MPC,30}$, i.e., the input of the
  stability filter is now re-defined as $u_{MPC,30}+0.1$. Notice that
  $\hat u_{MPC,30}$, unaffected by the input shift, still converges to
  0, which coincides with our analysis after
  Theorem~\ref{thm:two-layer-control}.}
Figure~\ref{fig:trajectories-30-detailed}\subref{fig:saddle-points}
shows how the saddle-point dynamics~\eqref{sube:eqn:saddle-points}
converges to the value of $u_{MPC,30}(50)$ starting from an initial
guess. Here we have used $\epsilon_{Z}=5\cdot 10^{-4}$ and
$\epsilon_{\eta}=\epsilon_{\mu}=2.5\cdot 10^{-4}$ to ensure
convergence is attained within 1$s$,
cf. Table~\ref{table:control-parameter}.

\change{To illustrate the closed-loop system performance under
  uncertainty, in Figure~\ref{fig:trajectories-uncertainty} we
  simulate three different scenarios. In
  Figure~\ref{fig:frequency-inaccurate-power}, instead of having an
  accurate forecasted power injection, at every $t\geqslant 0$, we let
  $p^{fcst}_{t}(\tau)= p(t)$ for all $\tau\in[t,t+\tilde t]$, i.e.,
  the forecasted power injection is simply the current power
  injection. Note that in this case the frequencies of all four
  controlled nodes stay within the safe region,
  cf. Remark~\ref{rmk:TV-power}; in
  Figure~\ref{fig:frequency-generator-dynamics}, for each generator
  node (i.e., node 30 to 39), we adopt a first-order
  model~\cite{ZW-FL-SHL-CZ-SM:18} with a time constant of $5s$ as the
  generator dynamics, and note that the frequencies still stay within
  the safe region most of the time; in
  Figure~\ref{fig:frequency-generator-dynamics-inaccurate-power}, we
  consider both inaccurate forecasted power injection and the
  generator dynamics, and the frequencies still behave well after a
  short fluctuation.}
%
% \marginJC{These are crazily small numbers and given without context
%   here. Numbers like 0.1, maybe 0.05 would be more easy to accept. In
%   terms of number of communication rounds/variables, etc., how
%   demanding are these numbers? Reviewers are gonna want to have more
%   detailed information, and also some discussion about pros/cons, and
%   ways to address them.}
% 
% \marginy{Since the saddle-points dynamics is continous, theoretically,
%   it doesn't really matter how small these numbers are. Indeed, in
%   practise, we have to disrectize the saddle-points dynamics, in
%   smaller $\epsilon$s, requires more frequent the communications and
%   the higher bandwidth.}
% %
% \marginJC{Yifu, I doubt the reviewers would be satisfied with your
%   response in the margin. It does matter, becase the numbers directly
%   correspond to implementation costs. This part is weak and we need to
% strengthen it.}
%

\begin{figure*}[tbh!]
  \centering
  \subfigure[\label{fig:control-response-30}]{\includegraphics[width=.24\linewidth]{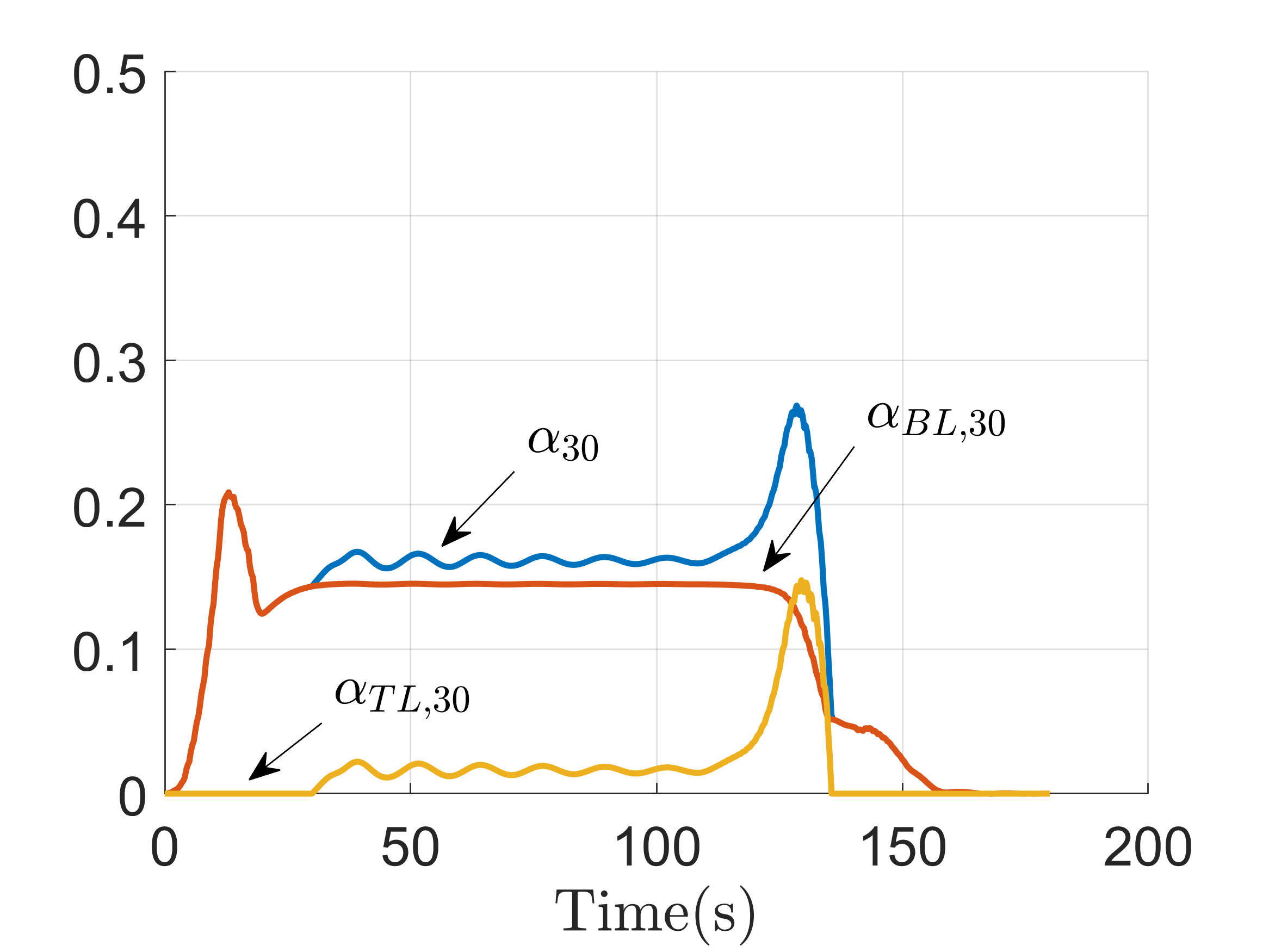}}
  \subfigure[\label{fig:control-response-30-smaller-penalty}]{\includegraphics[width=.24\linewidth]{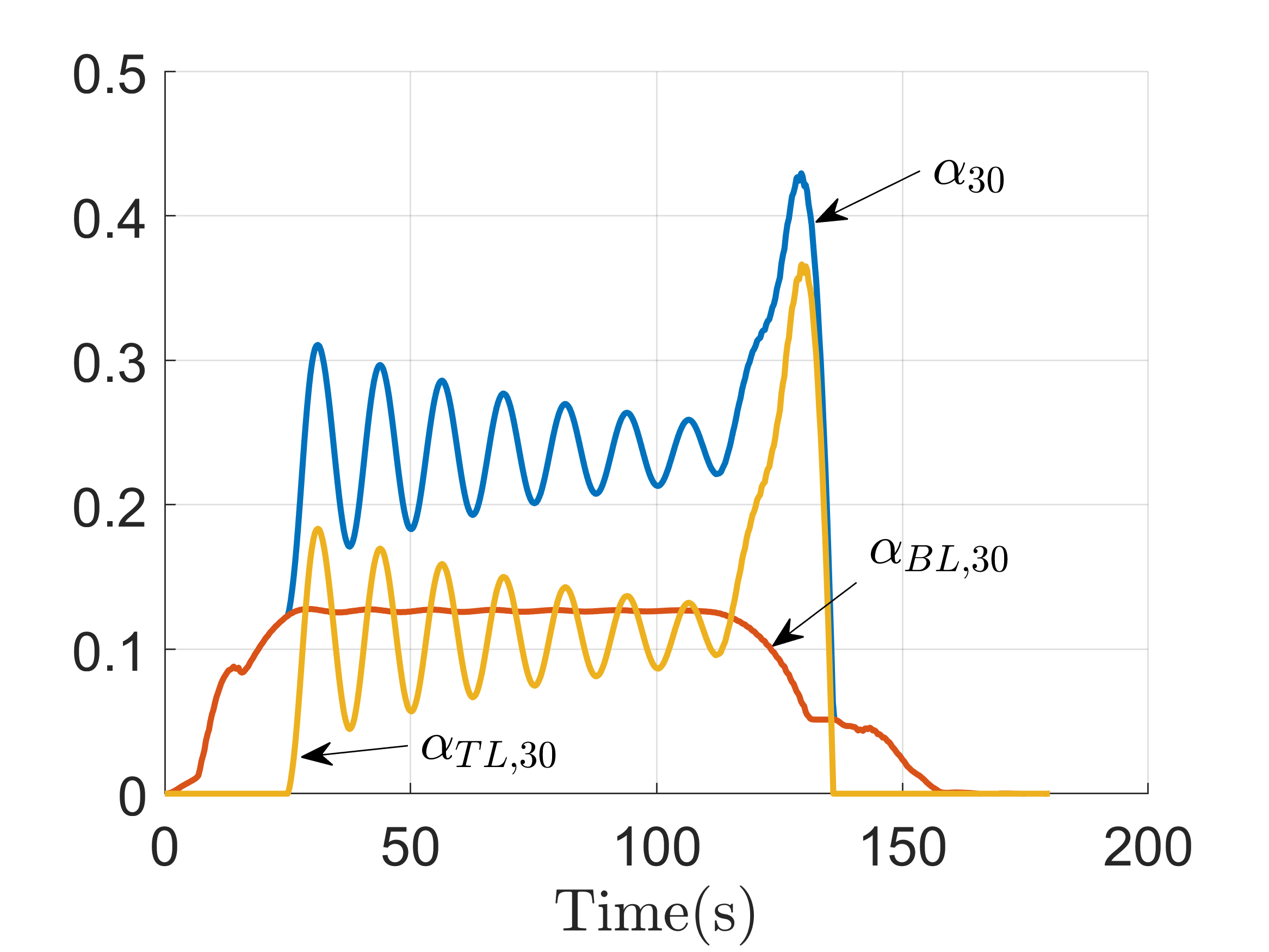}}
  \subfigure[\label{fig:control-response-30-large-penalty}]{\includegraphics[width=.24\linewidth]{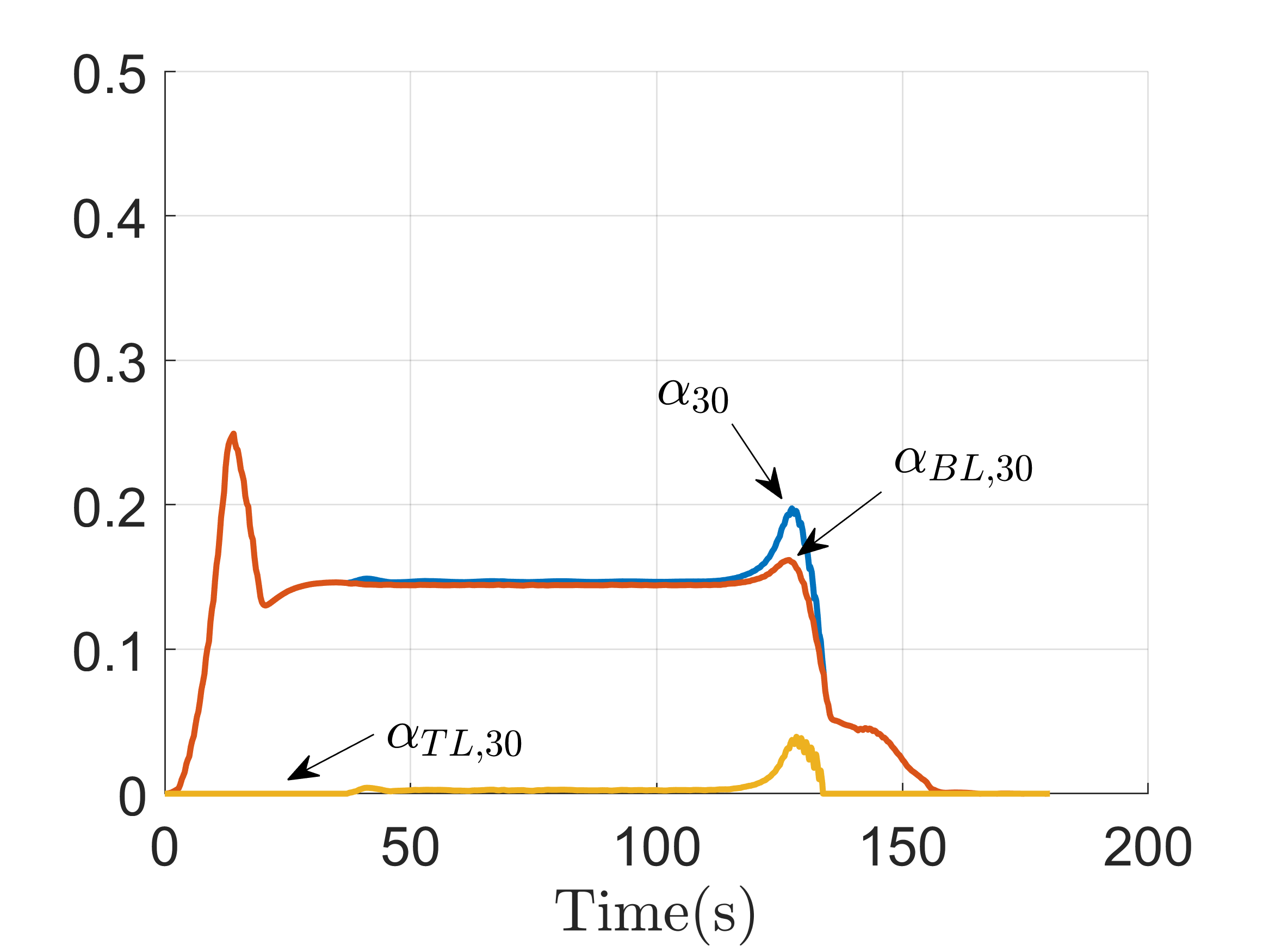}}
%   \subfigure[\label{fig:acclerated-gradient}]{\includegraphics[width=.24\linewidth]{epsfiles/trajectory-acclerated-gradient.png}}
     \caption{Decomposition of the control signal at
       node~30. Plot~\subref{fig:control-response-30},~\subref{fig:control-response-30-smaller-penalty},
       and~\subref{fig:control-response-30-large-penalty} show the
       signals generated by the two control layers at node $30$ using
       $d_{30} = 10^2$, $d_{30}=10$, and $d_{30}=10^3$, respectively,
       as values for the frequency safety violation penalty
       coefficient in the MPC component. With a larger penalty, the
       bottom layer plays a more significant role in the overall
       control signal.}\label{fig:trajectories-30}
\end{figure*}

\begin{figure*}[tbh!]
  \centering
  \subfigure[\label{fig:control-response-30-BL-normal}]{\includegraphics[width=.24\linewidth]{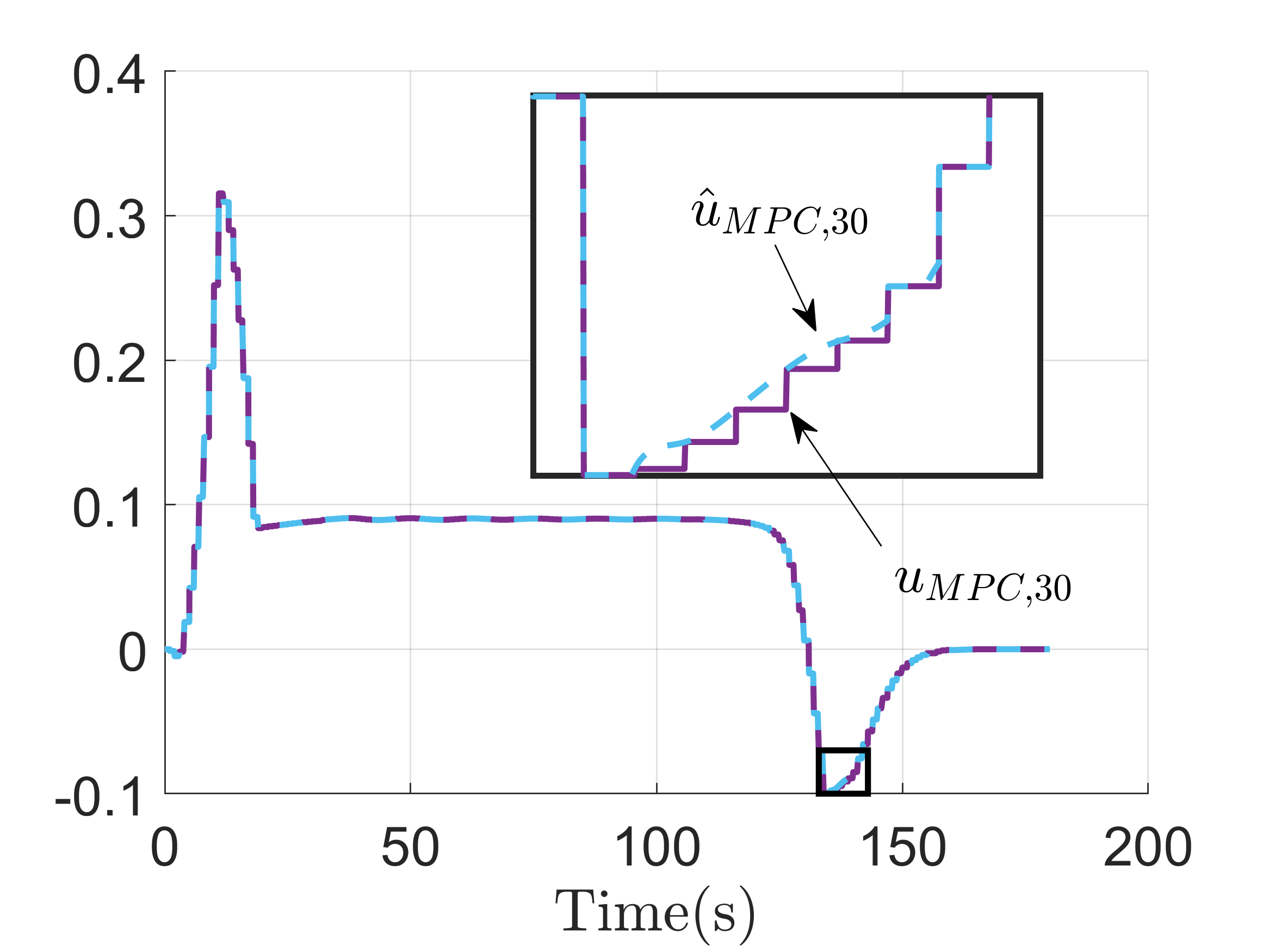}}
       \subfigure[\label{fig:control-response-30-BL-shifted}]{\includegraphics[width=.24\linewidth]{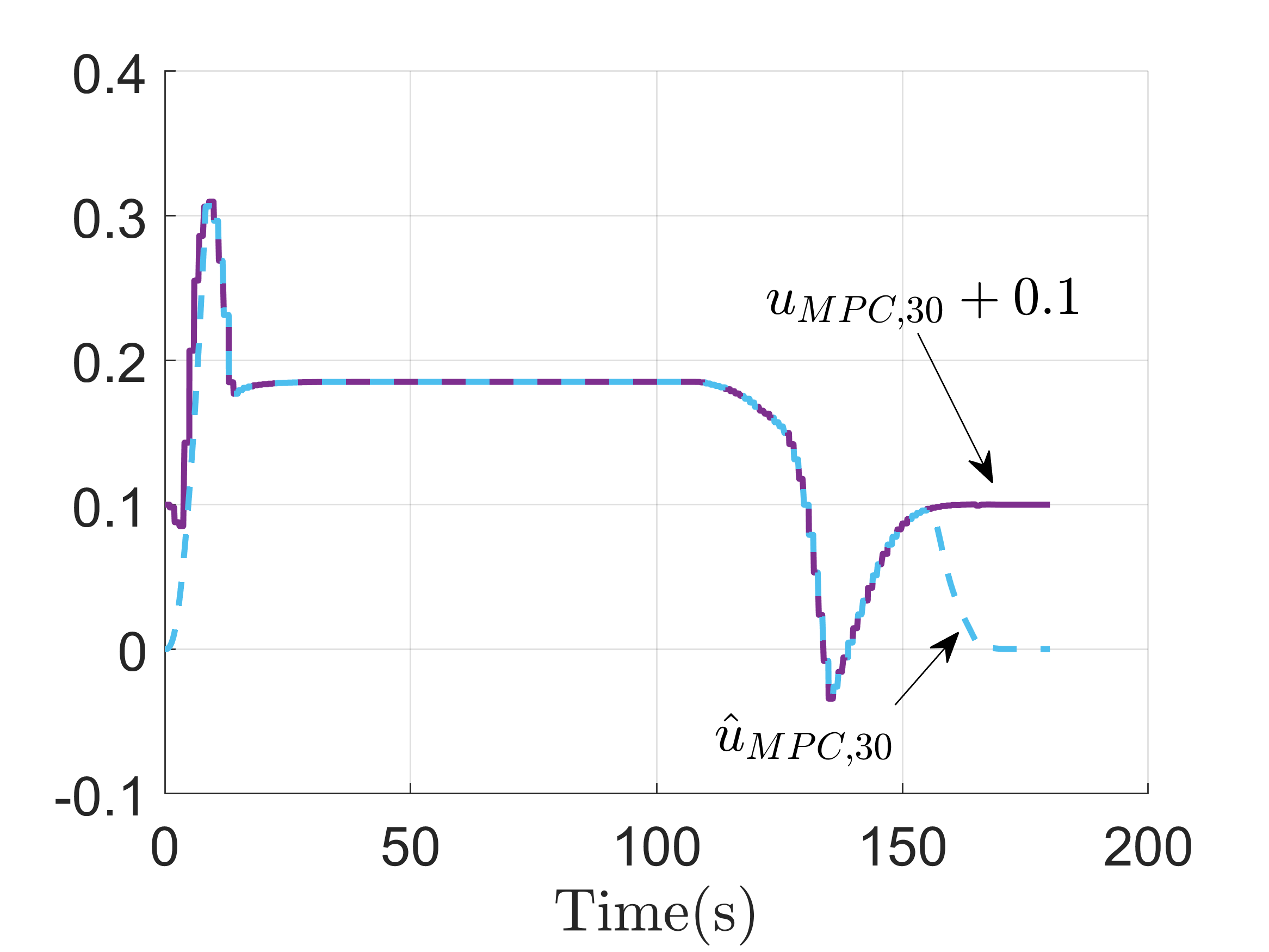}}
            \subfigure[\label{fig:saddle-points}]{\includegraphics[width=.24\linewidth]{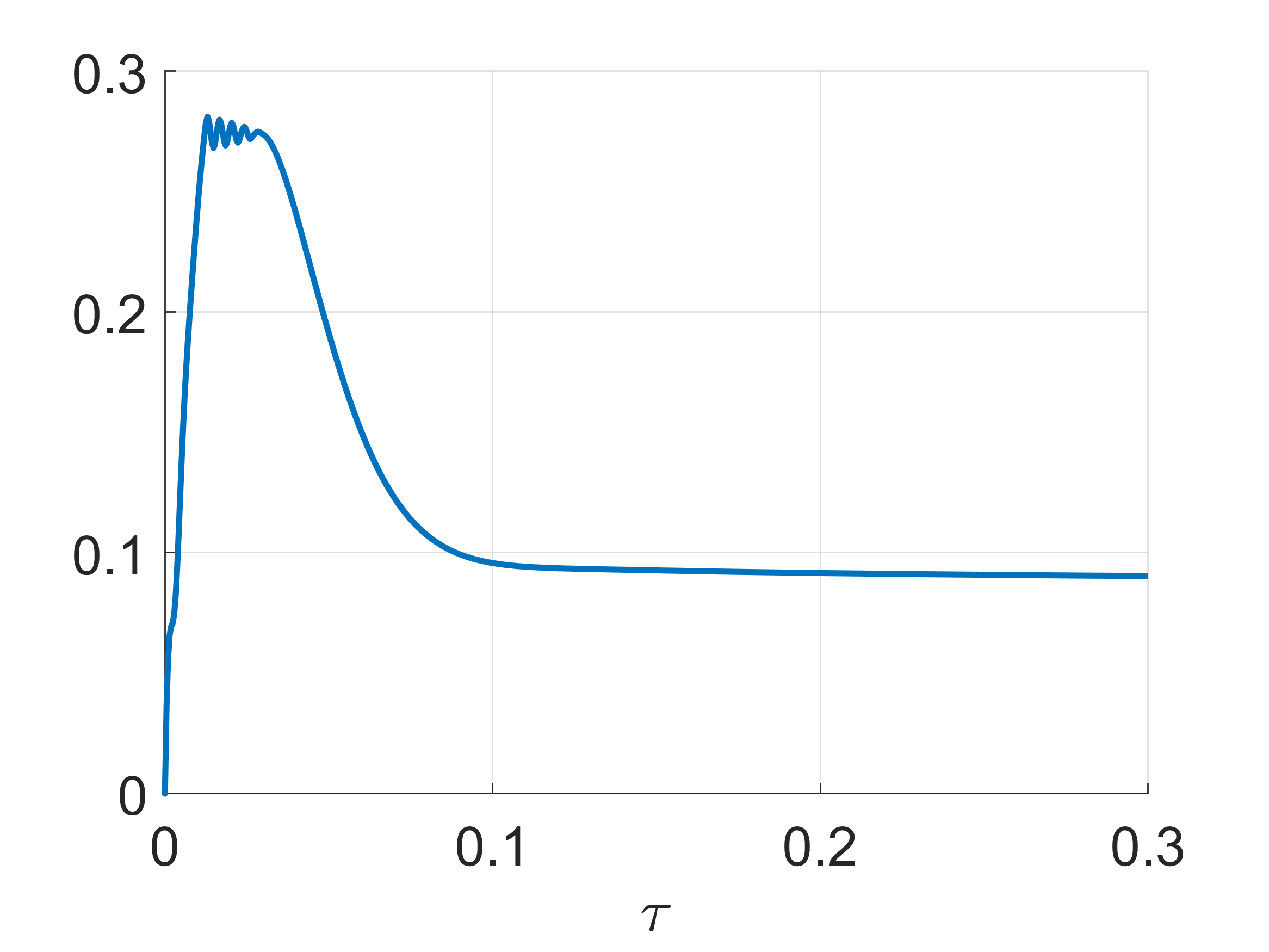}}
            \caption{\change{Decomposition of the bottom-layer control
                signal at
                node~30. Plot~\subref{fig:control-response-30-BL-normal}
                shows the MPC component output signal $u_{MPC,30}$ and
                the stability filter output signal $\hat
                u_{MPC,30}$. Both are almost identical except for a
                minor discrepancy appearing around 140s. To make their
                difference more prominent, in
                plot~\subref{fig:control-response-30-BL-shifted}, we
                add a constant 0.1 shift to $u_{MPC,30}$, which does
                not affect the convergence of $\hat u_{MPC,30}$ to 0
                (highlighting again the fact that the MPC component
                cannot jeopardize system closed- loop asymptotic
                stability, cf. Remark IV.4).
                Plot~\subref{fig:saddle-points} shows the convergence
                of the saddle-point
                dynamics~\eqref{sube:eqn:saddle-points} computing
                $u_{MPC,30}(50)$ in $0.1$s.  }
            }\label{fig:trajectories-30-detailed}
\end{figure*}

Lastly, we show that the distributed controller is able to steer the
frequency to the safe region from unsafe initial conditions.  To do
this, we consider the set-up of Figure~\ref{fig:trajectories} but
intentionally disable the controller for the first 30 seconds.  For
clarity, we only show the frequency and control trajectories at node
30 in Figure~\ref{fig:frequency-response-bad-initial}. Note that the
frequency quickly moves above the safe lower bound after the
controller becomes active at
$t=30$s. Figure~\ref{fig:control-response-bad-initial} shows the
control signal, where after some brief transient, $\alpha_{BL,30}$
still dominates the overall control signal.

\begin{figure*}[tbh!]
  \centering
  \subfigure[\label{fig:frequency-inaccurate-power}]{\includegraphics[width=.24\linewidth]{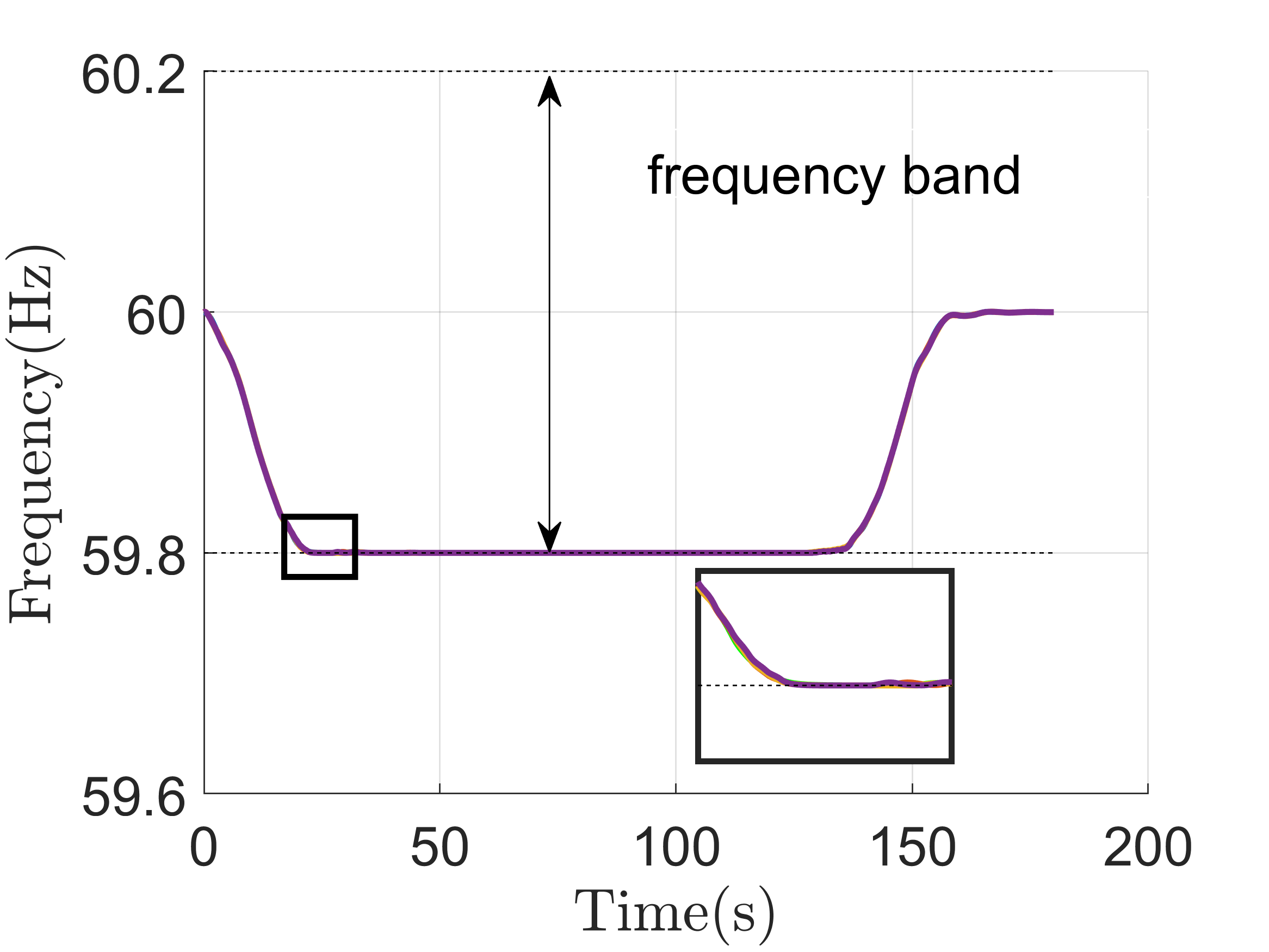}}
  \subfigure[\label{fig:frequency-generator-dynamics}]{\includegraphics[width=.24\linewidth]{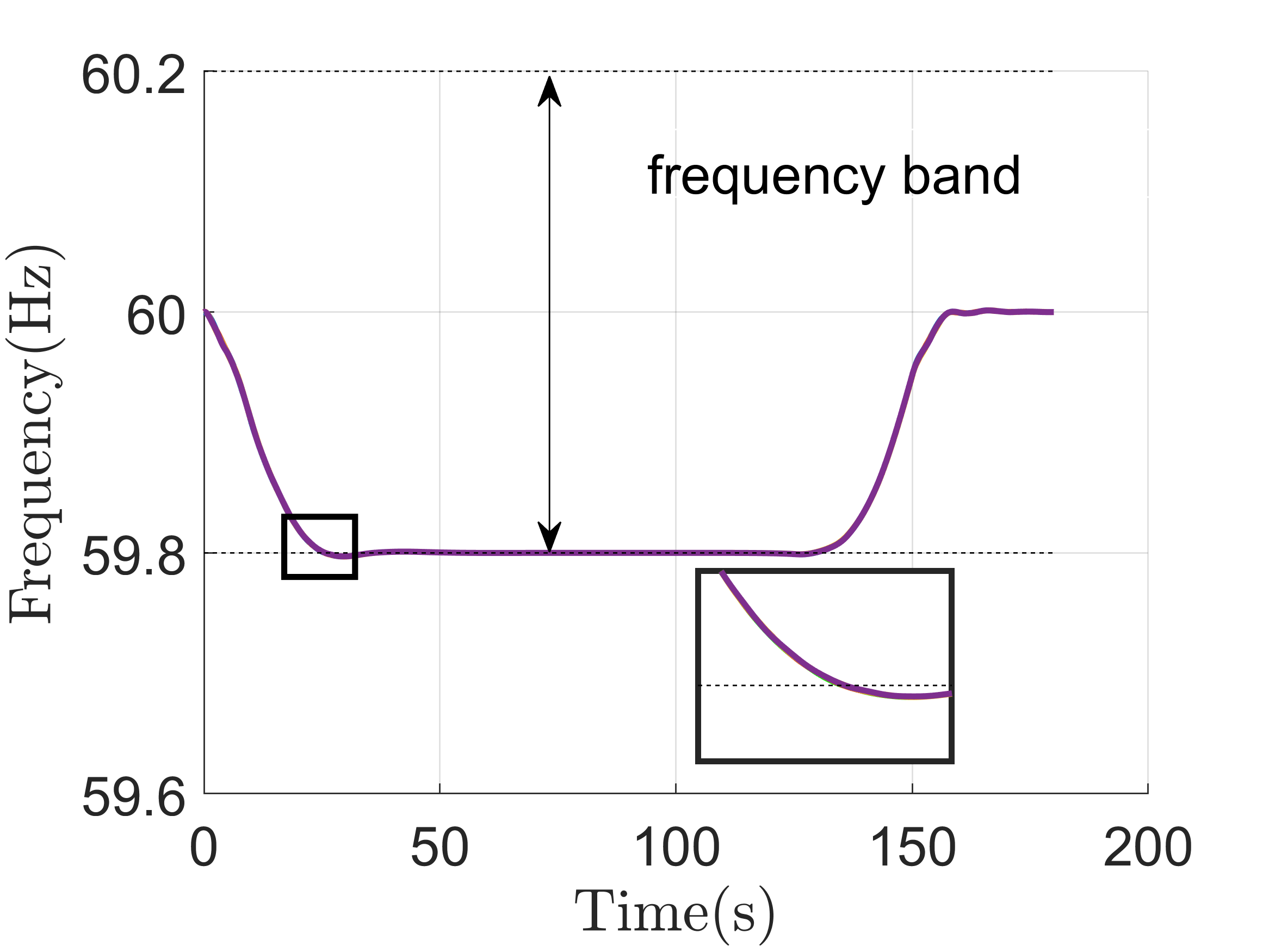}}
  \subfigure[\label{fig:frequency-generator-dynamics-inaccurate-power}]{\includegraphics[width=.24\linewidth]{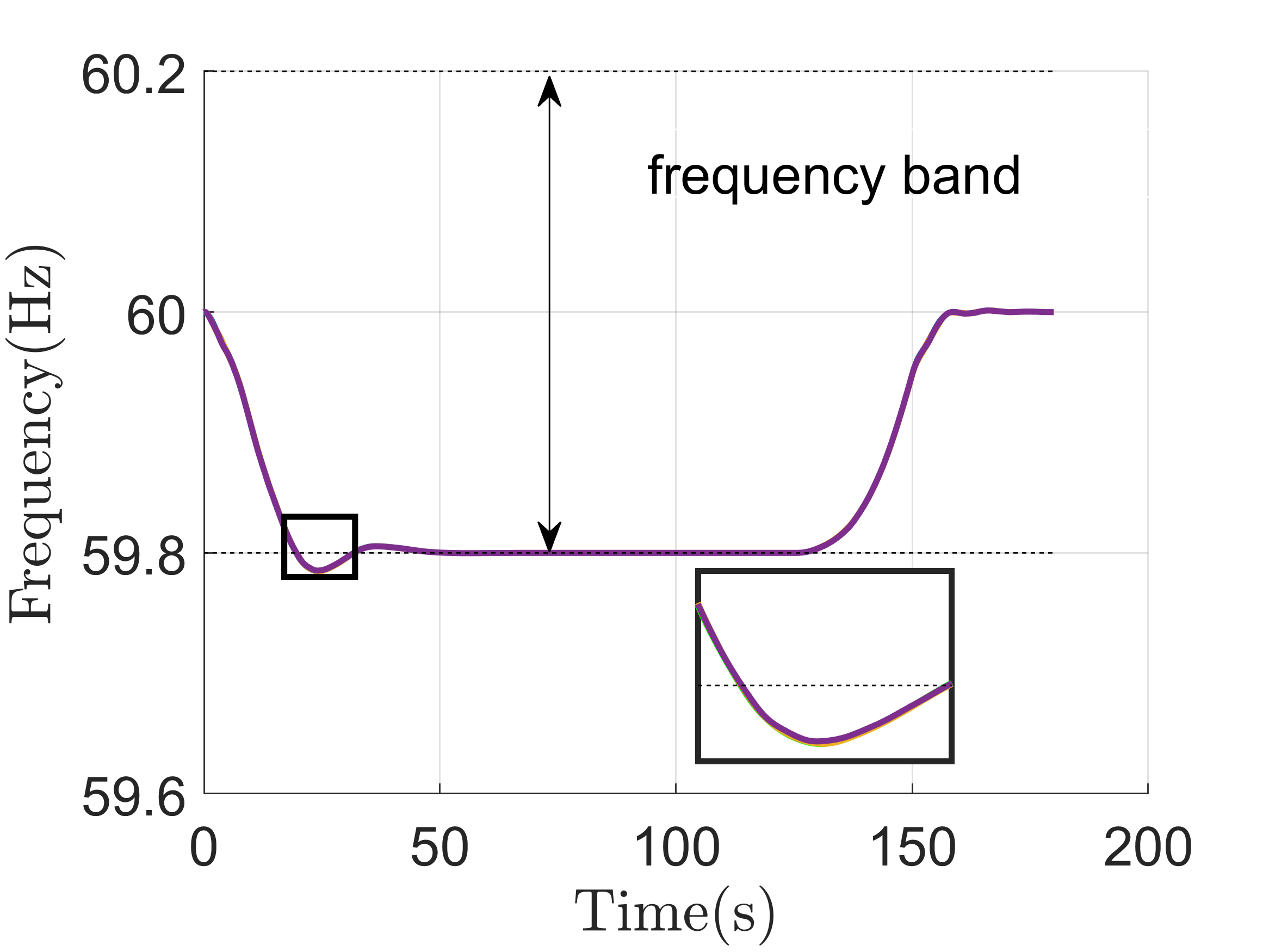}}
  \caption{\change{Frequency responses under inaccurate information
      and unmodeled
      dynamics. Plot~\subref{fig:frequency-inaccurate-power},~\subref{fig:frequency-generator-dynamics},
      and~\subref{fig:frequency-generator-dynamics-inaccurate-power}
      show the frequency responses at nodes 30, 31, 32, and 37 under
      inaccurate forecasted power injection, first-order generator
      dynamics, and both, respectively.}
  }\label{fig:trajectories-uncertainty}
\end{figure*}

\begin{figure}[tbh]
  \centering
  \subfigure[\label{fig:frequency-response-bad-initial}]{\includegraphics[width=.48\linewidth]{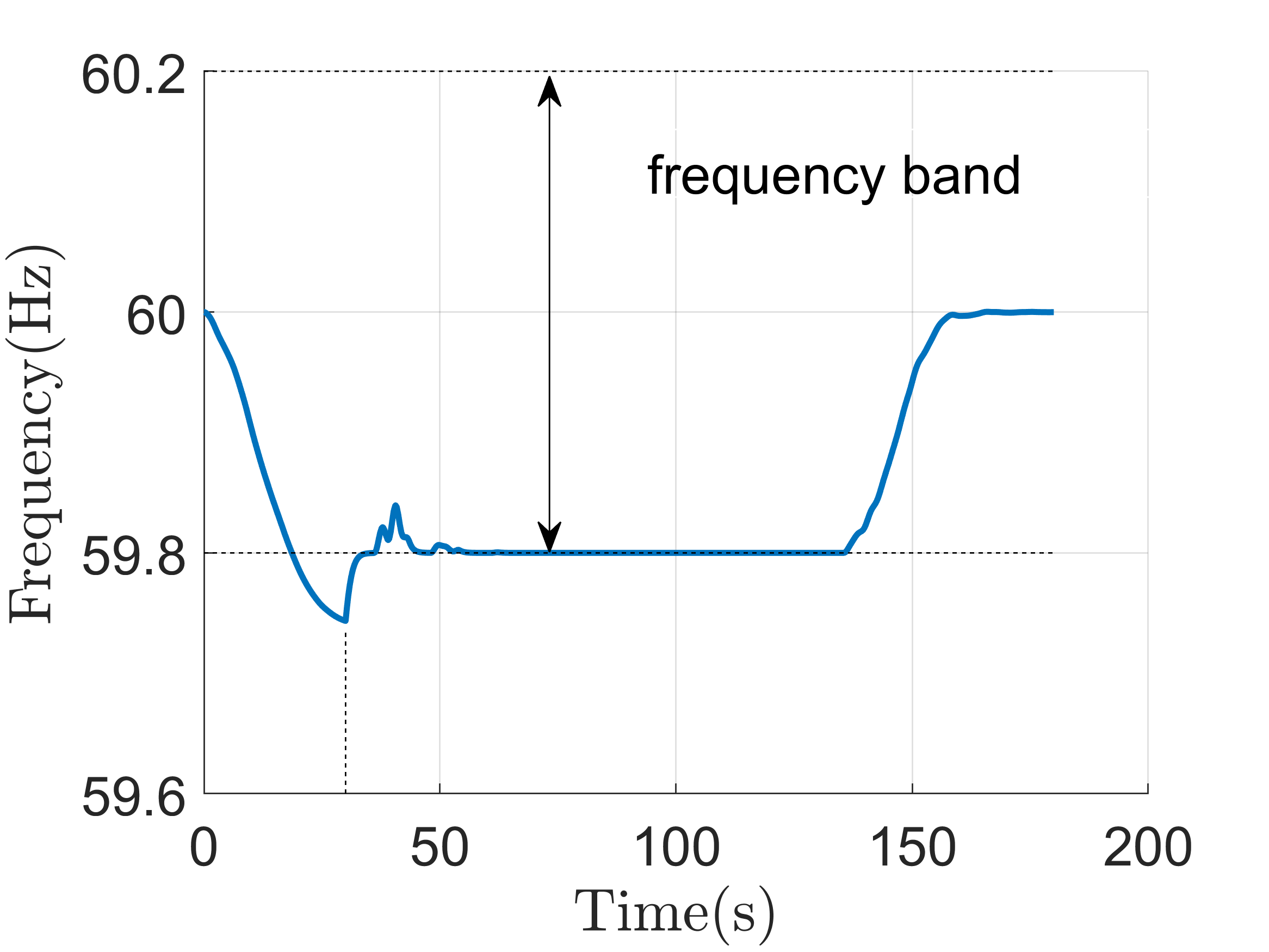}}
  \subfigure[\label{fig:control-response-bad-initial}]{\includegraphics[width=.48\linewidth]{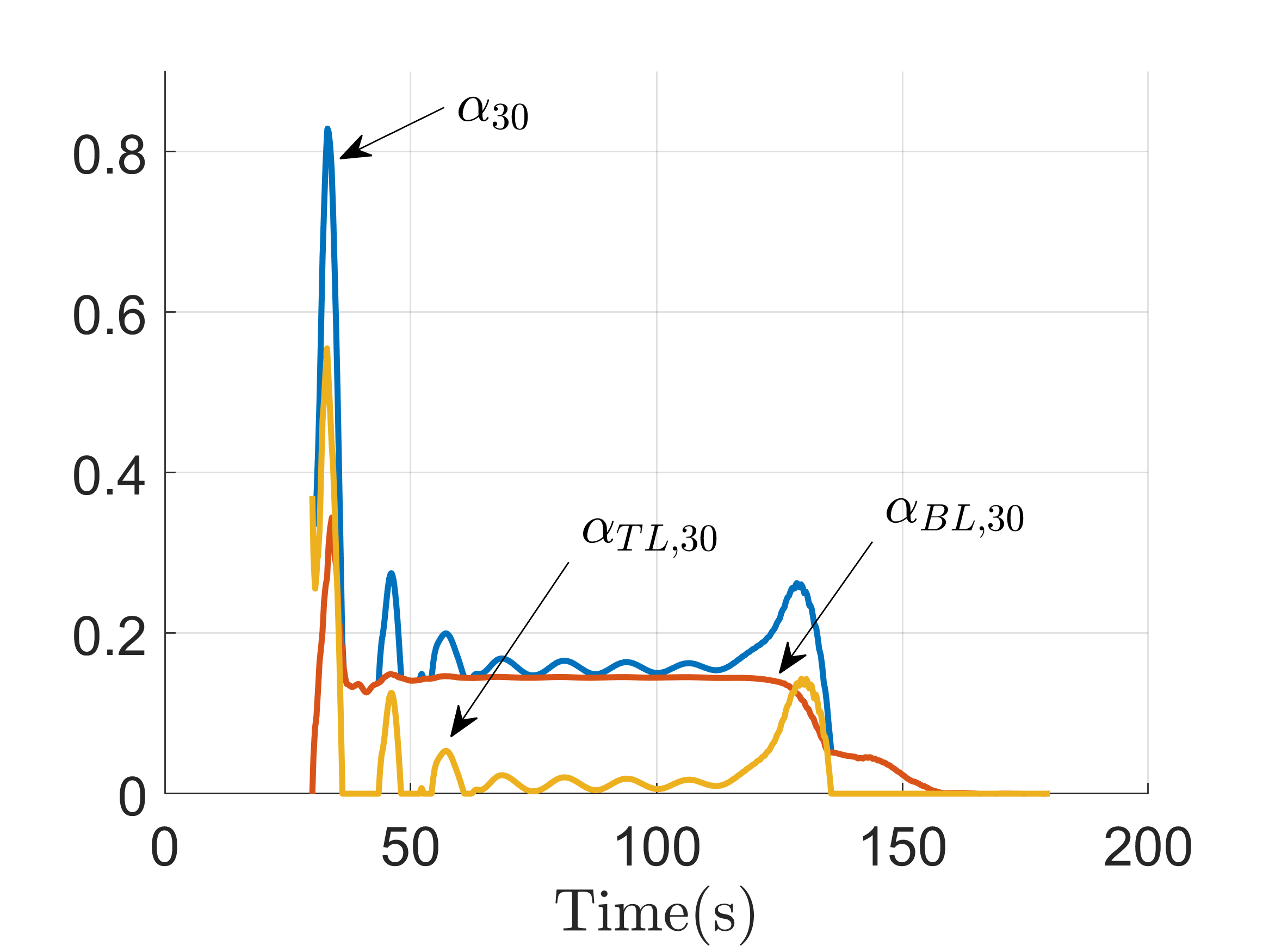}}
  \caption{Frequency and control trajectories at node 30 when the
    controller is turned on after $30$s.  In
    plot~\subref{fig:frequency-response-bad-initial}, the frequency
    gradually returns to the safe region once the controller kicks
    in. Plot~\subref{fig:control-response-bad-initial} shows the
    control signals.}\label{fig:trajectories-bad-initial}
\end{figure}

\section{Conclusions}
We have considered power networks governed by swing nonlinear dynamics
and introduced a bilayered control strategy to regulate transient
frequency in the presence of disturbances while maintaining network
stability. Adopting a receding horizon approach, the bottom-layer
controller periodically updates its output, enabling global
cooperation among buses to reduce the overall control effort while
respecting stability and soft frequency constraints.  The top-layer
controller, as a continuous state feedback controller,  tunes
the output of the bottom-layer control signal as required to
rigorously enforces frequency safety and attractivity. We have shown
that the entire control structure can be implemented in a distributed
fashion, where the control signal can be computed by having nodes
interact with up to 2-hop neighbors in the power network. Future work
will explore the optimization of the sampling sequences employed in
the bottom layer to improve performance, the quantitative evaluation
of the contributions of the top- and bottom-layer control signals, and
the analysis of the robustness of the proposed controller against
delays and saturation.

\bibliographystyle{IEEEtran}%
\bibliography{alias,JC,Main,Main-add}

% Generated by IEEEtran.bst, version: 1.14 (2015/08/26)
\begin{thebibliography}{10}
\providecommand{\url}[1]{#1}
\csname url@samestyle\endcsname
\providecommand{\newblock}{\relax}
\providecommand{\bibinfo}[2]{#2}
\providecommand{\BIBentrySTDinterwordspacing}{\spaceskip=0pt\relax}
\providecommand{\BIBentryALTinterwordstretchfactor}{4}
\providecommand{\BIBentryALTinterwordspacing}{\spaceskip=\fontdimen2\font plus
\BIBentryALTinterwordstretchfactor\fontdimen3\font minus
  \fontdimen4\font\relax}
\providecommand{\BIBforeignlanguage}[2]{{%
\expandafter\ifx\csname l@#1\endcsname\relax
\typeout{** WARNING: IEEEtran.bst: No hyphenation pattern has been}%
\typeout{** loaded for the language `#1'. Using the pattern for}%
\typeout{** the default language instead.}%
\else
\language=\csname l@#1\endcsname
\fi
#2}}
\providecommand{\BIBdecl}{\relax}
\BIBdecl

\bibitem{YZ-JC:19-acc}
Y.~Zhang and J.~Cort\'es, ``Double-layered distributed transient frequency
  control with regional coordination,'' in \emph{{A}merican {C}ontrol
  {C}onference}, Philadelphia, PA, Jul. 2019, pp. 658--663.

\bibitem{PK-JP:04}
P.~Kundur, J.~Paserba, V.~Ajjarapu, G.~Andersson, A.~Bose, C.~Canizares,
  N.~Hatziargyriou, D.~Hill, A.~Stankovic, C.~Taylor, T.~V. Cutsem, and
  V.~Vittal, ``Definition and classification of power system stability,''
  \emph{IEEE Transactions on Power Systems}, vol.~19, no.~2, pp. 1387--1401,
  2004.

\bibitem{NERC:11}
NERC, ``Balancing and frequency control,'' North American Electric Reliability
  Council, Tech. Rep., 2011.

\bibitem{FM-FD-GH-DJH-GV:2018}
F.~Milano, F.~D\"orfler, G.~Hug, D.~J. Hill, and G.~Verbi\v{c}, ``Foundations
  and challenges of low-inertia systems,'' in \emph{Power Systems Computation
  Conference}, Dublin, Ireland, June 2018, electronic proceedings.

\bibitem{FD-MC-FB:13}
F.~D{\"o}rfler, M.~Chertkov, and F.~Bullo, ``Synchronization in complex
  oscillator networks and smart grids,'' \emph{Proceedings of the National
  Academy of Sciences}, vol. 110, no.~6, pp. 2005--2010, 2013.

\bibitem{HDC:11}
H.~D. Chiang, \emph{Direct Methods for Stability Analysis of Electric Power
  Systems: Theoretical Foundation, BCU Methodologies, and Applications}.\hskip
  1em plus 0.5em minus 0.4em\relax John Wiley and Sons, 2011.

\bibitem{TLV-HDN-AM-JS-KT:18}
T.~L. Vu, H.~D. Nguyen, A.~Megretski, J.~Slotine, and K.~Turitsyn, ``Inverse
  stability problem and applications to renewables integration,'' \emph{IEEE
  Control Systems Letters}, vol.~2, no.~1, pp. 133--138, 2018.

\bibitem{JF-HL-YT-FB:18}
J.~Fang, H.~Li, Y.~Tang, and F.~Blaabjerg, ``Distributed power system virtual
  inertia implemented by grid-connected power converters,'' \emph{IEEE
  Transactions on Power Electronics}, vol.~33, no.~10, pp. 8488--8499, 2018.

\bibitem{SSG-CZ-ED-YCC-SVD:18}
S.~S. Guggilam, C.~Zhao, E.~Dall'Anese, Y.~C. Chen, and S.~V. Dhople,
  ``Optimizing {DER} participation in inertial and primary-frequency
  response,'' \emph{IEEE Transactions on Power Systems}, vol.~33, no.~5, pp.
  5194--5205, 2018.

\bibitem{FT-MA-DP-GS:15}
F.~Teng, M.~Aunedi, D.~Pudjianto, and G.~Strbac, ``Benefits of demand-side
  response in providing frequency response service in the future {GB} power
  system,'' \emph{Frontiers in Energy Research}, vol.~3, no.~36, 2015.

\bibitem{ADA-SC-ME-GM-KS-PT:19}
A.~D. Ames, S.~Coogan, M.~Egerstedt, G.~Notomista, K.~Sreenath, and P.~Tabuada,
  ``Control barrier functions: theory and applications,'' in \emph{{E}uropean
  {C}ontrol {C}onference}, Naples, Italy, Jun. 2019, pp. 3420--3431.

\bibitem{HKK:02}
H.~K. Khalil, \emph{Nonlinear Systems}, 3rd~ed.\hskip 1em plus 0.5em minus
  0.4em\relax Prentice Hall, 2002.

\bibitem{YZ-JC:19-auto}
Y.~Zhang and J.~Cort\'es, ``Distributed transient frequency control for power
  networks with stability and performance guarantees,'' \emph{Automatica}, vol.
  105, pp. 274--285, 2019.

\bibitem{YZ-JC:20-auto}
------, ``Model predictive control for transient frequency regulation of power
  networks,'' \emph{Automatica}, 2020, submitted.

\bibitem{HJ-JL-YS-DJH:15}
H.~Jiang, J.~Lin, Y.~Song, and D.~J. Hill, ``{MPC}-based frequency control with
  demand-side participation: A case study in an isolated wind-aluminum power
  system,'' \emph{IEEE Transactions on Power Systems}, vol.~30, no.~6, pp.
  3327--3337, 2015.

\bibitem{ANV-IAH-JBR-SJW:08}
A.~N. Venkat, I.~A. Hiskens, J.~B. Rawlings, and S.~J. Wright, ``Distributed
  {MPC} strategies with application to power system automatic generation
  control,'' \emph{IEEE Transactions on Control Systems Technology}, vol.~16,
  no.~6, pp. 1192--1206, 2008.

\bibitem{AF-MI-TD-MM:14}
A.~Fuchs, M.~Imhof, T.~Demiray, and M.~Morari, ``Stabilization of large power
  systems using {VSC}-{HVDC} and model predictive control,'' \emph{IEEE
  Transactions on Power Delivery}, vol.~29, no.~1, pp. 480 -- 488, 2014.

\bibitem{FB-JC-SM:08cor}
F.~Bullo, J.~Cort{\'e}s, and S.~Martinez, \emph{Distributed Control of Robotic
  Networks}, ser. Applied Mathematics Series.\hskip 1em plus 0.5em minus
  0.4em\relax Princeton University Press, 2009.

\bibitem{ARB-DJH:81}
A.~R. Bergen and D.~J. Hill, ``A structure preserving model for power system
  stability analysis,'' \emph{IEEE Transactions on Power Apparatus and
  Systems}, vol. 100, no.~1, pp. 25--35, 1981.

\bibitem{AP:12}
A.~Pai, \emph{Energy Function Analysis for Power System Stability}.\hskip 1em
  plus 0.5em minus 0.4em\relax New York: Springer, 1989.

\bibitem{FB:03}
F.~Borrelli, \emph{Constrained Optimal Control of Linear and Hybrid
  Systems}.\hskip 1em plus 0.5em minus 0.4em\relax New York: Springer, 2003.

\bibitem{AA-BA:09}
A.~Alessio and B.~Alberto, ``A survey on explicit model predictive control,''
  in \emph{Nonlinear Model Predictive Control}.\hskip 1em plus 0.5em minus
  0.4em\relax Springer, 2009, pp. 345--369.

\bibitem{SB-LV:04}
S.~Boyd and L.~Vandenberghe, \emph{Convex Optimization}.\hskip 1em plus 0.5em
  minus 0.4em\relax Cambridge University Press, 2004.

\bibitem{AC-EM-SHL-JC:18-tac}
A.~Cherukuri, E.~Mallada, S.~H. Low, and J.~Cort\'{e}s, ``The role of convexity
  in saddle-point dynamics: Lyapunov function and robustness,'' \emph{IEEE
  Transactions on Automatic Control}, vol.~63, no.~8, pp. 2449--2464, 2018.

\bibitem{EM-CZ-SL:17}
E.~Mallada, C.~Zhao, and S.~H. Low, ``Optimal load-side control for frequency
  regulation in smart grids,'' \emph{IEEE Transactions on Automatic Control},
  vol.~62, no.~12, pp. 6294--6309, 2017.

\bibitem{MHN-etal:14}
M.~H. Nazari, Z.~Costello, M.~J. Feizollahi, S.~Grijalva, and M.~Egerstedt,
  ``Distributed frequency control of prosumer-based electric energy systems,''
  \emph{IEEE Transactions on Power Systems}, vol.~29, pp. 2934--2942, 2014.

\bibitem{PT-AR:13}
P.~Trodden and A.~Richards, ``Cooperative distributed {MPC} of linear systems
  with coupled constraints,'' \emph{Automatica}, vol.~49, no.~2, pp. 479--487,
  2013.

\bibitem{PG-MDD-TK-BDS-AR:13}
P.~Giselsson, M.~D. Doanb, T.~Keviczky, B.~D. Schutter, and A.~Rantzer,
  ``Accelerated gradient methods and dual decomposition in distributed model
  predictive control,'' \emph{Automatica}, vol.~49, no.~3, pp. 829--833, 2013.

\bibitem{XW-SM-BDOA:19}
X.~Wang, S.~Mou, and B.~D.~O. Anderson, ``Scalable, distributed algorithms for
  solving linear equations via double-layered networks,'' \emph{IEEE
  Transactions on Automatic Control}, 2020, to appear.

\bibitem{MZ-SM:12}
M.~Zhu and S.~Mart{\'\i}nez, ``On distributed convex optimization under
  inequality and equality constraints,'' \emph{IEEE Transactions on Automatic
  Control}, vol.~57, no.~1, pp. 151--164, 2012.

\bibitem{KWC-JC-GR:09}
K.~W. Cheung, J.~Chow, and G.~Rogers, \emph{Power System Toolbox, v 3.0.}\hskip
  1em plus 0.5em minus 0.4em\relax Rensselaer Polytechnic Institute and Cherry
  Tree Scientific Software, 2009.

\bibitem{ZW-FL-SHL-CZ-SM:18}
Z.~Wang, F.~Liu, S.~H. Low, C.~Zhao, and S.~Mei, ``Distributed frequency
  control with operational constraints, part {I}: Per-node power balance,''
  \emph{IEEE Transactions on Smart Grid}, vol.~9, no.~4, pp. 1798--1811, 2018.

\end{thebibliography}

\vspace*{-7ex}

\begin{IEEEbiography}[{\includegraphics[width=1in,height=1.25in,clip,keepaspectratio]{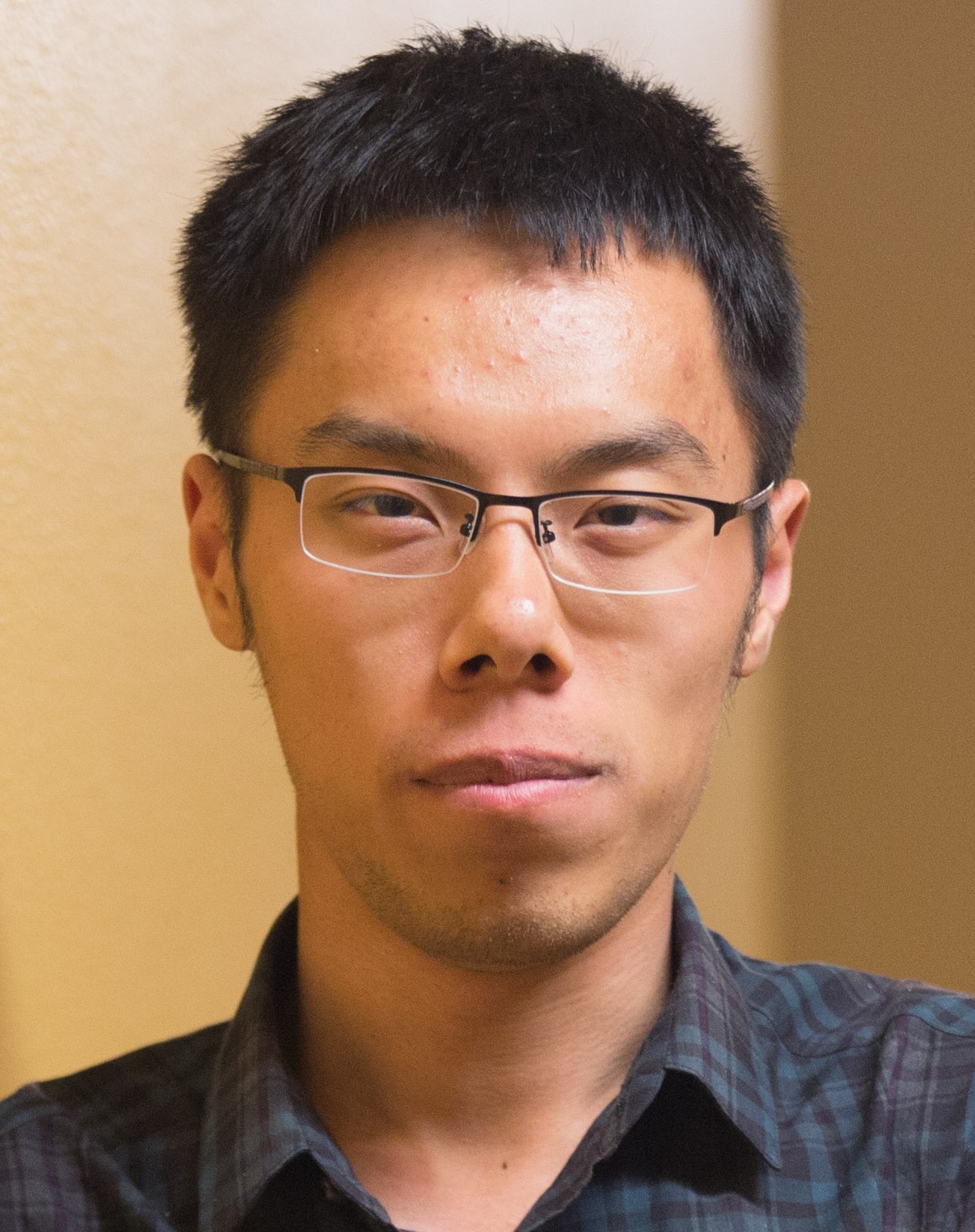}}]{Yifu
    Zhang}
  received the B.S. degree in automatic control from the Harbin
  Institute of Technology, Heilongjiang, China, in 2014, and the
  Ph.D. degree in mechanical engineering from the
  University of California, San Diego, CA, USA, in 2019. In winter
  2019, he interned at Mitsubishi Electric Research Laboratories, MA,
  USA. Currently he is a senior software quality engineer at The
  MathWorks, Inc., MA, USA. His research interests include distributed
  control and computation, model predictive control, adaptive
  control, data type optimization, function approximation, and  neural network compression.
\end{IEEEbiography}

\vspace*{-2ex}

\begin{IEEEbiography}[{\includegraphics[width=1in,height=1.25in,clip,keepaspectratio]{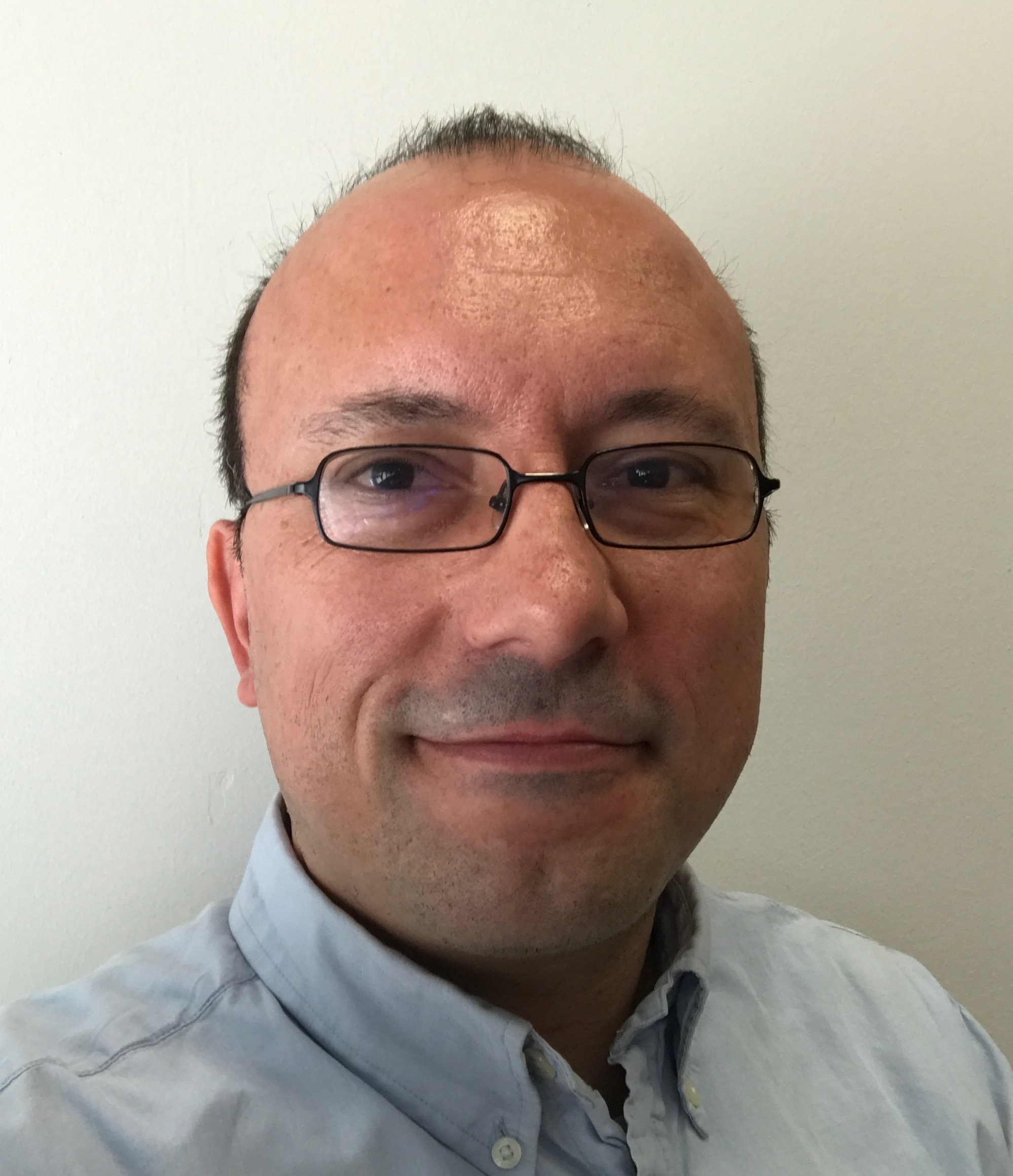}}]{Jorge
    Cort\'{e}s} 
  (M'02, SM'06, F'14) received the Licenciatura degree in mathematics
  from Universidad de Zaragoza, Zaragoza, Spain, in 1997, and the
  Ph.D. degree in engineering mathematics from Universidad Carlos III
  de Madrid, Madrid, Spain, in 2001. He held postdoctoral positions
  with the University of Twente, Twente, The Netherlands, and the
  University of Illinois at Urbana-Champaign, Urbana, IL, USA. He was
  an Assistant Professor with the Department of Applied Mathematics
  and Statistics, University of California, Santa Cruz, CA, USA, from
  2004 to 2007. He is now a Professor in the Department of Mechanical
  and Aerospace Engineering, University of California, San Diego, CA,
  USA. He is the author of Geometric, Control and Numerical Aspects of
  Nonholonomic Systems (Springer-Verlag, 2002) and co-author (together
  with F. Bullo and S.  Mart{\'\i}nez) of Distributed Control of
  Robotic Networks (Princeton University Press, 2009).  At the IEEE
  Control Systems Society, he has been a Distinguished Lecturer
  (2010-2014) and is currently its Director of Operations and an
  elected member (2018-2020) of its Board of Governors.  His research
  interests include distributed control and optimization, network
  science, resource-aware control, nonsmooth analysis, reasoning and
  decision making under uncertainty, network neuroscience, and
  multi-agent coordination in robotic, power, and transportation
  networks.
\end{IEEEbiography}

\end{document}